\RequirePackage{fix-cm}
\documentclass[smallextended]{svjour3}       
\smartqed  

\usepackage{graphicx}      

\usepackage{lmodern,enumerate,rotating,tensor,etoolbox,comment,csquotes,lipsum}
\usepackage{amsmath,amsfonts,amssymb,subfigure,mathtools,mathabx,blkarray,centernot}
\usepackage{tikz-cd}
\usetikzlibrary{cd}
\usetikzlibrary{shapes,shapes.geometric,arrows,fit,calc,positioning,automata}
\usepackage[normalem]{ulem}

\usepackage{stmaryrd}
\SetSymbolFont{stmry}{bold}{U}{stmry}{m}{n}

\usepackage{longtable}

\usepackage{tabularx}

\newcolumntype{L}[1]{>{\raggedright\let\newline\\\arraybackslash\hspace{0pt}}m{#1}}
\newcolumntype{C}[1]{>{\centering\let\newline\\\arraybackslash\hspace{0pt}}m{#1}}
\newcolumntype{R}[1]{>{\raggedleft\let\newline\\\arraybackslash\hspace{0pt}}m{#1}}

\DeclarePairedDelimiter\ceil{\lceil}{\rceil}

\DeclareMathOperator{\CCa}{CC_A}

\DeclareMathOperator{\size}{size}
\DeclareMathOperator{\tim}{time}

\DeclareMathOperator{\WEIGHT}{WT}
\DeclareMathOperator{\trim}{trim}
\DeclareMathOperator{\head}{head}
\DeclareMathOperator{\tail}{tail}
\DeclareMathOperator{\TW}{TW}
\DeclareMathOperator{\RA}{RA}

\usepackage{scalerel,stackengine}
\stackMath
\newcommand\reallywidehat[1]{%
\savestack{\tmpbox}{\stretchto{%
  \scaleto{%
    \scalerel*[\widthof{\ensuremath{#1}}]{\kern-.6pt\bigwedge\kern-.6pt}%
    {\rule[-\textheight/2]{1ex}{\textheight}}
  }{\textheight}%
}{0.5ex}}%
\stackon[1pt]{#1}{\tmpbox}%
}
\allowdisplaybreaks

\newcommand{\Z}{\mathbb{Z}}
\newcommand{\R}{\mathbb{R}}
\newcommand{\N}{\mathbb{N}}
\newcommand{\Q}{\mathbb{Q}}
\newcommand{\LM}{\mathcal{L}}

\newcommand{\underN}{\underline{\mathbb{N}}}
\newcommand{\underR}{\underline{\mathbb{R}}}
\newcommand{\underQ}{\underline{\mathbb{Q}}}
\newcommand{\Acal}{\mathcal{A}}
\newcommand{\dt}{\delta}
\newcommand{\Dt}{\Delta}

\newcommand{\ep}{\epsilon}

\newcommand{\Sig}{\Sigma}
\newcommand{\s}{\sigma}

\newcommand{\Mt}{\mathcal{M}}

\newcommand{\frakM}{\mathfrak{M}}

\newcommand{\frakB}{\mathfrak{B}}
\newcommand{\frakE}{\mathfrak{E}}
\newcommand{\frakT}{\mathfrak{T}}

\newcommand{\PSPACE}{\mathsf{PSPACE}}
\newcommand{\coNPSPACE}{\mathsf{coNPSPACE}}
\newcommand{\NPSPACE}{\mathsf{NPSPACE}}
\newcommand{\EXPTIME}{\mathsf{EXPTIME}}
\newcommand{\NEXPTIME}{\mathsf{NEXPTIME}}
\newcommand{\PTIME}{\mathsf{PTIME}}

\newcommand{\EXPSPACE}{\mathsf{EXPSPACE}}
\newcommand{\NL}{\mathsf{NL}}

\newcommand{\NP}{\mathsf{NP}}
\newcommand{\coNP}{\mathsf{coNP}}

\newcommand{\llb}{\llbracket}
\newcommand{\rrb}{\rrbracket}

\tikzset{
->, 
node distance=3cm, 
every state/.style={thick, fill=gray!10}, 
initial text=$ $, 
}

\tikzset{
  ,decision/.style=
    {
      diamond, draw, fill=blue!20, text width=4.5em, text badly centered, 
      node distance=3cm, inner sep=0pt
    }
  ,block/.style=
    {
      rectangle, draw, fill=gray!20, text width=7.5em, rounded corners, 
      minimum height=4.1em
    }
  ,blocks/.style=
    {
      rectangle, fill=white!20, text width=9em, rounded corners,
      minimum height=4em
    }
  ,line/.style={draw, -latex'}
  ,cloud/.style={ellipse,fill=white!20, node distance=2cm, minimum height=2em}
}

\usetikzlibrary{decorations.pathmorphing}
\newcommand\xrsquigarrow[1]{%
    \mathrel{%
        \begin{tikzpicture}[%
            baseline={(current bounding box.south)}
            ]
        \node[%
            ,inner sep=.44ex
            ,align=center
            ] (tmp) {$\scriptstyle #1$};
        \path[%
            ,draw,<-
            ,decorate,decoration={%
                ,zigzag
                ,amplitude=0.7pt
                ,segment length=1.2mm,pre length=3.5pt
                }
            ] 
        (tmp.south east) -- (tmp.south west);
        \end{tikzpicture}
        }
    }

\AtEndEnvironment{remark}{\qed}
\AtEndEnvironment{example}{\qed}
\usepackage{times,amsmath,amssymb,graphicx,tikz,algorithm,algorithmic,url,hhline,booktabs,soul}
\usetikzlibrary{automata,arrows,positioning}
\usetikzlibrary{petri}

\makeatletter
\newcommand{\subalign}[1]{%
  \vcenter{%
    \Let@ \restore@math@cr \default@tag
    \baselineskip\fontdimen10 \scriptfont\tw@
    \advance\baselineskip\fontdimen12 \scriptfont\tw@
    \lineskip\thr@@\fontdimen8 \scriptfont\thr@@
    \lineskiplimit\lineskip
    \ialign{\hfil$\m@th\scriptstyle##$&$\m@th\scriptstyle{}##$\crcr
      #1\crcr
    }%
  }
}
\makeatother

\usepackage[colorlinks]{hyperref}

\usepackage{colortbl}

\newcommand{\blue}{\color{blue}}
\definecolor{green}{rgb}{0.1,0.7,0.1}

\begin{document}

\tikzset{elliptic state/.style={draw,ellipse}}

\title{Detectability of labeled weighted automata over monoids
}


\author{Kuize Zhang 
}


\institute{K. Zhang \at
			  Control Systems Group, Technical University of Berlin, 10587 Berlin, Germany\\
              \email{kuize.zhang@campus.tu-berlin.de}           
}

\date{Received: date / Accepted: date}

\maketitle

\begin{abstract}

	In this paper, by developing appropriate methods, we for the first time 
	obtain characterization of four fundamental notions of detectability for general labeled weighted 
	automata over monoids (denoted by $\Acal^{\frakM}$ for short),
	where the four notions are
	strong (periodic) detectability (SD and SPD) and weak (periodic) detectability (WD and WPD).
	The contributions of the current paper are as follows.
	Firstly, we formulate the notions of concurrent composition, observer, and detector for $\Acal^{\frakM}$.
	Secondly, we use the concurrent composition to
	give an equivalent condition for SD, use the detector to give an equivalent condition for SPD,
	and use the observer to give equivalent conditions for WD and WPD, all for general 
	$\Acal^{\frakM}$ without any assumption. Thirdly, we prove that for a labeled weighted automaton
	over monoid $(\Q^k,+)$ (denoted by $\Acal^{\Q^k}$),
	its concurrent
	composition, observer, and detector can be computed in $\NP$, $2$-$\EXPTIME$, and $2$-$\EXPTIME$, respectively,
	by developing novel connections between $\Acal^{\Q^k}$ and the $\NP$-complete exact path length problem (proved
	by [Nyk\"{a}nen and Ukkonen, 2002]) and a subclass of Presburger arithmetic. 
	As a result, we prove that for $\Acal^{\Q^k}$, SD can be verified in $\coNP$, while SPD, WD, and WPD
	can be verified in $2$-$\EXPTIME$.
	Particularly, for $\Acal^{\Q^k}$ in which from every state, a distinct state can be
	reached through some unobservable, instantaneous path, detector $\Acal^{\Q^k}_{det}$ can be computed in 
	$\NP$, and SPD can be verified in $\coNP$.
	Finally, we prove that the problems of verifying SD and SPD of deterministic, deadlock-free,
	and divergence-free $\Acal^{\N}$ over monoid $(\N,+)$
	are both $\coNP$-hard.

	The original methods developed in this paper will provide foundations for characterizing other 
	fundamental properties (e.g., diagnosability and opacity) in labeled weighted automata over monoids.

	In addition, in order to differentiate labeled weighted automata over monoids from labeled timed automata,
	we also initially explore detectability in labeled timed automata, and prove
	that the SD verification problem is $\PSPACE$-complete, while WD and WPD are undecidable.

\end{abstract}

\section{Introduction}

\subsection{Background and motivation}
\label{subsec:background}

The state detection problem of partially-observed (aka labeled) dynamical systems has been a fundamental problem in 
both computer science \cite{Moore1956} and control science \cite{Kalman1963MathDescriptionofLDS} 
since the 1950s and the 1960s, respectively.
\emph{Detectability} is a basic property of labeled dynamical systems:
when it holds one can use an observed label/output sequence
generated by a system to reconstruct its \emph{current} state
\cite{Giua2002ObservabilityPetriNets,Shu2007Detectability_DES,Sandberg2005HomingSynchronizingSequence,Zhang2016WPGRepresentationReconBCN}.
This property plays a fundamental role in many related
control problems such as observer design and controller synthesis. Hence in different application
scenarios, it is meaningful to characterize different notions of detectability.
On the other hand, detectability is strongly related to another fundamental property of diagnosability
where the latter describes whether one can use an observed output sequence to determine whether 
some special events (called faulty events) have occurred 
\cite{Sampath1995DiagnosabilityDES,Hadjicostis2020DESbook}.
Recently, a decentralized setting of strong detectability and diagnosability (together with another 
property called predictability) were unified into one mathematical framework in labeled finite-state automata
\cite{Zhang2021UnifyingDetDiagPred}.
Moreover, detectability is also related to several cyber-security properties, e.g.,
the property of opacity that was originally proposed to describe information flow security in
computer science in the early 2000s \cite{Mazare2004Opacity} can be seen as the absence of detectability.

\emph{Discrete-event systems} (DESs) are usually composed of transitions between discrete states caused by
spontaneous occurrences of labeled events \cite{WonhamSupervisoryControl,Cassandras2009DESbook}.
For DESs modeled by \emph{labeled finite-state automata} and \emph{labeled Petri nets},
the detectability problem
has been widely studied,
see related results on labeled finite-state automata 
\cite{Shu2007Detectability_DES,Shu2011GDetectabilityDES,Zhang2017PSPACEHardnessWeakDetectabilityDES,Zhang2019KDelayStrDetDES,Masopust2018ComplexityDetectabilityDES},
and also see related results on labeled Petri nets \cite{Zhang2018WODESDetectabilityLPS,Masopust2019DetectabilityPetriNet,Zhang2020DetPNFA}, and on labeled bounded Petri nets \cite{Lan2020C_Det_Bounded_PetriNet}.
Detectability has also been studied for probabilistic finite-state automata 
\cite{Keroglou2017DetProbAutomata,Yin2017InitialStateDetectabilityStoDES}.

The above models, either logic systems (labeled finite-state automata and labeled Petri nets),
or probabilistic finite-state automata, are untimed. In such models, the time consumption for
a transition's execution is not specified. In spite of this, one can infer from the above literature
that all unobservable transitions' executions are assumed to consume no time by default, and the executions of
every pair of observable transitions with the same label are assumed to consume the same time.
In order to make these models more realistic, 
measures to time consumptions for transitions' executions have been added, so that timed models have been studied,
e.g., 
\emph{labeled timed automata}\footnote{In the current paper, we call the timed automata studied in 
\cite{Tripakis2002DiagnosisTimedAutomata,Cassez2012ComplexityCodiagnosability,Li2021ObserverSpecialTimedAutomata}
labeled timed automata, because the events therein are endowed with labels/outputs,
while in the standard timed automata proposed in \cite{Alur1994TimedAutomaton}, events are unlabeled.}
\cite{Tripakis2002DiagnosisTimedAutomata,Cassez2012ComplexityCodiagnosability},
special classes \emph{of labeled weighted 
automata} over \emph{semirings} 
\cite{Lai2021DetUnambiguousWAutomata,Lai2021ObserverPolyAmbiguousWPA}, etc. 

In this paper, we study \emph{labeled weighted 
automata over monoids}
\cite{Daviaud2017DegreeSequentialityWeightedAutomata}, denoted by $\Acal^{\frakM}$.
Such systems have various features. When monoid $\frakM$ is specified as $(\Q_{\ge0},+)$, where $\Q_{\ge0}$ denotes
the set of nonnegative rational numbers, $\Acal^{\frakM}$
becomes a one-clock labeled timed automaton in which the clock is reset along with every occurrence 
of every event and all clock constraints are singletons (details are shown in Section~\ref{sec:detLTA})
(such automata are exactly the automata studied  
in \cite{Li2021ObserverSpecialTimedAutomata}), and hence can represent timed DESs;
when $\frakM$ is specified as $(\Q^k,+)$, the weights can represent deviations of positions of a moving object
in some region; when $\frakM$ is specified as $A^*\times B^*$ over alphabets $A,B$,
where $A^*$ and $B^*$ are free monoids, $\Acal^{\frakM}$ becomes a (finite-state) transducer
\cite{Beal2002DeterminizationofTransducers}. We will characterize detectability for general $\Acal^{\frakM}$
and will also prove that the results obtained in $\Acal^{\Q^k}$ can be implemented algorithmically.
Consider the motivating example as follows.

\begin{example}\label{exam10_MPautomata} 
	Consider a finite region shown in Fig.~\ref{fig23_det_MPautomata}, in which $P_1,P_2,P_3,P_4$
	denote $4$ positions. Assume a robot $A$ walking between these positions to finish a prescribed task.
	Assume the energy levels of $A$ are quantized into $0,1,\dots,10$. When $A$ is in position $P_1$ and
	moves to $P_2$, it sends signal $a$ and its target position $P_2$ along with its energy level decreasing by
	$1$. The other movements can be
	described analogously, where signal $a$ corresponds to energy level decreasing by $1$, signal $u$
	corresponds to energy level decreasing by $0$ or $1$, $b$ corresponds to
	energy level increasing by $1$. $A$ sends $a$ and $b$, but never sends $u$. When $A$ sends a signal,
	it meanwhile sends the corresponding target position. Particularly when the energy level of $A$ is $10$
	it never increases,
	i.e., when $A$ is in energy level $10$ and moves from $P_i$ to $P_{i-1}$, $i=4,3,2$, the energy level of
	$A$ remains to be $10$. Particularly when the energy level of $A$ is $0$ it never decreases, i.e., in this case
	$A$ never moves from $P_i$ to $P_{i+1}$, $i=1,2,3$.
	All these information is known to a person $B$. Assume that whenever
	$A$ sends $a$ or $b$, $B$ receives/observes the symbols and the corresponding target positions.
	Then as time advances, $B$ cannot use these observations to determine what the energy level $A$ will be
	in unless $A$ never visits $P_3$. For example, assume initially $A$ is in energy level $5$ and in position $P_1$,
	and assume $B$ knows the initial energy level. When $A$ moves to $P_2$, $B$
	observes $a/P_2$, and knows that $A$ is in energy level $4$. And then $A$ moves to $P_3$, 
	$A$ could be in energy level $4$ or $3$ but $B$ does not know $A$ is in $P_3$. From now on,
	$B$ will never know what exact energy level $A$ will be in. 
	For example, $A$ moves back to $P_2$ and could be in energy level $5$ or $4$, $B$ observes $b/P_2$;
	and then $A$ moves to $P_3$ again, $A$ could be in energy level $5$, $4$, or $3$; and then $A$ moves back
	to $P_2$ again, $A$ could be in energy level $6$, $5$, or $4$, $B$ observes $b/P_2$ (see
	Table~\ref{tab3_det_MPautomata}).
	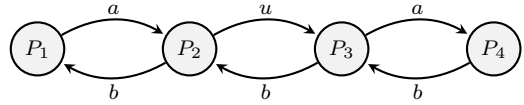
\begin{figure}[!htbp]
   \centering
	\begin{tikzpicture}
	[>=stealth',shorten >=1pt,thick,auto,node distance=2.0 cm, scale = 1.0, transform shape,
	->,>=stealth,inner sep=2pt]

	\tikzstyle{emptynode}=[inner sep=0,outer sep=0]

	\node[state] (p1) {$P_1$};
	\node[state] (p2) [right of = p1] {$P_2$};
	\node[state] (p3) [right of = p2] {$P_3$};
	\node[state] (p4) [right of = p3] {$P_4$};

	\path [->]
	(p1) edge [bend left] node {$a$} (p2)
	(p2) edge [bend left] node {$u$} (p3)
	(p3) edge [bend left] node {$a$} (p4)
	(p4) edge [bend left] node {$b$} (p3)
	(p3) edge [bend left] node {$b$} (p2)
	(p2) edge [bend left] node {$b$} (p1)
	;

    \end{tikzpicture}
	\caption{The finite region in the motivating Example~\ref{exam10_MPautomata}.}
	\label{fig23_det_MPautomata}  
	\end{figure}
	\begin{table}[!htbp]
		\centering
		\begin{tabular}{c|cccccccccccccccccc}
			\hline\rowcolor{lightgray}
			position deviation & $P_1$ & $\to P_2$ & $\to P_3$ & $\to P_2$ & $\to P_3$ & $\to P_2$ & $\to P_3$ & $\cdots$\\\hline
			energy level estimation & $5$ & $4$ & $4,3$ & $5,4$ & $5,4,3$ & $6,5,4$ & $6,5,4,3$ & $\cdots$\\
			observation & & $a/P_2$ & $\ep$ & $b/P_2$ & $\ep$ & $b/P_2$ & $\ep$ & $\cdots$\\
			\hline
		\end{tabular}
		\caption{Energy level estimation in the motivating Example~\ref{exam10_MPautomata}.}
		\label{tab3_det_MPautomata}
	\end{table}
	Later in Example~\ref{exam11_MPautomata}, we will show this model can be represented by a labeled 
	weighted automaton over some monoid, but cannot be described by a labeled timed automaton.
\end{example}

\subsection{Literature review}
\label{subsec:LiterRev}

Two fundamental definitions are \emph{strong detectability} and \emph{weak detectability}
\cite{Shu2007Detectability_DES}.
The former implies that there exists a positive integer $k$ such that for \emph{every}
infinite-length trajectory, each prefix of its label/output sequence of length no less than $k$
allows reconstructing the current state. The latter relaxes the former
by changing ``\emph{every}'' to ``\emph{some}''.
In order to adapt to different application scenarios,
variants of strong detectability and weak detectability are also considered,
which are called \emph{strong periodic detectability}
(a variant of strong detectability, requiring to determine states periodically along all output sequences)
and \emph{weak periodic detectability} (a variant of weak detectability, requiring
to determine states periodically along some output sequence) \cite{Shu2007Detectability_DES}.
Other essentially different variants of detectability such as \emph{eventual strong detectability}
and \emph{weak approximate detectability} can be found in \cite{Zhang2020DetPNFA}.

Most results on detectability of labeled finite-state automata are based on two fundamental 
assumptions of \emph{deadlock-freeness} 
(which implies that a system can always run) and \emph{divergence-freeness},
i.e., having no unobservable cycles
(which implies that the running of a system will always be eventually  observed).
For labeled finite-state automata, under the two assumptions,
an \emph{observer} method (actually the powerset construction used for determinizing nondeterministic finite automata
with $\ep$-transitions \cite{Sipser2006TheoryofComputation}) was proposed to 
verify weak (periodic) detectability in exponential time \cite{Shu2007Detectability_DES}, later a \emph{detector} method
(a reduced version of the observer, obtained by splitting the states of an observer into subsets of cardinality $2$,
previously used in \cite{Caines1988,Caines1991ObserverFiniteAutomata}) 
was proposed verify strong (periodic) detectability in polynomial time \cite{Shu2011GDetectabilityDES}.
Also under the two assumptions, verifying weak (periodic) detectability
was proven to be $\PSPACE$-complete \cite{Zhang2017PSPACEHardnessWeakDetectabilityDES}, verifying strong (periodic)
detectability was proven to be $\NL$-complete \cite{Masopust2018ComplexityDetectabilityDES}.
Recently, be developing a \emph{concurrent-composition} method in \cite{Zhang2019KDelayStrDetDES,Zhang2020DetPNFA}
(similar to but technically different from the structures used in 
\cite{Cassez2008FaultDiagnosisStDyObser,Tripakis2002DiagnosisTimedAutomata}), strong detectability was verified
in polynomial time without any assumption, removing the two assumptions used for years.

For labeled Petri nets with inhibitor arcs, weak detectability was proven to be undecidable
in \cite{Zhang2018WODESDetectabilityLPS} by reducing the undecidable language equivalence 
problem of labeled Petri nets (see \cite{Hack1975PetriNetLanguage}) to negation of weak detectability.
For labeled Petri nets, strong detectability was proven to be decidable under the two previously mentioned 
fundamental assumptions reformulated in labeled Petri nets,
it was also proven that it is $\EXPSPACE$-hard to verify strong detectability,
but weak detectability is undecidable \cite{Masopust2019DetectabilityPetriNet},
which strengthens the related undecidability result proven in \cite{Zhang2018WODESDetectabilityLPS}.
In \cite{Masopust2019DetectabilityPetriNet}, the undecidable language inclusion problem (but not
the language equivalence problem) of labeled Petri nets (also see \cite{Hack1975PetriNetLanguage}) was reduced to
negation of weak detectability, so that the same idea in the reduction constructed in 
\cite{Zhang2018WODESDetectabilityLPS},
i.e., clearing all tokens of the first of the two basic labeled Petri nets, was also implemented.
Later, the decidability result for strong detectability was strengthened to hold under only the 
divergence-freeness assumption 
\cite{Zhang2020bookDDS} by developing a new tool called extended concurrent
composition. All decidable results on labeled Petri nets proven in 
\cite{Masopust2019DetectabilityPetriNet,Zhang2020bookDDS} were obtained by reducing negation of strong 
detectability to satisfiability of some Yen's path formulae 
\cite{Yen1992YenPathLogicPetriNet,Atig2009YenPathLogicPetriNet}.

The notion of observer has been recently extended to a subclass of labeled timed automata in which the automata are
deterministic, there is a single clock that is reset along with every occurrence of every event
and all clock constraints in all transitions are singletons
\cite{Li2021ObserverSpecialTimedAutomata}. This class of labeled timed automata are exactly labeled 
weighted automata over the monoid $(\Q_{\ge0},+)$, denoted by $\Acal^{\Q_{\ge0}}$,
which are a strict subclass of the automata studied
in the current paper. The observer defined in \cite{Li2021ObserverSpecialTimedAutomata} was computed in 
$2$-$\EXPTIME$, and computed in $\EXPTIME$ when the considered automata are divergence-free. 
The method of computing an observer is via unfolding every state $q$ to a finite number $n$ of new states and then 
compute the observer of the newly obtained labeled finite-state automaton as in \cite{Shu2007Detectability_DES},
where $n$ is the maximum among the weights of all outgoing
transitions of $q$. Hence the method does not apply to labeled weighted automata with weights being negative
rational numbers. In addition, the authors also give an example to show that if some weights are irrational 
numbers, the observer may have infinitely many states and infinitely many transitions. In
Remark~\ref{rem3_det_MPautomata}, we will also use an example in the current paper to illustrate how to compute the 
observer defined in \cite{Li2021ObserverSpecialTimedAutomata}.

The notion of observer has also been extended to subclasses of labeled max-plus automata 
over the semiring $\underQ:=(\Q\cup\{-\infty\},\max,+,-\infty,0)$, denoted by $\Acal^{\underQ}$.
In \cite{Lai2021DetUnambiguousWAutomata}, the observer was computed for a divergence-free 
$\Acal^{unam,\underQ}$ ($\Acal^{unam,\underQ}$ is short for an unambiguous $\Acal^{\underQ}$) in $\EXPTIME$,
and in \cite{Lai2021ObserverPolyAmbiguousWPA} the observer was computed 
for a divergence-free, polynomially ambiguous $\Acal^{\underQ}$ with the clones property but no upper bound 
for time complexity was given. In \cite{Lai2021DetUnambiguousWAutomata}, the above mentioned four notions of
detectability of divergence-free $\Acal^{unam,\underQ}$ were verified in $\EXPTIME$ by using the observer. In 
\cite{Lai2021DetUnambiguousWAutomata,Lai2021ObserverPolyAmbiguousWPA}, the authors adopted the max-plus manner
to define detectability but not the real-time manner adopted in \cite{Li2021ObserverSpecialTimedAutomata}
and the current paper. The detectability results in untimed models mentioned above (e.g., \cite{Shu2007Detectability_DES,Zhang2017PSPACEHardnessWeakDetectabilityDES,Masopust2018ComplexityDetectabilityDES,Masopust2019DetectabilityPetriNet,Zhang2020DetPNFA}) 
are all in the real-time manner.
Consider a sequence $q_0\xrightarrow[]{e_1}\cdots\xrightarrow[]{e_n}q_n$ of transitions (called a path),
where $q_i$, $0\le i\le n$, are states,
$e_j$, $1\le j\le n$, are events; in the real-time manner, the timed word of the path is
$(e_1,t_1)\dots(e_n,t_n)$, where $t_j$ is the instant when $e_j$ occurs in the path; while in the max-plus manner,
the timed sequence of the path is $(e_1,t_1')\dots(e_n,t_n')$, $t_j'$ is the maximal time
for $e_1,\dots,e_j$ to occur among all different paths having $e_1\dots e_j$ as their event sequence and having 
$q_j$ as the final state, 
so $t_j'\ge t_j$. A detailed comparison will be given in Remark~\ref{rem2_det_MPautomata}.
The overlaps between the results in \cite{Lai2021ObserverPolyAmbiguousWPA} and the results in the current paper 
are the results of \cite{Lai2021DetUnambiguousWAutomata}, because in $\Acal^{unam,\underQ}$, 
under every event sequence, there exists at most one path from the initial 
states to any given state, resulting in that the max-plus manner coincides with the real-time manner.
The overlaps between the results in \cite{Lai2021ObserverPolyAmbiguousWPA} and the results in 
\cite{Li2021ObserverSpecialTimedAutomata} are a strict subset of the results in \cite{Lai2021DetUnambiguousWAutomata},
i.e., the observer of a divergence-free $\Acal^{unam,\underline{\Q_{\ge0}}}$ ($\Acal^{unam,\underline{\Q_{\ge0}}}$
is short for a labeled unambiguous weighted automaton over the semiring  
$\underline{\Q_{\ge0}}:=(\Q_{\ge0}\cup\{-\infty\},\max,+,-\infty,0)$).
The relations of the results in \cite{Lai2021DetUnambiguousWAutomata,Lai2021ObserverPolyAmbiguousWPA,Li2021ObserverSpecialTimedAutomata} 
and the current paper are shown in Fig.~\ref{fig22_det_MPautomata}.
\begin{figure}[!htbp]
	\begin{center}
		\begin{tikzpicture}
			\def\firstcircle{(1,0) circle (1)}
			\def\secondcircle{(0,0) circle (1)}
			\def\thirdcircle{(1,0.0) circle (0.6)}

			\scope
			\clip (-1,-1) rectangle (2,1)
				  \firstcircle;
			\fill [red!50] \secondcircle;
			\endscope
			\scope
			\clip (-1,-1) rectangle (2,1)
				  \secondcircle;
			\fill [green!50] \thirdcircle;
			\endscope
			\scope
			\clip (-1,-1) rectangle (2,1)
				  \secondcircle;
			\clip (-1,-1) rectangle (2,1)
				  \thirdcircle;
			\fill [blue!50] \firstcircle;
			\endscope
			\scope
			\clip \secondcircle;
			\fill [cyan!50] \thirdcircle;
			\endscope
			\scope
			\clip (-1,-1) rectangle (2,1)
				  \thirdcircle;
			\clip \firstcircle;
			\fill [yellow!50] \secondcircle;
			\endscope

			\node[black] at (-0.5,0) {$A$};
			\node[black] at (1.25,0.0) {$E$};
			\node[black] at (1.8,0.0) {$B$};
			\node[black] at (0.75,0.0) {$D$};
			\node[black] at (0.2,0.0) {$C$};
		\end{tikzpicture}
	\end{center}
	\caption{Relations of the results in
	\cite{Lai2021DetUnambiguousWAutomata,Lai2021ObserverPolyAmbiguousWPA,Li2021ObserverSpecialTimedAutomata}
	and the current paper, where \cite{Lai2021DetUnambiguousWAutomata}$=C\cup D$ (i.e., the results in 
	\cite{Lai2021DetUnambiguousWAutomata} are represented by $C\cup D$, the following equalities have similar meanings),
	\cite{Lai2021ObserverPolyAmbiguousWPA}$=A\cup C\cup D$, \cite{Li2021ObserverSpecialTimedAutomata}$=D\cup E$,
	the current paper$=B\cup C\cup D\cup E$. The automata considered in $A\cup C\cup D$ are divergence-free, the 
	automata considered in $E\cup B$ are not divergence-free.}
	\label{fig22_det_MPautomata}   
\end{figure}
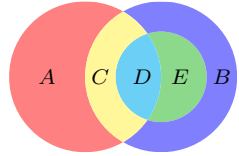

The results in \cite{Lai2021DetUnambiguousWAutomata} generally do not apply to an automaton $\Acal^{unam,\underQ}$
that is not divergence-free (see Remark~\ref{rem7_det_MPautomata}). Although the observer is computed
for a general automaton $\Acal^{\Q_{\ge0}}$ 
in \cite{Li2021ObserverSpecialTimedAutomata}, the observer cannot always be used for verifying detectability 
for an automaton $\Acal^{\Q_{\ge0}}$ that is not divergence-free; this is even true in labeled finite-state automata 
\cite[Remark~2]{Zhang2020DetPNFA}. \textbf{Detectability of general $\Acal^{unam,\underQ}$ and general
$\Acal^{\Q_{\ge0}}$ can be verified by using the methods developed in the current paper, for the first time.}
Apart from the above advantages of the current paper compared with
\cite{Li2021ObserverSpecialTimedAutomata,Lai2021DetUnambiguousWAutomata}, we also
show that $\Acal^{unam,\underN}$ and $\Acal^{unam,\N}$ are already fundamentally more complicated than
a labeled finite-state automaton $\Acal$ by showing in Theorem~\ref{thm12_det_MPautomata} that, the problems of 
verifying strong (periodic) detectability of deterministic, deadlock-free, and divergence-free 
$\Acal^{unam,\underN}$ and $\Acal^{unam,\N}$ are $\coNP$-hard, where $\Acal^{unam,\underN}$ denotes a labeled
unambiguous weighted automaton over the max-plus semiring $\underN:=(\N\cup\{-\infty\},\max,+,-\infty,0)$,
$\Acal^{unam,\N}$ denotes a labeled unambiguous weighted automaton over monoid $(\N,+,0)$,
because as previously mentioned, strong (periodic) detectability of automaton $\Acal$ can be verified in polynomial
time \cite{Shu2011GDetectabilityDES,Zhang2019KDelayStrDetDES}.

Probabilistic finite-state automata were also widely studied models in DESs (e.g., in
\cite{Keroglou2017DetProbAutomata,Yin2017InitialStateDetectabilityStoDES}). 
They are actually weighted automata over the semiring $(\R,+,\cdot,0,1)=:\underR$, 
but the weights are only chosen from $[0,1]$. Because only
probabilities are computed, no computation produces a value outside $[0,1]$. 
In such automata, the reliability of an event sequence is
defined as the sum of the reliabilities of its successful paths, where the reliability of a successful 
path is the product of the probabilities (i.e., weights) of all its transitions. 
Hence the detectability notions studied in 
\cite{Keroglou2017DetProbAutomata,Yin2017InitialStateDetectabilityStoDES}
are defined in a totally different way compared with those in \cite{Lai2021DetUnambiguousWAutomata,Li2021ObserverSpecialTimedAutomata}
and the current paper. 
On the other hand, in this paper we will give equivalent conditions for four fundamental 
notions of detectability of labeled weighted automata over any monoids, which include the results 
on labeled weighted automata over the monoid $(\R,\cdot,1)$ as a special class.

\subsection{Contributions of the paper}


The first contribution is on a general labeled weighted automaton over a monoid, denoted by $\Acal^{\frakM}$.
\begin{enumerate}
	\item We for the first time formulate the notions of concurrent composition, observer, and detector 
		for $\Acal^{\frakM}$, which are natural but nontrivial extensions of those for labeled
		finite-state automata.
		We use the notion of concurrent composition to give an equivalent condition for strong detectability,
		use the notion of observer to give equivalent conditions for weak detectability and weak 
		periodic detectability, and use the notion of detector to give an equivalent
		condition for strong periodic detectability, all for general $\Acal^{\frakM}$ without any assumption.
\end{enumerate}


The second contribution of the paper is on labeled weighted automata over the monoid $(\Q^k,+,0_k)$ (denoted by
$\Acal^{\Q^k}$), where $0_k\in\Q^k$ denotes the $k$-length zero vector,
in which the special results on unambiguous $\Acal^{\Q}$ also hold for labeled unambiguous weighted
automata over semiring $\underQ$ (denoted by $\Acal^{unam,{\underQ}}$), because the four notions of detectability
of $\Acal^{unam,\underQ}$ in \cite{Lai2021DetUnambiguousWAutomata} coincide with the four notions 
of detectability of $\Acal^{unam,\Q}$ in the current paper correspondingly (except for minor and neglectable 
differences, see Remark~\ref{rem2_det_MPautomata}).

\begin{enumerate}\setcounter{enumi}{1} 
	\item We find novel connections between $\Acal^{\Q^k}$ and the exact path length problem 
		\cite{Nykanen2002ExactPathLength} and a subclass of Presburger arithmetic 
		\cite{Graedel1988PresburgerArithmeticComplexity}
		so that detectability of $\Acal^{\Q^k}$ can be verified. 
	\item For $\Acal^{\Q^k}$, we prove that its observer and detector
		can be computed in $2$-$\EXPTIME$, its self-composition 
		can be computed in $\NP$, all in the size of $\Acal^{\Q^k}$.
	\item We prove that strong detectability of $\Acal^{\Q^k}$ can be verified in $\coNP$, 
		while strong periodic detectability,
		weak detectability, and weak periodic detectability of $\Acal^{\Q^k}$ can be 
		verified in $2$-$\EXPTIME$, all in the size of $\Acal^{\Q^k}$.
		Particularly for $\Acal^{\Q^k}$ in which from every state, a distinct state can be reached through 
		some unobservable, instantaneous path, its detector can be computed in $\NP$, and its strong periodic
		detectability can be verified in $\coNP$. 
		We also prove that the problems of verifying strong detectability
		and strong periodic detectability of deterministic, deadlock-free, and divergence-free $\Acal^{\N}$ 
		are both $\coNP$-hard. See Table~\ref{tab1_det_MPautomata} and Table~\ref{tab2_det_MPautomata} as 
		collections of related results.
\end{enumerate}

Finally, in order to differentiate labeled weighted automata over monoids from labeled timed
automata, we also initially explore detectability in labeled timed automata.

\begin{enumerate}\setcounter{enumi}{4}
	\item We prove that in labeled timed automata, the strong detectability verification problem
		is PSPACE-complete, while weak (periodic) detectability is undecidable.
\end{enumerate}

\begin{table}[!htbp]
	\centering{
	\begin{tabular}{|c||c|c|c|}
		\hline\rowcolor{lightgray}
		& SD & SPD & WD, WPD \\
		\hline
		\begin{tabular}[]{c}
		$\Acal$
		\end{tabular}
		&
		\begin{tabular}[]{c}
			$\PTIME$ \cite{Zhang2019KDelayStrDetDES}\\
			$\NL$ \cite{Zhang2021UnifyingDetDiagPred}
		\end{tabular}
		
		&
		$\PTIME$ (Cor.~\ref{cor7_det_MPautomata})
		&
		$\PSPACE$ \cite{Zhang2020DetPNFA}
		\\
		\hline
		\begin{tabular}[]{c}
			d.d. $\Acal$
		\end{tabular}
		&
		\begin{tabular}{c}
			$\PTIME$ \cite{Shu2011GDetectabilityDES}\\
			$\NL$-c. \cite{Masopust2018ComplexityDetectabilityDES}
		\end{tabular}
		&
		\begin{tabular}{c}
			$\PTIME$ \cite{Shu2011GDetectabilityDES}\\
			$\NL$-c. \cite{Masopust2018ComplexityDetectabilityDES}
		\end{tabular}
		&
		\begin{tabular}{c}
			$\EXPTIME$ \cite{Shu2007Detectability_DES}\\
			$\PSPACE$-c. \cite{Zhang2017PSPACEHardnessWeakDetectabilityDES}
		\end{tabular}
		\\
		\hline
		\begin{tabular}{c}
			$\Acal^{\Q^k}$
		\end{tabular} & 
		\begin{tabular}[]{c}
			$\coNP$-c.\\
			(Thms.~\ref{thm4_det_MPautomata},~\ref{thm12_det_MPautomata})
		\end{tabular}
		&
		\begin{tabular}[]{c}
			$2$-$\EXPTIME$\\
			(Thm.~\ref{thm11_det_MPautomata})\\
			$\coNP$-h.\\
			(Thm.~\ref{thm12_det_MPautomata})
		\end{tabular}
		& 
		\begin{tabular}[]{c}
			$2$-$\EXPTIME$\\
			(Thms.~\ref{thm6_det_MPautomata},~\ref{thm10_det_MPautomata})
		\end{tabular}
		\\
		\hline
		\begin{tabular}[]{c}
			d.d. $\Acal^{\Q^k}$
		\end{tabular} 
		&
		\begin{tabular}[]{c}
			$\coNP$-c.\\
			(Thms.~\ref{thm4_det_MPautomata},~\ref{thm12_det_MPautomata})
		\end{tabular}
		&
		\begin{tabular}[]{c}
			$\EXPTIME$\\
			(Thm.~\ref{thm14_det_MPautomata})\\
			$\coNP$-h.\\
			(Thms.~\ref{thm12_det_MPautomata})
		\end{tabular}
		& 
		\begin{tabular}[]{c}
			$\EXPTIME$\\
			(Cor.~\ref{cor6_det_MPautomata})
		\end{tabular}
		\\\hline
		\begin{tabular}{c}
			$\Acal^{u.,\Q^k}$\\ $\Acal^{u.,\underQ}$
		\end{tabular} & 
		\begin{tabular}[]{c}
			$\coNP$-c.\\
			(Cor.~\ref{cor4_det_MPautomata})
		\end{tabular}
		&
		\begin{tabular}[]{c}
			$2$-$\EXPTIME$\\
			(Thm.~\ref{thm11_det_MPautomata})\\
			$\coNP$-h.\\
			(Thm.~\ref{thm12_det_MPautomata})
		\end{tabular}
		& 
		\begin{tabular}[]{c}
			$2$-$\EXPTIME$\\
			(Thms.~\ref{thm6_det_MPautomata},~\ref{thm10_det_MPautomata})
		\end{tabular}\\\hline
		\begin{tabular}{c}
			d.d. $\Acal^{u.,\Q^k}$\\
			d.d. $\Acal^{u.,\underQ}$
		\end{tabular} & 
		\begin{tabular}[]{c}
			$\EXPTIME$ \cite{Lai2021DetUnambiguousWAutomata}\\
			$\coNP$-c.
			(Cor.~\ref{cor4_det_MPautomata})
		\end{tabular}
		&
		\begin{tabular}[]{c}
			$\EXPTIME$ \cite{Lai2021DetUnambiguousWAutomata}\\
			(Thm.~\ref{thm14_det_MPautomata})\\
			$\coNP$-h.
			(Thm.~\ref{thm12_det_MPautomata})
		\end{tabular}
		& 
		\begin{tabular}[]{c}
			$\EXPTIME$ \cite{Lai2021DetUnambiguousWAutomata}\\
			(Cor.~\ref{cor6_det_MPautomata})
		\end{tabular}\\
		\hline
		\begin{tabular}{c}
			d.d.d. $\Acal^{\N}$\\
			d.d.d $\Acal^{\underN}$
		\end{tabular} &
		\begin{tabular}{c}
			$\coNP$ (Cor.~\ref{cor4_det_MPautomata})\\
			$\coNP$-h. (Thm.~\ref{thm12_det_MPautomata})
		\end{tabular}
		&
		\begin{tabular}{c}
			$\EXPTIME$\\ (Thm.~\ref{thm14_det_MPautomata})\\
			$\coNP$-h. (Thm.~\ref{thm12_det_MPautomata})
		\end{tabular}
		&
		\begin{tabular}[]{c}
			$\EXPTIME$ \cite{Lai2021DetUnambiguousWAutomata}\\
			(Cor.~\ref{cor6_det_MPautomata})
		\end{tabular}\\
		\hline
	\end{tabular}}
	\caption{Results on complexity of verifying four notions of detectability of automata, where
	SD, SPD, WD, and WPD are short for strong detectability, strong periodic detectability, weak
	detectability, and weak periodic detectability, respectively;
	$\Acal$ denotes a labeled finite-state automaton, $\Acal^{\Q^k}$ (resp., $\Acal^{\underQ}$, $\Acal^{\underN}$,
	$\Acal^{\N}$)
	denotes a labeled weighted automaton over monoid $(\Q^k,+,0_k)$ (resp., max-plus semiring 
	$(\Q\cup\{-\infty\},\max,+,-\infty,0)$, max-plus semiring $(\N\cup\{-\infty\},\max,+,-\infty,0))$,
	monoid $(\N,+,0)$), ``$u.$'' is short for ``unambiguous'', ``d.d.'' is short for ``deadlock-free, divergence-free'',
	``d.d.d'' is short for ``deterministic, deadlock-free, divergence-free'', ``c.'' denotes ``complete'',
	``h.'' denotes ``hard''.} 
	\label{tab1_det_MPautomata}
\end{table}

\begin{table}[!htbp]
	\centering{
	\begin{tabular}{|c||c|c|c|}
		\hline
		\rowcolor{lightgray} & observer & detector & self-composition \\\hline
		$\Acal$ & $\EXPTIME$ \cite{Shu2007Detectability_DES} & $\PTIME$ \cite{Shu2011GDetectabilityDES} 
		& $\PTIME$ \cite{Zhang2019KDelayStrDetDES,Zhang2020DetPNFA} \\\hline
		$\Acal^{\Q^k}$ &
		$2$-$\EXPTIME$ (Thm.~\ref{thm5_det_MPautomata})
		& $2$-$\EXPTIME$ (Thm.~\ref{thm13_det_MPautomata})
		& $\NP$ (Thm.~\ref{thm3_det_MPautomata})
		\\\hline
		$\Acal^{\Q_{\ge0}}$ &
		$2$-$\EXPTIME$
		\cite{Li2021ObserverSpecialTimedAutomata} 
		(Thm.~\ref{thm5_det_MPautomata})
		& $2$-$\EXPTIME$ 
		(Thm.~\ref{thm13_det_MPautomata}) 
		& $\NP$ (Thm.~\ref{thm3_det_MPautomata})
		\\\hline
		\begin{tabular}{c}
			d. $\Acal^{\Q^k}$ 
		\end{tabular} &
		$\EXPTIME$ (Cor.~\ref{cor5_det_MPautomata})
		& $\EXPTIME$ (Cor.~\ref{cor8_det_MPautomata}) & $\NP$ (Thm.~\ref{thm3_det_MPautomata})
		\\\hline
		d. $\Acal^{\Q_{\ge0}}$ &
		$\EXPTIME$
		\cite{Li2021ObserverSpecialTimedAutomata} (Cor.~\ref{cor5_det_MPautomata})
		& $\EXPTIME$ 
		(Cor.~\ref{cor8_det_MPautomata}) &
		$\NP$ (Thm.~\ref{thm3_det_MPautomata})
		\\\hline
		\begin{tabular}{c}
			d. $\Acal^{u.,\Q^k}$\\
			d. $\Acal^{u.,\underQ}$
		\end{tabular} & $\EXPTIME$ \cite{Lai2021DetUnambiguousWAutomata} (Cor.~\ref{cor5_det_MPautomata})
		& $\EXPTIME$ 
		(Cor.~\ref{cor8_det_MPautomata}) 
		& $\NP$ (Thm.~\ref{thm3_det_MPautomata})
		\\\hline
	\end{tabular}} 
	\caption{ 
	Results on complexity of computing observers, detectors, and self-compositions of automata,
	where $\Acal$, $\Acal^{\Q^k}$, $\Acal^{\underQ}$, $u.$ 
	are the same as those in Table~\ref{tab1_det_MPautomata},  
	$\Acal^{\Q_{\ge0}}$ denotes a labeled weighted automaton over the monoid $(\Q_{\ge0},+,0)$,
	``d.'' is short for ``divergence-free''.}
	\label{tab2_det_MPautomata}
\end{table}

\section{Preliminaries}

\subsection{Notation}

Symbols $\N$, $\Z$, $\Z_{+}$, $\Q$, $\Q_{\ge0}$, $\R$, and $\R_{\ge0}$ denote the sets of nonnegative integers,
integers, positive integers, rational numbers, nonnegative rational numbers, real numbers, and nonnegative
real numbers, respectively. Symbol $0_k$ denotes the $k$-length zero vector.
For a finite alphabet $\Sig$, $\Sig^*$ and $\Sig^{\omega}$ are used to denote the set of
\emph{words} (i.e., finite-length sequences of elements of $\Sig$) over $\Sig$ including the empty word $\epsilon$
and the set of \emph{configurations} (i.e., infinite-length sequences of elements of $\Sig$) over $\Sig$,
respectively. $\Sig^{+}:=\Sig^*\setminus\{\epsilon\}$.
For a word $s\in \Sig^*$,
$|s|$ stands for its length, and we set $|s'|=+\infty$ for all $s'\in \Sig^{\omega}$.
For $s\in \Sig^+$ and $k\in\N$, $s^k$ and $s^{\omega}$ denote the concatenations of $k$ copies of
$s$ and infinitely many copies of $s$, respectively. Analogously, $L_1L_2:=\{e_1e_2|e_1\in L_1,
e_2\in L_2\}$, where $L_1,L_2\subset \Sig^*$.
For a word (configuration) $s\in \Sig^*(\Sig^{\omega})$, a word $s'\in \Sig^*$ is called a \emph{prefix} of $s$,
denoted as $s'\sqsubset s$,
if there exists another word (configuration) $s''\in \Sig^*(\Sig^{\omega})$ such that $s=s's''$.
For two nonnegative integers $i\le j$, $\llb i,j\rrb$ denotes the set of all integers no less than $i$ and no greater 
than $j$; and for a set $S$, $|S|$ denotes its cardinality and $2^S$ its power set. Symbols $\subset$
and $\subsetneq$ denote the subset and strict subset relations, respectively.

We will use the \emph{exact path length} (EPL) problem, the \emph{subset sum}
problem, and a subclass of \emph{Presburger arithmetic} in the literature to prove the main results.

\subsection{The exact path length problem}

Consider a $k$-dimensional weighted directed graph $G=(\Q^k,V,A)$, where $k\in\Z_{+}$, 
$\Q^{k}=\underbrace{\Q\times\cdots\times\Q}_{k}$,
$V$ is a finite set of vertices, $A\subset V\times\Q^k\times V$ a finite set of weighted edges (arcs) with
weights in $\Q^k$. For a path $v_1\xrightarrow[]{z_1}\cdots\xrightarrow[]{z_{n-1}}v_n$, its weight 
is defined by $\sum_{i=1}^{n-1}z_i$. For an edge $a=(v_1,z,v_2)\in A$, also denoted by $v_1\xrightarrow[]
{z}v_2$, we call $v_1$ and $v_2$ the \emph{tail} (denoted by $\tail(a)$) and the \emph{head} (denoted by
$\head(a)$) of $a$, respectively. The EPL problem \cite{Nykanen2002ExactPathLength}
is stated as follows.
\begin{problem}[EPL]\label{prob1_det_MPautomata}
	Given a positive integer $k$, a $k$-dimensional weighted directed graph $G=(\Q^k,V,A)$,
	two vertices $v_1,v_2\in V$, and a vector $z\in \Q^k$, determine
	whether there exists a path from $v_1$ to $v_2$ with weight $z$.
\end{problem}

We set as usual that for a positive integer $n$, the size $\size(n)$ of $n$ to be the length of its binary
representation, i.e., $\size(n)=\ceil{\log_2{(n+1)}}$ ($\ceil{\cdot}$ is the ceiling function),
$\size(-n)=1+\size(n)$; $\size(0)=1$; 
for a rational number $m/n$, where $m,n$ are relatively prime integers,
$\size(m/n)=\size(m)+\size(n)$; then for a vector $z\in\Q^{k}$, its size is the sum of the sizes of
its entries. The size of an instance $(k,G,v_1,v_2,z)$ of the EPL problem is defined by
$\size(k)+\size(G)+2+\size(z)$, where $\size(G)=|V|+\size(A)$, $\size(A)=\sum_{(v_1,z',v_2)\in A}
(2+\size(z'))$.

\begin{lemma}[\cite{Nykanen2002ExactPathLength}]\label{lem1_det_MPautomata}
	The EPL problem belongs to $\NP$.\footnote{
	Note that the original EPL problem studied in \cite{Nykanen2002ExactPathLength} is on graph $(\Z^k,V,A)$.
	However, the proof (a polynomial-time reduction from EPL to integer linear programming)
	in \cite{Nykanen2002ExactPathLength} also applies to the more general case for graph
	$(\Q^k,V,A)$, resulting in Lemma~\ref{lem1_det_MPautomata}.}
	The EPL problem is $\NP$-hard already for graph $(\N,V,A)$. 
\end{lemma}

\subsection{The subset sum problem}

The subset sum problem \cite[p. 223]{Garey1990ComputerIntractability} is as follows.
\begin{problem}[subset sum]\label{prob2_det_MPautomata}
	Given positive integers $n_1,\dots,n_m$, and $N$, determine whether $N=\sum_{i\in I}n_i$ for some
	$I\subset\llb 1,m\rrb$.
\end{problem}

\begin{lemma}[\cite{Garey1990ComputerIntractability}]\label{lem2_det_MPautomata}
	The subset sum problem is $\NP$-complete.
\end{lemma}

\subsection{Presburger arithmetic}

We will use a subclass of Presburger arithmetic.
A Presburger formula/sentence of this subclass is as follows:
\begin{align}\label{eqn34_det_MPautomata}
	(Q_1 x_1\in\N)\dots(Q_s x_s\in\N)[\Phi(x_1,\dots,x_{s})],
\end{align}
where $Q_1,\dots,Q_s$ are any \emph{quantifier prefix} ($Q_i=\exists$ (\emph{existential quantifier}) or $\forall$
(\emph{universal quantifier})),
$x_1,\dots,x_s$ are variables, 
$\Phi(x_1,\dots,x_{s})$ is a formula consisting of a Boolean combination of 
linear inequalities of the form 
\begin{align}\label{eqn35_det_MPautomata}
	a_1x_1+\cdots+a_{s}x_{s} \le  b 
\end{align}
with $a_1,\dots,s_{s},b$ constant integers.

For example,
$(x_1=1)=(x_1\ge 1)\wedge(x_1\le 1)=\neg(x_1<1)\wedge(x_1\le 1)=
\neg(x_1\le 0)\wedge(x_1\le 1)$ and 
$(x_1=1)\implies (x_2>2)=\neg(x_1=1) \vee \neg (x_2\le 2)$ are such quantifier-free formulae.

\begin{lemma}[\cite{Graedel1988PresburgerArithmeticComplexity}]\label{lem4_det_MPautomata}
	Consider a Presburger sentence as in \eqref{eqn34_det_MPautomata} of length $r$ with $m$ quantifier
	alternations (i.e., with $m$ blocks of adjacent quantifiers of the same kind). Then the sentence
	is satisfied if and only if $$(Q_1 x_1\le w)\dots(Q_s x_s\le w)
	[\Phi(x_1,\dots,x_{s})]$$ is satisfied, where $w=2^{cr^{(s+3)^{m+1}}}$, $c$ is a constant. 
\end{lemma}

This yields a decision procedure:
one can first compute $w$ from $r,m,s$, and then check all
$s$-tuples of nonnegative integers $x_1,\dots,x_{s}$ with $x_i\le w$, whether $\Phi(x_1,\dots,x_{s})$
is true.

In \eqref{eqn34_det_MPautomata}, if $s\in\Z_{+}$ is also regarded as input, $Q_1=\cdots=Q_s=\exists$, and
$\Phi(x_1,\dots,x_{s})$ is conjunctions of 
linear equations of the form $a_1x_1+\cdots+a_{s}x_{s}= b$ with 
$a_1,\dots,a_{s},b$ constant rational numbers,
then \eqref{eqn34_det_MPautomata} becomes the $\NP$-complete integer linear programming
\cite{Papadimitriou1981IntegerProgramming,Schrijver1984LinearIntegerProgramming}.

\subsection{Labeled weighted automata over monoids}\label{subsec:WAM}

A \emph{monoid} is a triple $\frakM=(T,\otimes,{\bf1})$, where for all $a,b,c\in T$,
\begin{itemize}
	\item $a\otimes b\in T$,
	\item (associativity) $(a\otimes b)\otimes c=a\otimes(b\otimes c)$,
	\item $a\otimes {\bf 1}={\bf 1}\otimes a=a$ (${\bf 1}\in T$ is called \emph{identity} of $\frakM$).
\end{itemize}
Particularly, if there exists an element ${\bf0}\in T$ such that ${\bf0}\otimes a=a\otimes {\bf0}={\bf0}$
for all $a\in T$, then we call ${\bf0}$ \emph{zero} of $\frakM$. Any monoid has exactly one \emph{identity}
and at most one \emph{zero}.

A \emph{labeled weighted automaton over monoid $\frakM=(T,\otimes,{\bf1})$} is a tuple
\begin{align}\label{LWA_monoid_det_MPautomata}
	\Acal^{\frakM}=(\frakM,Q,E,Q_0,\Dt,\alpha,\mu,\Sig,\ell),
\end{align}
where $Q$ is a finite set of \emph{states},
$E$ a finite \emph{alphabet} (elements of $E$ are called \emph{events}),
$Q_0\subset Q$ a set of \emph{initial states}, $\Delta\subset Q\times E\times Q$ a \emph{transition relation}
(elements of $\Delta$ are called \emph{transitions}, a transition $(q,e,q')\in\Delta$ is interpreted as 
when $\Acal^\frakM$ is in state $q$ and event $e$ occurs, $\Acal^\frakM$ transitions to state $q'$),
$\alpha$ assigns to each initial state $q_0\in Q_0$
a nonzero \emph{weight} $\alpha(q_0)\in T$, $\mu$ assigns to each transition $(q,e,q')\in\Delta$ 
(or rewritten as $q\xrightarrow[]{e}q'$)
a nonzero \emph{weight} $\mu(e)_{qq'}\in T$, where the transition is also denoted by 
$q\xrightarrow[]{e/\mu(e)_{qq'}}q'$, $\Sig$ is a finite set of \emph{outputs/labels}, and
$\ell:E\to\Sig\cup\{\ep\}$ is a \emph{labeling function}.

Particularly, $\Acal^{\Q^k}$, $\Acal^{\Q_{\ge0}}$, $\Acal^{\N}$ denote $\Acal^{\frakM}$ in which
$\frakM$ is specified as $(\Q^k,+,0_k)$, $(\Q_{\ge0},+,0)$, $(\N,+,0)$, where $\Acal^{\Q_{\ge0}}$
and $\Acal^{\N}$ can represent timed DESs.
The size $\size(\Acal^{\Q^k})$ of a given $\Acal^{\Q^k}$ is defined by
$|Q|+|\Delta|+\size(\alpha)+\size(\mu)+\size(\ell)$,
where 
the size of a rational vector has already been defined before,
$\size(\alpha)=|Q_0|+\sum_{q\in Q_0}\size(\alpha(q))$, $\size(\mu)=\sum_{(q,e,q')\in \Delta}
\size(\mu(e)_{qq'})$, $\size(\ell)=|\{(e,\ell(e))|e\in E\}|$. The size of a given
$\Acal^{\Q_{\ge0}}$ (resp., $\Acal^{\N}$) can be defined analogously.

\begin{remark}\label{rem1_det_MPautomata}
A \emph{labeled finite-state automaton} (studied in \cite{Shu2007Detectability_DES,Shu2011GDetectabilityDES,Masopust2018ComplexityDetectabilityDES,Zhang2017PSPACEHardnessWeakDetectabilityDES}, etc.)
can be regarded as automaton $\Acal^{\N}$ such that all 
unobservable transitions
are instantaneous and every two observable transitions with the same label have
the same weight in $\N$. The observer of such $\Acal^{\N}$ can be computed in exponential time 
\cite{Shu2007Detectability_DES}. 
In the sequel, we use $\Acal$ to denote a labeled finite-state automaton (without weights).
\end{remark}

Events in $E_{uo}=\{e\in E|\ell(e)=\ep\}$ are called
\emph{unobservable}, events in $E_o=\{e\in E|\ell(e)\ne\ep\}$ are called \emph{observable}. When an observable
event $e\in E_o$ occurs, $\ell(e)$ is observed; but when an unobservable event $e\in E_{uo}$ occurs, nothing is
observed. For every $q\in Q$, we also regard $q\xrightarrow[]{\ep/{\bf1}}q$ as a transition.
A transition $q\xrightarrow[]{e/\mu(e)_{qq'}}q'$ is called \emph{instantaneous} if $\mu(e)_{qq'}
={\bf1}$, and called \emph{noninstantaneous} otherwise. A transition $q\xrightarrow[]{e/\mu(e)_{qq'}}q'$
is called \emph{observable (resp., unobservable)} if $e$ is observable (resp., unobservable).
We denote by $\Dt_o=\{(q,e,q')\in\Dt|\ell(e)\ne\ep\}$ and $\Dt_{uo}=\{(q,e,q')\in\Dt|\ell(e)=\ep,
e\ne\ep\}$ 
the sets of observable transitions and unobservable 
transitions, respectively. Particularly, we denote $E_a=\{e\in E_o|\ell(e)=a\}$,
$\Dt_a=\{(q,e,q')\in\Dt|\ell(e)=a\}$, where $a\in\Sig$.
Automaton $\Acal^\frakM$ is called
\emph{deterministic} if (1) $|Q_0|=1$ and (2) for all
states $q,q',q''\in Q$ and events $e\in E$, if $(q,e,q')\in \Delta$ and $(q,e,q'')\in\Delta$ then $q'=q''$
(hence one also has $\mu(e)_{qq'}=\mu(e)_{qq''}$).

From now on, without loss of generality, we assume for each initial state $q_0\in Q_0$, 
$\alpha(q_0)={\bf1}$, because otherwise we add a new initial state $q_0'\notin Q$ and set 
$\alpha(q_0')={\bf 1}$; and then for each initial state $q_0\in Q_0$ such that $\alpha(q_0)
\ne {\bf 1}$, we let $q_0$ not be initial any more, and add a transition $q_0'\xrightarrow[]{\varepsilon/
\alpha(q_0)} q_0$, where $\varepsilon$ is a new event not in $E$ and unobservable.
The case for $\Acal^{\Q_{\ge0}}$ is interpreted as follows: if an automaton was initially
in state $q_0$, then before instant $\alpha(q_0)$, no event occurred, hence nothing could be observed.
So it makes sense to set $\varepsilon$ to be unobservable.

Particularly for $\Acal^{\Q_{\ge0}}$, for initial state $q\in Q_0$, $\alpha(q)$ denotes its initial
time delay, and in a transition $q\xrightarrow[]{e/\mu(e)_{qq'}}q'$, $\mu(e)_{qq'}$ denotes its
time delay (i.e., the time consumption of its execution). Hence the execution of an instantaneous
transition has time delay $0$, i.e., does not cost time,
while the execution of a noninstantaneous transition has time delay a positive rational number
$\mu(e)_{qq'}$, i.e., costs time $\mu(e)_{qq'}$. 
As pointed out before, without loss of generality, we assume $\alpha(q_0)=0$ for all $q_0\in Q_0$.

For $q_0,\dots,q_n\in Q$ and $e_1,\dots,e_n\in E$, 
$n\in \Z_{+}$, we call a sequence
\begin{align}\label{path_det_MPautomaton}
	\pi:=q_0\xrightarrow[]{e_1}q_1\xrightarrow[]{e_2}\cdots\xrightarrow[]{e_n}q_n
\end{align}
of transitions a (\emph{finite}) \emph{path}. 
A state $q$ is called \emph{reachable} if it is initial or there exists a path from some initial state
to $q$.
A path $\pi$ is called \emph{simple} if $q_0,\dots,q_{n}$ are pairwise different.
A path $\pi$ is called a \emph{cycle} if $q_0=q_n$, and a \emph{simple cycle} is such that $q_0,\dots,q_{n-1}$
are pairwise different. A path $\pi$ is called \emph{unobservable} if $\ell(e_1\dots e_n)=\ep$,
and called \emph{observable} otherwise.
The set of paths starting at $q_0\in Q$ and ending at $q\in Q$ is denoted by $q_0\rightsquigarrow{}q$.
Particularly, for $e_1,\dots,e_n\in E$, $q_0\xrsquigarrow{e_1\dots e_n}q$ denotes the set of all paths
under $e_1\dots e_n$, i.e., the paths
$q_0\xrightarrow[]{e_1}q_1\xrightarrow[]{e_2}\cdots \xrightarrow[]{e_{n-1}}q_{n-1}\xrightarrow[]{e_n}q$, where
$q_1,\dots,q_{n-1}\in Q$. Automaton $\Acal^\frakM$ is called \emph{unambiguous} if under every event sequence,
there exists at most one path from the initial states to any given state, i.e., for all $s\in E^+$ and 
$q\in Q$, one has $\left|\bigcup_{q_0\in Q_0}(q_0\xrsquigarrow{s}q)\right|\le 1$. If $\Acal^{\frakM}$ is deterministic
then it is unambiguous.

The \emph{weighted word} of path $\pi$ is defined by 
\begin{align}\label{timedword_det_MPautomaton}
 	\tau(\pi):=(e_1,t_1)(e_2,t_2)\dots(e_n,t_n),
\end{align}
where for all $i\in\llb 1,n\rrb$, 
$t_i=\bigotimes_{j=1}^{i}\mu(e_j)_{q_{j-1}q_j}$.The \emph{weight} of path $\pi$ is defined by
$t_n=:\WEIGHT_{\pi}$.
A path $\pi$ is called \emph{instantaneous} if $t_1=\cdots=t_n={\bf1}$, and called \emph{noninstantaneous}
otherwise.

Particularly for $\Acal^{\Q_{\ge0}}$,
one has $t_i=\sum_{j=1}^{i}{\mu(e_j)_{q_{j-1}q_j}}$, hence
the $t_i$ in $\tau(\pi)$ can be used to denote the total time consumptions for
the first $i$ transitions in path $\pi$, $i\in\llb 1,n\rrb$.
Hence we also call a weighted word of automaton $\Acal^{\Q_{\ge0}}$ \emph{timed word}.
If $\Acal^{\Q_{\ge0}}$ generates a path
$\pi$ as in \eqref{path_det_MPautomaton}, consider its timed  word $\tau(\pi)$ as in 
\eqref{timedword_det_MPautomaton}, then at instant $t_i$, one observes $\ell(e_i)$ if $\ell(e_i)\ne\ep$;
and observes nothing otherwise, where $i\in\llb 1,n\rrb$. We simply say one observes $\ell(\tau(\pi))$.
With this intuitive observation, we will define the notion of current-state estimate in the next section.

We use $L(\Acal^\frakM)$ to denote the set of weighted words of all paths of $\Acal^{\frakM}$ starting 
from initial states.

For $q_0,q_1,\dots$$\in Q$ and $e_1,e_2,\dots$$\in E$, 
we call 
\begin{align}
	\pi:=q_0\xrightarrow[]{e_1}q_1\xrightarrow[]{e_2}\cdots
\end{align}
an \emph{infinite path}. 
The \emph{$\omega$-weighted word} of infinite path $\pi$ is defined by 
\begin{align}
	\tau(\pi):=(e_1,t_1)(e_2,t_2)\dots,
\end{align}
where for all $i\in\Z_{+}$, $t_i=\bigotimes_{j=1}^{i}\mu(e_j)_{q_{j-1}q_j}$.

We use $L^{\omega}(\Acal^\frakM)$ by to denote the set of $\omega$-{weighted} words of all infinite paths of
$\Acal^\frakM$ starting from initial states. Particularly, we also call an $\omega$-weighted word of 
automaton $\Acal^{\Q_{\ge0}}$ \emph{$\omega$-timed word}.

Labeling function $\ell$ is recursively extended to
$E^*\cup E^{\omega}\to \Sig^*\cup\Sig^{\omega}$ as $\ell(e_1e_2\dots)\\=\ell(e_1)\ell(e_2)\dots$.
We also extend $\ell$ as follows:
for all $(e,t)\in E\times T$, $\ell((e,t))=(\ell(e),t)$ if $\ell(e)\ne\ep$, and $\ell((e,t))=\ep$
otherwise. Hence $\ell$ is also recursively extended to $(E\times T)^*\cup(E\times T)^{\omega}\to
(\Sig\times T)^*\cup(\Sig\times T)^{\omega}$. For a weighted word $\tau(\pi)$,
where $\pi$ is a path of $\Acal^{\frakM}$, $\ell(\tau(\pi))$ is called the \emph{weighted label/output
sequence} of both $\pi$ and $\tau(\pi)$. We also extend the previously defined function
$\tau$ as follows: for all $\gamma=(\s_1,t_1)\dots(\s_n,t_n)\in(\Sig\times T)^{*}$, 
\begin{equation}\label{eqn11_det_MPautomata} 
	\tau(\gamma)=(\s_1,t_1')\dots(\s_n,t_n'),
\end{equation}
where $t_j'=\bigotimes_{i=1}^{j}t_i$ for all $j\in\llb 1,n\rrb$. Moreover, $\tau$ is also extended to
$(\Sig\times T)^{\omega}$ recursively.

The \emph{weighted language} $\LM(\Acal^{\frakM})$ and \emph{$\omega$-weighted language} $\LM^{\omega}
(\Acal^{\frakM})$ generated by $\Acal^\frakM$ are defined by
\begin{align}
	\LM(\Acal^\frakM) &:= \{\gamma\in(\Sig\times T)^*|(\exists w\in L(\Acal^\frakM))[\ell(w)=\gamma]\}
\end{align}
and
\begin{align}
	\LM^{\omega}(\Acal^\frakM) &:= \{\gamma\in(\Sig\times T)^\omega|(\exists w\in L^\omega(\Acal^\frakM))[\ell(w)=\gamma]\},
\end{align}
respectively. Particularly, we also call $\LM(\Acal^{\Q_{\ge0}})$ and $\LM^{\omega}(\Acal^{\Q_{\ge0}})$ \emph{timed
language} and \emph{$\omega$-timed language}, respectively.

\begin{example}
	A labeled unambiguous weighted automaton $\Acal_0^{\underN}$ over semiring $\underN$ is shown in 
	Fig.~\ref{fig10_det_MPautomata}. Because $\Acal_0^{\underN}$ is unambiguous, it is the same as
	the labeled weighted automaton $\Acal_0^{\N}$ over monoid $(\N,+)$,
	\begin{figure}[!htbp]
        \centering
	\begin{tikzpicture}
	[>=stealth',shorten >=1pt,thick,auto,node distance=2.0 cm, scale = 1.0, transform shape,
	->,>=stealth,inner sep=2pt, initial text = 0]

	\tikzstyle{emptynode}=[inner sep=0,outer sep=0]

	\node[initial, state, initial where = left] (q0) {$q_0$};
	\node[state] (q2) [right of = q0] {$q_2$};
	\node[state] (q1) [above of = q2] {$q_1$};
	\node[state] (q3) [right of = q1] {$q_3$};
	\node[state] (q4) [right of = q2] {$q_{4}$};

	\path [->]
	(q0) edge node [above, sloped] {$u/10$} (q1)
	(q0) edge node [above, sloped] {$u/1$} (q2)
	(q2) edge [loop above] node [above, sloped] {$u/1$} (q2)
	(q1) edge node [above, sloped] {$a/1$} (q3)
	(q2) edge node [above, sloped] {$a/1$} (q4)
	(q3) edge [loop right] node {$a/1$} (q3)
	(q4) edge [loop right] node {$a/1$} (q4)
	;

     \end{tikzpicture}
	 \caption{Labeled unambiguous weighted automaton $\Acal_0^{\underN}$ ($\Acal_0^{\N}$) that is not divergence-free,
	where event $u$ is unobservable (we denote $\ell(u)=\ep$), event $a$ is observable (we denote $\ell(a)=a$).}
	\label{fig10_det_MPautomata}
\end{figure}
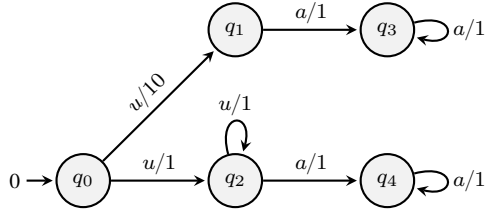
\end{example}

\begin{example}\label{exam1_MPautomata}
	Consider labeled weighted automaton $\Acal^{\N}_1$ shown in Fig.~\ref{fig6_det_MPautomata},
	where only $q_0$ is initial,
	event $u$ is unobservable, events $a$ and $b$ are observable,
	$\ell(a)=\ell(b)=\rho$. Automaton $\Acal_1^{\N}$ is ambiguous, because
	$q_3$ can be reached from $q_0$ through two paths
	$q_0\xrightarrow[]{a} q_1\xrightarrow[]{b} q_3$ and $q_0\xrightarrow[]{a} q_2
	\xrightarrow[]{b} q_3$ under the same event sequence $ab$.
	\begin{figure}[!htbp]
        \centering
	\begin{tikzpicture}
	[>=stealth',shorten >=1pt,thick,auto,node distance=2.0 cm, scale = 1.0, transform shape,
	->,>=stealth,inner sep=2pt, initial text = 0]

	\tikzstyle{emptynode}=[inner sep=0,outer sep=0]

	\node[initial, state, initial where = left] (q0) {$q_0$};
	\node[state] (q1) [above right of = q0] {$q_1$};
	\node[state] (q2) [below right of = q0] {$q_2$};
	\node[state] (q3) [above right of = q2] {$q_3$};
	\node[state] (q4) [right of = q3] {$q_{4}$};

	\path [->]
	(q0) edge node [above, sloped] {$a/1$} (q1)
	(q0) edge node [above, sloped] {$a/1$} (q2)
	(q1) edge [loop above] node [above, sloped] {$u/1$} (q1)
	(q2) edge [loop below] node [below, sloped] {$u/1$} (q2)
	(q1) edge node [above, sloped] {$b/2$} (q3)
	(q2) edge node [above, sloped] {$b/1$} (q3)
	(q3) edge node [above, sloped] {$u/1$} (q4)
	(q4) edge [loop right] node {$a/1$} (q4)
	;

        \end{tikzpicture}
		\caption{Labeled ambiguous weighted automaton $\Acal_1^{\N}$, where $\ell(u)=\ep$, $\ell(a)=\ell(b)=\rho$.}
		\label{fig6_det_MPautomata}
	\end{figure}
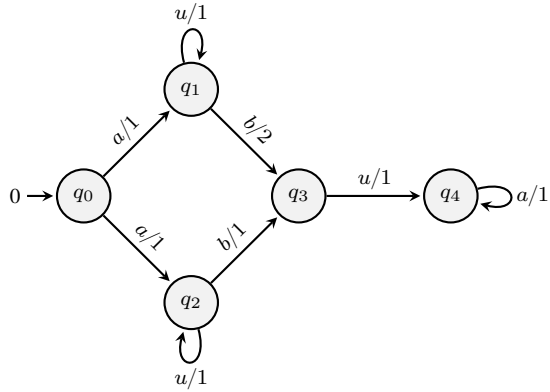

	Consider paths 
	\begin{subequations}\label{eqn12_det_MPautomata}
		\begin{align}
			\pi_1=q_0\xrightarrow[]{a} q_1\xrightarrow[]{b} q_3, & \quad	
			\pi_2=q_0\xrightarrow[]{a} q_2\xrightarrow[]{b} q_3,\label{eqn12_1_det_MPautomata}\\
			\pi_3=q_0\xrightarrow[]{a} q_1\xrightarrow[]{u} q_1\xrightarrow[]{b} q_3, & \quad
			\pi_4=q_0\xrightarrow[]{a} q_2\xrightarrow[]{u} q_2\xrightarrow[]{b} q_3,\\
			\pi_5=q_0\xrightarrow[]{a} q_2\xrightarrow[]{b} q_3\xrightarrow[]{u} q_4, &
		\end{align}
	\end{subequations}
	one then has 
	\begin{subequations}\label{eqn13_det_MPautomata}
		\begin{align}
			\tau(\pi_1)=(a,1)(b,3), & \quad \ell(\tau(\pi_1))=(\rho,1)(\rho,3),\label{eqn13_1_det_MPautomata}\\
			\tau(\pi_2)=(a,1)(b,2), & \quad \ell(\tau(\pi_2))=(\rho,1)(\rho,2),\\
			\tau(\pi_3)=(a,1)(u,2)(b,4), & \quad \ell(\tau(\pi_3))=(\rho,1)(\rho,4),\\
			\tau(\pi_4)=(a,1)(u,2)(b,3), & \quad \ell(\tau(\pi_4))=(\rho,1)(\rho,3),\\
			\tau(\pi_5)=(a,1)(b,2)(u,3), & \quad \ell(\tau(\pi_5))=(\rho,1)(\rho,2).
		\end{align}
	\end{subequations}
	Path $\pi_1$ has the following meaning: $\Acal_1^{\N}$ starts at initial state $q_0$; 
	when event $a$ occurs after time segment $1$, $\Acal_1^{\N}$ transitions to state $q_1$,
	we observe $\rho$ at time instant $1$;
	when event $b$ occurs after time segment $2$ since the occurrence of the previous event $a$,
	$\Acal_1^{\N}$ transitions to state $q_3$, and we observe $\rho$ at time instant $3$.
	The other paths have similar interpretations.
\end{example}

\section{Main results}

\subsection{The definition of current-state estimate}

For automaton $\Acal^\frakM$, for $\ep$, we define the \emph{instantaneous initial-state estimate} by
\begin{equation}\label{IISE_det_MPautomata}
	\begin{split}
	\Mt(\Acal^{\frakM},\ep):=Q_0\cup
								\{q\in Q|&(\exists q_0\in Q_0)(\exists s\in (E_{uo})^+)(\exists\pi\in q_0\xrsquigarrow{s}q)\\
								&[\tau(\pi)\in(E_{uo}\times\{ {\bf1}\})^+]\}.
	\end{split}
\end{equation}

Analogously, for a subset $x\subset Q$, we define its \emph{instantaneous-state estimate} by
\begin{equation}\label{ISE_det_MPautomata}
	\begin{split}
		\Mt(\Acal^{\frakM},\ep|x):=x\cup
								\{q\in Q|&(\exists q'\in x)(\exists s\in (E_{uo})^+)(\exists\pi\in q'\xrsquigarrow{s}q)\\
								&[\tau(\pi)\in(E_{uo}\times\{ {\bf1}\})^+]\}.
	\end{split}
\end{equation}

$\Mt(\Acal^{\frakM},\ep)$ denotes the set of states $\Acal^{\frakM}$ can be in at the initial time
before any possible observable event occurs (note that an observable event may occur at the initial time
when the event is in some instantaneous transition starting from some initial state),
so we only consider unobservable, instantaneous paths,
which is represented by $\tau(\pi)\in(E_{uo}\times\{ {\bf1}\})^+$.

More generally, for $\Acal^{\frakM}$, for weighted label/output sequence
$\gamma\in(\Sig\times T)^+$, we define the \emph{current-state estimate} as 
\begin{equation}\label{CSE_det_MPautomata}
	\begin{split}
		\Mt(\Acal^{\frakM},\gamma):=\{q\in Q|& (\exists\text{ a path }q_0\xrightarrow[]{s'}q'\xrightarrow[]{s}q)[\\
									&(q_0\in Q_0)\wedge(s'\in E^*E_o)\wedge(s\in(E_{uo})^*)\wedge\\
									&(\ell(\tau(q_0\xrightarrow[]{s'}q'))=\gamma)\wedge
									(\tau(q'\xrightarrow[]{s}q)\in(E_{uo}\times\{{\bf1}\})^*)]\}.
	\end{split}
\end{equation}

Intuitively, for $\gamma=(\s_1,t_1)\dots(\s_n,t_n)\in(\Sig\times T)^+$, $\Mt(\Acal^{\frakM},\gamma)$ 
denotes the set of states $\Acal^{\frakM}$ can be in when $\gamma$ has just been generated by $\Acal^{\frakM}$.
Particularly for $\Acal^{\Q_{\ge0}}$, $\Mt(\Acal^{\Q_{\ge0}},\gamma)$
denotes the set of states $\Acal^{\Q_{\ge0}}$ can be in when we just observe $\gamma\in(\Sig\times\Q_{\ge0})^*$.
In order to fit the setting of current-state estimate, after the occurrence of the last observable event
(denoted by $e_o$) in $s'$ (i.e., $e_o$ occurs at the current time), we only allow unobservable,
instantaneous paths, which is represented by $\tau(q'\xrightarrow[]{s}q)\in(E_{uo}\times\{{\bf1}\})^*$.

Analogously, for a subset $x\subset Q$, for a weighted label sequence $\gamma\in(\Sig\times T)^+$,
we define the current-state estimate when
automaton $\Acal^{\frakM}$ starts from some state of $x$ by
\begin{equation}\label{CSE_x_det_MPautomata}
	\begin{split}
	\Mt(\Acal^{\frakM},\gamma|x):=\{q\in Q|& (\exists\text{ a path }q_0\xrightarrow[]{s'}q'\xrightarrow[]{s}q)[\\
									&(q_0\in x)\wedge(s'\in E^*E_o)\wedge(s\in(E_{uo})^*)\wedge\\
									&(\ell(\tau(q_0\xrightarrow[]{s'}q'))=\gamma)\wedge
									(\tau(q'\xrightarrow[]{s}q)\in(E_{uo}\times\{{\bf1}\})^*)]\}.
	\end{split}
\end{equation}
Then one directly sees that $\Mt(\Acal^{\frakM},\gamma)=\Mt(\Acal^{\frakM},\gamma|Q_0)$ for all
$\gamma\in(\Sig\times T)^*$.

\begin{example}\label{exam2_MPautomata}
	Reconsider automaton $\Acal_1^{\N}$ in Fig.~\ref{fig6_det_MPautomata}. 
	By considering paths $\pi_1,\dots,\pi_5$ in \eqref{eqn12_det_MPautomata} and their timed words
	and timed label sequences in \eqref{eqn13_det_MPautomata}, we have
	\begin{subequations}\label{eqn17_det_MPautomata}
		\begin{align}
			&\Mt(\Acal_1^{\N},(\rho,1)(\rho,2))=\{q_3\},\label{eqn17_1_det_MPautomata}\\
			&\Mt(\Acal_1^{\N},(\rho,1)(\rho,3))=\{q_3\}.\label{eqn17_2_det_MPautomata}
		\end{align}
	\end{subequations}
	\eqref{eqn17_1_det_MPautomata} holds, because $\pi_2,\pi_5$ are all the paths such that their timed 
	label sequences are $(\rho,1)(\rho,2)$; in $\pi_2$, once the last observable event
	$b$ occurs, $q_3$ is reached, so $q_3$ in $\pi_2$ is consistent with timed label
	sequence (i.e., observation) $(\rho,1)(\rho,2)$; in $\pi_5$, $q_3$ is consistent with $(\rho,1)(\rho,2)$
	for the same reason, however, $q_4$ is not consistent with $(\rho,1)(\rho,2)$ because $q_4$ is reached 
	once $u$ occurs, i.e., at instant $3$. Nevertheless, if at instant $3$ we observe nothing,
	we know that event $u$ occurs and $\Acal_1^{\N}$ transitions to state $q_4$.
	Similarly, \eqref{eqn17_2_det_MPautomata} holds. $q_4$ is not consistent with $(\rho,1)(\rho,3)$
	because at instant $3$, $q_4$ can only be reached through path $\pi_5$, but at instant $3$
	one observes nothing.
\end{example}

\subsection{The definitions of detectability}

In this subsection, we formulate the four fundamental notions of detectability.

\begin{definition}[SD]\label{def_SD_MPautomata}
	A labeled weighted automaton $\Acal^{\frakM}$ \eqref{LWA_monoid_det_MPautomata}
	is called \emph{strongly detectable} if
	there is $t\in\Z_+$, for every $\omega$-{weighted} word $w\in L^{\omega}(\Acal^\frakM)$, for each
	prefix $\gamma$ of $\ell(w)$, if $|\gamma|\ge t$, then $|\Mt(\Acal^{\frakM},\gamma)|=1$.
\end{definition}

\begin{definition}[SPD]\label{def_SPD_MPautomata}
	A labeled weighted automaton $\Acal^{\frakM}$ \eqref{LWA_monoid_det_MPautomata}
	is called \emph{strongly periodically
	detectable} if there is $t\in\Z_+$, for every $\omega$-{weighted} word $w\in L^{\omega}(\Acal^\frakM)$, 
	for every prefix $w'\sqsubset w$, there is $w''\in (E\times T)^*$ such that 
	$|\ell(w'')|<t$, $w'w''\sqsubset w$, and $|\Mt(\Acal^{\frakM},\ell(w'w''))|=1$.
\end{definition}

\begin{definition}[WD]\label{def_WD_MPautomata}
	A labeled weighted automaton $\Acal^{\frakM}$ \eqref{LWA_monoid_det_MPautomata}
	is called \emph{weakly detectable} if
	$L^{\omega}(\Acal^\frakM)\ne\emptyset$ implies that
	there is $t\in\Z_+$, for some $\omega$-{weighted} word $w\in L^{\omega}(\Acal^\frakM)$, for each
	prefix $\gamma$ of $\ell(w)$, if $|\gamma|\ge t$, then $|\Mt(\Acal^{\frakM},\gamma)|=1$.
\end{definition}

\begin{definition}[WPD]\label{def_WPD_MPautomata}
	A labeled weighted automaton $\Acal^{\frakM}$ \eqref{LWA_monoid_det_MPautomata}
	is called \emph{weakly periodically
	detectable} if $L^{\omega}(\Acal^\frakM)\ne\emptyset$ implies that
	there is $t\in\Z_+$, for some $\omega$-{weighted} word $w\in L^{\omega}(\Acal^\frakM)$, for each prefix
	$w'\sqsubset w$, there is $w''\in (E\times T)^*$ such that $|\ell(w'')|<t$, $w'w''\sqsubset w$,
	and $|\Mt(\Acal^{\frakM},\ell(w'w''))|=1$.
\end{definition}

Particularly, if $\Acal^{\Q_{\ge0}}$ is strongly (resp., weakly) detectable, then there exists $t\in\N$,
along every (resp., some) $\omega$-timed word $w\in L^{\omega}(\Acal^{\Q_{\ge0}})$, if we observe at least
$t$ outputs, we can determine the corresponding current state. If $\Acal^{\Q_{\ge0}}$ is strongly (resp., weakly)
periodically detectable, then there exists $t\in\N$, along every (resp., some) $\omega$-timed word
$w\in L^{\omega}(\Acal^{\Q_{\ge0}})$, no matter how many outputs we have observed, we can determine the 
corresponding state after observing at most $t$ outputs.

Strong detectability and strong periodic detectability are incomparable.
Consider a labeled finite-state automaton $\Acal_1$ that contains two states $q_0$ and $q_1$
which are both initial, and two transitions $q_0\xrightarrow[]{u}q_0$ and $q_1\xrightarrow[]{u}q_1$
with $u$ unobservable. $\Acal_1$ is strongly detectable vacuously,
but not strongly periodically detectable by 
definition. Consider another labeled finite-state automaton $\Acal_2$ that
contains three states $q_0,q_1,q_2$ such that only $q_0$ is initial, the transitions of $\Acal_2$ are 
$q_0\xrightarrow[]{a}q_1$, $q_0\xrightarrow[]{a}q_2$, $q_1\xrightarrow[]{b}q_0$,
$q_2\xrightarrow[]{b}q_0$,
where $a$ and $b$ are observable. $\Acal_2$ is not strongly detectable but strongly periodically
detectable also by definition. Particularly, if an automaton $\Acal^{\frakM}$ is deadlock-free (i.e., for each 
reachable state $q$, there exists a transition starting at $q$) and divergence-free (i.e., there exists no
reachable unobservable cycle)\footnote{The two conditions imply that at each reachable state, there 
exists an infinitely long path whose label sequence is also of infinite length.},
then strong detectability is stronger than strong periodic detectability.
Weak detectability and weak periodic detectability also have similar relations.

Particularly for $\Acal^{\N}$, if we assume that every observable transition $q\xrightarrow[]{e/\mu(e)
_{qq'}}q'$ satisfies $\mu(e)_{qq'}>0$ (i.e., $\mu(e)_{qq'}\ge1$),
then there will be no two observable events occurring at the 
same time in one path. In this case, in Definition~\ref{def_SD_MPautomata} and Definition~\ref{def_WD_MPautomata}, 
$|\gamma|\ge k$ implies that the total time consumption is no less than $k$ when $\gamma$ has just
been generated.

\begin{example}\label{exam11_MPautomata}
	Reconsider the model in Example~\ref{exam10_MPautomata}.
	In this model, if the pairs of energy levels and positions are considered as the states,
	signals $a,u,b$ are considered 
	as the events, where $a,b$ are observable and $u$ is unobservable, and the position deviations are
	considered as the weights of transitions, then the model can be regarded as a
	labeled weighted automaton over a monoid. In detail, the states are $(i,P_j)$, $i\in\llb 0,10 \rrb$,
	$j\in\llb 1,4 \rrb$, particularly if state $(i,P_j)$ is initial, then it has an input arrow with weight $P_j$.
	The transitions are $(i,P_k)\xrightarrow[]{a/P_{k+1}-P_k}(i-1,P_{k+1})$,
	$(i,P_2)\xrightarrow[]{u/P_{3}-P_2}(i-1,P_{3})$,
	$(i,P_2)\xrightarrow[]{u/P_{3}-P_2}(i,P_{3})$, $i\in\llb 1,10 \rrb$, $k = 1,3 $,
	$(j,P_l)\xrightarrow[]{a/P_{l-1}-P_l}(j+1,P_{l-1})$,
	$(10,P_l)\xrightarrow[]{a/P_{l-1}-P_l}(10,P_{l-1})$, $j\in\llb 0,9 \rrb$, $l\in\llb 2,4 \rrb$.
	The monoid is generated by $\{P_1,\dots,P_4,P_{i+1}-P_i,P_i-P_{i+1}|i=1,2,3\}$ under $+$, where 
	$P_1,\dots,P_4$ are considered as monomials. Assume $(5,P_1)$ is the initial state, then the model
	becomes a labeled weighted
	automaton over the monoid which is weakly detectable (by, for example, infinite path $(5,P_1)(\xrightarrow
	[]{a/P_2-P_1}(4,P_2)\xrightarrow[]{b/P_1-P_2}(5,P_1))^{\omega}$ that produces $\omega$-weighted label sequence
	$((a,P_2)(b,P_1))^{\omega}$) but not strongly detectable (as shown in Table~\ref{tab3_det_MPautomata}).
	Apparently, the model cannot be represented by a labeled timed automaton.
\end{example}

\subsection{The definition of concurrent composition}

In order to give an equivalent condition for strong detectability, we define a notion of
\emph{concurrent composition} for a labeled 
weighted automaton $\Acal^{\frakM}$ and itself (i.e., the self-composition of $\Acal^{\frakM}$).
This notion can be regarded as an extension of the notion 
of self-composition $\CCa(\Acal)$ of a labeled finite-state automaton $\Acal$ proposed in
\cite{Zhang2020DetPNFA,Zhang2019KDelayStrDetDES}.
In \cite{Zhang2019KDelayStrDetDES}, $\CCa(\Acal)$ is proposed to give a
polynomial-time algorithm for verifying strong versions of detectability of $\Acal$ without
any assumption, removing two standard assumptions of deadlock-freeness and divergence-freeness
used in \cite{Shu2007Detectability_DES,Shu2011GDetectabilityDES}, etc.
In \cite{Zhang2020DetPNFA}, $\CCa(\Acal)$ and another tool called \emph{bifurcation automaton}
is used to verify a different variant of
detectability called eventual strong detectability, which is strictly weaker than strong detectability
even for labeled deterministic finite-state automata.
In $\CCa(\Acal)$, observable transitions of $\Acal$ are synchronized 
and unobservable transitions of $\Acal$ interleave. Differently, in order to define $\CCa(\Acal^{\frakM})$,
we need to consider both how to synchronize paths and how to synchronize weights of paths,
where the difficulty lies in the latter.
$\CCa(\Acal)$ can be computed in time polynomial in the size of $\Acal$ (see Table~\ref{tab2_det_MPautomata}).
However, the case for $\Acal^{\frakM}$
is much more complicated. The computability of $\CCa(\Acal^{\frakM})$ heavily depends on
$\frakM$. For example, generally $\CCa(\Acal^{\R})$
is uncomputable. 
Particularly, we  will show that $\CCa(\Acal^{\Q^k})$ is computable in time nondeterministically 
polynomial in the size of $\Acal^ {\Q^k}$ (in Section~\ref{subsubsec:computation_self_composition}),
by connecting $\Acal^{\Q^k}$ with the EPL problem (Problem~\ref{prob1_det_MPautomata}),
and generally it is unlikely that the time
consumption can be reduced, although the size of $\CCa(\Acal^{\Q^k})$ is polynomial in the size of
$\Acal^{\Q^k}$.

\begin{definition}\label{def_CC_MPautomata}
	Consider a labeled weighted automaton $\Acal^{\frakM}$ \eqref{LWA_monoid_det_MPautomata}.
	We define its \emph{self-composition} by a labeled finite-state automaton 
\begin{align}\label{CC_MPautomata}
	\CCa(\Acal^{\frakM})=(Q',E',Q_0',\dt',\Sig,\ell'),
\end{align}
where $Q'=Q\times Q$; $E'=\{(e_1,e_2)\in E_o\times E_o|\ell(e_1)=\ell(e_2)\}$; $Q_0'=Q_0\times Q_0$;
$\dt'\subset Q'\times E'\times Q'$ is the transition relation,
for all states $(q_1,q_2),(q_3,q_4)\in Q'$ and events $(e_1,e_2)\in E'$,
$((q_1,q_2),(e_1,e_2),(q_3,q_4))\in\dt'$if and only if in $\Acal^{\frakM}$, there exist states $q_5,q_6,q_7,q_8\in
Q$, event sequences $s_1,s_2,s_3,s_4\in (E_{uo})^*$, and paths
\begin{equation}\label{eqn6_det_MPautomata}
\begin{split}
	\pi_1 &:= q_1\xrightarrow[]{s_1}q_5\xrightarrow[]{e_1}q_7\xrightarrow[]{s_3}q_3,\\
	\pi_2 &:= q_2\xrightarrow[]{s_2}q_6\xrightarrow[]{e_2}q_8\xrightarrow[]{s_4}q_4,
\end{split}
\end{equation}
such that $\tau(q_1\xrightarrow[]{s_1}q_5\xrightarrow[]{e_1}q_7)=w_1(e_1,t_1)$,
$\tau(q_2\xrightarrow[]{s_2}q_6\xrightarrow[]{e_2}q_8)=w_2(e_2,t_2)$, $w_1,w_2\in(E_{uo}\times T
)^*$, $t_1=t_2\in T$, $q_7\xrightarrow[]{s_3}q_3$ and $q_8\xrightarrow[]{s_4}q_4$ are instantaneous;
for all $(e_1,e_2)\in E'$, $\ell'((e_1,e_2))=\ell(e_1)$,
and $\ell'$ is recursively extended to $(E')^*\cup(E')^{\omega}\to \Sig^*\cup\Sig^{\omega}$. For a state
$q'$ of $\CCa(\Acal^{\frakM})$, we write $q'=(q'(L),q'(R))$, where ``$L$'' and ``$R$'' denote 
``left'' and ``right'', respectively.
\end{definition}

Intuitively, there is a transition $(q_1,q_2)\xrightarrow[]{(e_1,e_2)}(q_3,q_4)$ in $\CCa(\Acal^{\Q_{\ge0}})$ 
if and only if in $\Acal^{\Q_{\ge0}}$, starting from $q_1$ and $q_2$ at the same time,
after some common time delay, $e_1$ and $e_2$ 
occur as the unique observable events, state $q_1$ and $q_2$ can transition to $q_3$ and $q_4$, 
respectively. Since we consider an observation at exactly the instant when the observable events
$e_1,e_2$ occur, we only consider unobservable, instantaneous paths after the occurrences of $e_1,e_2$
(see \eqref{eqn6_det_MPautomata}). See the following example. Whenever we draw $\CCa(\Acal^{\Q^k})$
for some given $\Acal^{\Q^k}$, we only draw reachable states and transitions.

\begin{example}\label{exam3_MPautomata}
	Reconsider labeled weighted automaton $\Acal_1^{\N}$ in Fig.~\ref{fig6_det_MPautomata}.
	Its self-composition $\CCa(\Acal_1^{\N})$ is depicted in Fig.~\ref{fig7_det_MPautomata}.
	$(q_1,q_2)\xrightarrow[]{(b,b)}(q_3,q_3)$ is a transition of $\CCa(\Acal_1^{\N})$ because 
	we have two paths $q_1\xrightarrow[]{b}q_3=\pi_1$ and
	$q_2\xrightarrow[]{u}q_2\xrightarrow[]{b}q_3=\pi_2$ such that $\tau(\pi_1)=(b,2)=\tau(\pi_2)$.

	\begin{figure}[!htbp]
        \centering
\begin{tikzpicture}
	[>=stealth',shorten >=1pt,thick,auto,node distance=2.0 cm, scale = 1.0, transform shape,
	->,>=stealth,inner sep=2pt]

	\tikzstyle{emptynode}=[inner sep=0,outer sep=0]

	\node[initial, state, initial where = above] (00) {$q_0q_0$};
	\node[state] (11) [above left of = 00] {$q_1q_1$};
	\node[state] (12) [above right of = 00] {$q_1q_2$};
	\node[state] (21) [below left of = 00] {$q_2q_1$};
	\node[state] (22) [below right of = 00] {$q_2q_2$};
	\node[state] (33) [below right of = 12] {$q_3q_3$};
	\node[state] (44) [right of = 33] {$q_4q_4$};

	\path [->]
	(00) edge node [above, sloped] {$(a,a)$} (11)
	(00) edge node [above, sloped] {$(a,a)$} (12)
	(00) edge node [above, sloped] {$(a,a)$} (21)
	(00) edge node [above, sloped] {$(a,a)$} (22)
	(12) edge node [above, sloped] {$(b,b)$} (33)
	(22) edge node [above, sloped] {$(b,b)$} (33)
	(33) edge node [above, sloped] {$(a,a)$} (44)
	(44) edge [loop right] node {$(a,a)$} (44)
	;
	
	\node at (2.1,2.3) {$(b,b)$};
	\node at (2.1,-2.3) {$(b,b)$};

	\draw 
	(11) .. controls (2.1,2.5) .. (33) 
	;
	\draw
	(21) .. controls (2.1,-2.5) .. (33)
	;

        \end{tikzpicture}
		\caption{Self-composition $\CCa(\Acal_1^{\N})$ of automaton $\Acal_1^{\N}$
		in Fig.~\ref{fig6_det_MPautomata}.}
		\label{fig7_det_MPautomata}
	\end{figure}
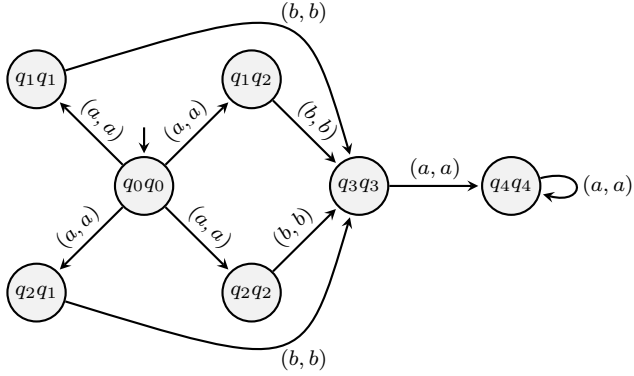
\end{example}

\subsection{The definition of observer}\label{subsec:observer}

We next define a notion of \emph{observer} to concatenate current-state estimates along weighted
label sequences. Later, we will use the notion of observer to give equivalent conditions for weak detectability
and weak periodic 
detectability of labeled weighted automaton $\Acal^{\frakM}$. An observer $\Acal_{obs}^{\frakM}$
of $\Acal^{\frakM}$ is a natural but nontrivial
extension of the observer $\Acal_{obs}$ of labeled finite-state automaton $\Acal$ proposed in 
\cite{Shu2007Detectability_DES}. 
Since in automaton $\Acal$, no weights need be considered, its observer $\Acal_{obs}$ can be computed
by directly concatenating the current-state estimates along label sequences,
so $\Acal_{obs}$ can be computed in exponential time. However, 
$\Acal_{obs}^{\frakM}$ is much more complicated, because when we concatenate
current-state estimates along weighted label sequences, we must additionally consider how to synchronize
weights. In order to define observer $\Acal^{\frakM}_{obs}$ that is a finite automaton,
we need to define a more general notion of \emph{pre-observer}
$\tensor*[^{pre}]{\Acal}{^{\frakM}_{obs}}$ that can be regarded a deterministic automaton, in which
there may exist infinitely mainly events, because the events of $\tensor*[^{pre}]{\Acal}{_{obs}^{\frakM}}$ 
are pairs of events of $\Acal^{\frakM}$ and weights chosen from
$\frakM$, and $\frakM$ may be of infinite cardinality.
$\Acal^{\frakM}_{obs}$ is a reduced version of
$\tensor*[^{pre}]{\Acal}{^{\frakM}_{obs}}$. The computability of ${\Acal}{_{obs}^{\frakM}}$
depends on $\frakM$. Generally, ${\Acal}{_{obs}^{\R}}$
is uncomputable. Despite of this difficulty, particularly for automaton
$\Acal^{\Q^k}$, we will prove that ${\Acal}{_{obs}^{\Q^k}}$
can be computed in $2$-$\EXPTIME$ in the size of $\Acal^{\Q^k}$ (Section~\ref{subsubsec:observer}),
which shows an essential difference between labeled finite-state automata and labeled
weighted automata over monoids.


\begin{definition}\label{def_pre_observer_MPautomata}
	For labeled weighted automaton $\Acal^{\frakM}$, we define its \emph{pre-observer}
as a deterministic automaton 
\begin{align}\label{eqn_pre_observer_MPautomata}
	\tensor*[^{pre}]{\Acal}{_{obs}^{\frakM}}=(X,\Sig\times T,x_0,\bar\dt_{obs}),
\end{align}
where $X\subset 2^Q\setminus\{\emptyset\}$ is the state set, $\Sig\times T$ the alphabet, 
$x_0=\Mt(\Acal^{\frakM},\ep) \in X$ the unique initial state, $\bar\dt_{obs}\subset X\times (\Sig\times T)
\times X$ the transition relation.
Note that $\Sig\times T$ may be infinite. For all $x\subset Q$ different from $x_0$, $x\in X$ if and only
if there is $\gamma\in(\Sig\times T)^{+}$ such that $x=\Mt(\Acal^{\frakM},\gamma)$. For all $x,x'\in
X$ and $(\s,t)\in \Sig\times T$, $(x,(\s,t),x')\in\bar\dt_{obs}$ if and only if $x'=\Mt(\Acal^{\frakM},
(\s,t)|x)$ (defined in \eqref{CSE_x_det_MPautomata}).
\end{definition}


In Definition~\ref{def_pre_observer_MPautomata}, after $\bar\dt_{obs}$ is recursively
extended to $\bar\dt_{obs}\subset X\times (\Sig\times T)^*\times X$ as usual, one has for all $x\in X$ and 
$(\s_1,t_1)\dots(\s_n,t_n)=:\gamma\in (\Sig\times T)^{+}$, $(x_0,\gamma,x)\in\bar\dt_{obs}$ if and only if
$\Mt(\Acal^{\frakM},\tau(\gamma))=x$, where $\tau(\gamma)$ is defined in \eqref{eqn11_det_MPautomata}.

On the other hand, the alphabet $\Sig\times T$ may not be finite, so generally we cannot compute
the whole $\tensor*[^{pre}]{\Acal}{^{\frakM}_{obs}}$. However, in order to study
weak detectability and weak periodic detectability, it is enough to consider a subautomaton of
$\tensor*[^{pre}]{\Acal}{^{\frakM}_{obs}}$ with finitely many events.
\begin{definition}\label{def_observer_MPautomata}
	For labeled weighted automaton $\Acal^{\frakM}$, we define its \emph{observers}
	as deterministic finite automata
	\begin{equation}\label{eqn_observer_sim_MPautomata}
		{\Acal}{^{\frakM}_{obs}}=(X,\Sig_{obs}^T,x_0,\dt_{obs}),
	\end{equation}
	which are subautomata of its pre-observer $\tensor*[^{pre}]{\Acal}{^{\frakM}_{obs}}$,
	where $\Sig_{obs}^T$ is a finite subset of $\Sig\times T$, $\dt_{obs}\subset X\times\Sig_{obs}^T\times X$
	is such that if there is a transition
	from $x\in X$ to $x'\in X$ in $\bar\dt_{obs}$ then at least one transition from $x$ to $x'$ in 
	$\bar\dt_{obs}$ belongs to $\dt_{obs}$.
\end{definition}

Note that a given automaton $\Acal^{\frakM}$ may have more than one observer.

\begin{example}\label{exam4_MPautomata}
	Reconsider automaton $\Acal_1^{\N}$ in Fig.~\ref{fig6_det_MPautomata}.
	Its pre-observer $\tensor*[^{pre}]{\Acal}{^{\N}_{1obs}}$ is shown in 
	Fig.~\ref{fig8_det_MPautomata}. From the pre-observer, one sees that
	for all $n\in\Z_{+}$, $(\{q_1,q_2\},(\rho,n),\{q_3\})$ are transitions. Hence there exist infinitely 
	many transitions. However, there exist finitely many states. In order to obtain one of its 
	observers, one only need replace $(\rho,\Z_+)$ by $(\rho,1)$. 

	\begin{figure}[!htbp]
        \centering
\begin{tikzpicture}
	[>=stealth',shorten >=1pt,thick,auto,node distance=2.0 cm, scale = 1.0, transform shape,
	->,>=stealth,inner sep=2pt]

	\tikzstyle{emptynode}=[inner sep=0,outer sep=0]

	\node[initial, state, initial where = left] (0) {$q_0$};
	\node[state] (12) [right of = 0] {$q_1,q_2$};
	\node[state] (3) [right of = 12] {$q_3$};
	\node[state] (4) [right of = 3] {$q_4$};

	\path [->]
	(0) edge node [above, sloped] {$(\rho,1)$} (12)
	(12) edge node [above, sloped] {$(\rho,\Z_+)$} (3)
	(3) edge node [above, sloped] {$(\rho,2)$} (4)
	(4) edge [loop right] node {$(\rho,1)$} (4)
	;
	
	\end{tikzpicture}
	\caption{Pre-observer $\tensor*[^{pre}]{\Acal}{^{\N}_{1obs}}$ of automaton $\Acal_1^{\N}$ in 
	Fig.~\ref{fig6_det_MPautomata}, where $(\rho,\Z_+)$ means that the events can be $(\rho,n)$ for any $n\in\Z_+$.}
	\label{fig8_det_MPautomata}
	\end{figure}
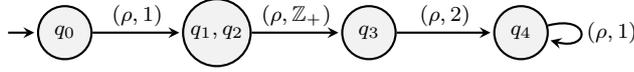
\end{example}

\subsection{The definition of detector}

In order to give an equivalent condition for strong periodic detectability,
we define a notion of \emph{detector} $\Acal^{\frakM}_{det}$ (a nondeterministic finite automaton)
for labeled weighted automaton $\Acal^{\frakM}$,
which can be regarded as a simplified version of observer $\Acal_{obs}^{\frakM}$
\eqref{eqn_observer_sim_MPautomata}. Detector
$\Acal^{\frakM}_{det}$ can be regarded as
a nontrivial extension of the detector $\Acal_{det}$ of labeled finite-state automaton
$\Acal$ proposed in \cite{Shu2011GDetectabilityDES}. In order to define $\Acal_{det}
^{\frakM}$, we must additionally consider how to synchronize weights of paths. 
Moreover, we also need to define a more general notion of
\emph{pre-detector} $\tensor*[^{pre}]
{\Acal}{^{\frakM}_{det}}$ (similar to pre-observer $\tensor*[^{pre}]{\Acal}{^{\frakM}_{obs}}$), in which 
there may exist infinitely many events.
The detector $\Acal_{det}$ of $\Acal$ can be computed in time polynomial
in the size of $\Acal$. 
However, 
the computability of detector ${\Acal}{^{\frakM}_{det}}$ still depends on 
$\frakM$. Particularly we will prove that $\Acal_{det}^{\Q^k}$ can be computed in 
$2$-$\EXPTIME$ in the size of $\Acal^{\Q^k}$ (Section~\ref{subsubsec:detector}), although the size
of $\Acal_{det}^{\Q^k}$ is polynomial in the size of $\Acal^{\Q^k}$, and 
for $\Acal^{\Q^k}$ in which from each state, a distinct state can be
reached through some unobservable, instantaneous path, detector $\Acal^{\Q^k}_{det}$ can be computed in 
$\NP$.

\begin{definition}\label{def_pre_detector_MPautomata}
	For labeled weighted automaton $\Acal^{\frakM}$, we define its \emph{pre-detector}
as an automaton 
\begin{align}\label{eqn_pre_detector_MPautomata}
	\tensor*[^{pre}]{\Acal}{_{det}^{\frakM}}=(X,\Sig\times T,x_0,\bar\dt_{det}),
\end{align}
where $X=\{x_0\}\cup\{x\subset Q|1\le|x|\le 2\}$ is the state set, $\Sig\times T$ the alphabet,
$x_0=\Mt(\Acal^{\frakM},\ep)$ the unique initial state, 
$\bar\dt_{det}\subset X\times (\Sig\times T)\times X$ the transition relation.
For all $x\in X$ and $(\s,t)\in\Sig\times T$, $(x,(\s,t),\Mt(\Acal^{\frakM},(\s,t)|x))\in\bar\dt_{det}$
if $|\Mt(\Acal^{\frakM},(\s,t)|x)|=1$; $(x,(\s,t),x')\in\bar\dt_{det}$ for all $x'\subset \Mt(\Acal^{\frakM},(\s,t)|x)$
satisfying $|x'|=2$ if $|\Mt(\Acal^{\frakM},(\s,t)|x)|\ge 2$, where $\Mt(\Acal^{\frakM},(\s,t)|x)$ is
the current-state estimate when $\Acal^{\frakM}$ starts from some state of $x$ (defined
in \eqref{CSE_x_det_MPautomata}).
\end{definition}

Similarly to verifying weak (periodic) detectability, 
in order to verify strong periodic detectability, it is enough to consider a finite subautomaton of 
pre-detector $\tensor*[^{pre}]{\Acal}{_{det}^{\frakM}}$
with finitely many events.

\begin{definition}\label{def_detector_MPautomata}
	For labeled weighted automaton $\Acal^{\frakM}$, we define its \emph{detectors}
	as nondeterministic finite automata
	\begin{equation}\label{eqn_detector_sim_MPautomata}
		{\Acal}{^{\frakM}_{det}}=(X,\Sig_{det}^T,x_0,\dt_{det}),
	\end{equation} 
	which are subautomata of its pre-detector $\tensor*[^{pre}]{\Acal}{_{det}^{\frakM}}$,
	where $\Sig_{det}^T$ is a finite subset of $\Sig\times T$, $\dt_{det}\subset X\times\Sig_{det}^T\times X$ is
	such that if there exists a transition
	from $x\in X$ to $x'\in X$ in $\bar\dt_{det}$ then at least one transition from $x$ to $x'$ in 
	$\bar \dt_{det}$ belongs to $\dt_{det}$.
\end{definition}

\begin{example}\label{exam5_MPautomata}
	Reconsider automaton $\Acal_1^{\N}$ in Fig.~\ref{fig6_det_MPautomata}.
	Its pre-detector $\tensor*[^{pre}]{\Acal}{^{\N}_{1det}}$ is also shown in Fig.~\ref{fig8_det_MPautomata}.
	That is,
	its pre-observer is the same as its pre-detector. In order to obtain one of its detectors,
	one also only need replace $(b,\Z_+)$ by $(b,1)$. 
\end{example}


For the relationship between observer ${\Acal}{^{\frakM}_{obs}}$ and detector
${\Acal}{^{\frakM}_{det}}$, we have the following lemma.

\begin{lemma}\label{lem3_det_MPautomata}
	Consider a labeled weighted automaton $\Acal^{\frakM}$, 
	any observer ${\Acal}{^{\frakM}_{obs}}$ \eqref{eqn_observer_sim_MPautomata},
	and any detector ${\Acal}{^{\frakM}_{det}}$ \eqref{eqn_detector_sim_MPautomata}, such that $\Acal^{\frakM}_{obs}$
	and $\Acal^{\frakM}_{det}$ have the same event set.
	For every transition $(x,(\s,t),x')\in\dt_{obs}$,
	for every $\bar x'\subset x'$ satisfying $|\bar x'|=2$ if $|x'|\ge 2$ and $|\bar x'|=1$
	otherwise, there is
	$\bar x\subset x$ such that (1) $|\bar x|=2$ and $(\bar x,(\s,t),\bar x')\in\dt_{det}$
	if $|x|>1$ and (2) $|\bar x|=1$ and $(\bar x,(\s,t),\bar x')\in\dt_{det}$
	if $|x|=1$.
\end{lemma}

\begin{proof}
	We only need to prove the case $|x|\ge 2$ and $|x'|\ge 2$, the other cases hold similarly.
	Arbitrarily choose $\{q_1,q_2\}=\bar x'\subset x'$ such that $q_1\ne q_2$. By definition,
	there exist $q_3,q_3',q_4,q_5\in Q$, $e_1,e_2\in E_o$, $s_1,s_2,s_3,s_4\in (E_{uo})^*$,
	and paths
	\begin{align*}
		&q_3\xrightarrow[]{s_1e_1}q_4\xrightarrow[]{s_3}q_1,\\
		&q_3'\xrightarrow[]{s_2e_2}q_5\xrightarrow[]{s_4}q_2
	\end{align*}
	such that $\ell(e_1)=\ell(e_2)=\s$, the weights of paths $q_3\xrightarrow[]{s_1e_1}q_4$
	and $q_3'\xrightarrow[]{s_2e_2}q_5$ are both equal to $t$, and paths $q_4\xrightarrow[]{s_3}q_1$ and
	$q_5\xrightarrow[]{s_4}q_2$ are unobservable and instantaneous.
	If $q_3=q_3'$, we choose $\bar x=\{q_3,q_6\}$, where $q_6\in x\setminus\{q_3\}$;
	otherwise, we choose $\bar x=\{q_3,q_3'\}$. Then by definition, 
	one has $(\bar x,(\s,t),\bar x')\in\dt_{det}$.
\end{proof}

\subsection{Equivalent conditions for detectability of labeled weighted automata over monoids}

In this subsection, we give equivalent conditions for the four notions of detectability of 
labeled weighted automata over monoids by using the notions of self-composition, observer, and detector.

\subsubsection{For strong detectability:}

We use the notion of self-composition to give an equivalent condition for strong detectability
of labeled weighted automata over monoids.

\begin{theorem}\label{thm1_det_MPautomata}
	A labeled weighted automaton $\Acal^{\frakM}$ \eqref{LWA_monoid_det_MPautomata}
	is not strongly detectable if and only if in its self-composition $\CCa(\Acal^{\frakM})$
	\eqref{CC_MPautomata},
	\begin{enumerate}[(i)]
		\item there exists a transition sequence
			\begin{align}\label{eqn2_3_det_MPautomata}
				q_0'\xrightarrow[]{s_1'}q_1'\xrightarrow[]{s_2'}q_1'\xrightarrow[]{s_3'}q_2'
			\end{align}
			satisfying
			\begin{align}\label{eqn2_1_det_MPautomata}
				q_0'\in Q_0'; q_1',q_2'\in Q'; s_1',s_2',s_3'\in(E')^+; q_2'(L)\ne q_2'(R);
			\end{align}
		\item\label{eqn2_2_det_MPautomata}
			and in $\Acal^{\frakM}$, there exists a cycle reachable from $q_2'(L)$.
	\end{enumerate}
\end{theorem}

\begin{proof}
	By Definition~\ref{def_SD_MPautomata}, $\Acal^{\frakM}$ is not strongly detectable if and only if 
	for all $k\in\Z_+$, there exist $w_k\in L^{\omega}(\Acal^\frakM)$ and $\gamma\sqsubset
	\ell(w_k)$, such that $|\gamma|\ge k$ and $|\Mt(\Acal^{\frakM},\gamma)|>1$.

	``if'': Arbitrarily given $k\in\Z_{+}$, consider $q_0'\xrightarrow[]{s_1'}q_1'\xrightarrow[]{(s_2')^k}q_1'
	\xrightarrow[]{s_3'}q_2'$, then by \eqref{eqn2_1_det_MPautomata},
	in $\Acal^{\frakM}$ there exists a path $q_0'(L)\xrightarrow[]{\bar s_1}q_1'(L)
	\xrightarrow[]{\bar s_2}q_1'(L)\xrightarrow[]{\bar s_3}q_2'(L)=:\pi_L$ such that $\ell(\bar s_1)=\ell'(s_1')$,
	$\ell(\bar s_2)=\ell'((s_2')^k)$, $\ell(\bar s_3)=\ell'(s_3')$, and $\Mt(\Acal^{\frakM},\gamma)\supset
	\{q_2'(L),q_2'(R)\}$, where $\gamma=\ell(\tau(\pi_L))$; by \eqref{eqn2_2_det_MPautomata},
	there also exists a path $q_2'(L)\xrightarrow[]
	{\bar s_4}q_3\xrightarrow[]{\bar s_5}q_3$, where $\bar s_5\in E^+$.
	Note that $q_3\xrightarrow[]{\bar s_5}q_3$ can be repeated for 
	infinitely many times. Choose $$w_k=\tau\left(q_0'(L)\xrightarrow[]{\bar s_1}q_1'(L)
	\xrightarrow[]{\bar s_2}q_1'(L)\xrightarrow[]{\bar s_3}q_2'(L)\xrightarrow[]{\bar s_4}q_3\left(\xrightarrow[]
	{\bar s_5}q_3\right)^{\omega}\right),$$ one has $w_k\in L^{\omega}(\Acal^\frakM)$, $\gamma\sqsubset
	\ell(w_k)$ satisfies $|\gamma|\ge k+2$, and $|\Mt(\Acal^{\frakM},\gamma)|>1$. That is, $\Acal^{\frakM}$ is not strongly
	detectable.

	``only if'': Assume that $\Acal^{\frakM}$ is not strongly detectable. Choose $k>|Q|^2$, $w_k\in 
	L^{\omega}(\Acal^\frakM)$, and $\gamma\sqsubset \ell(w_k)$ such that $|\gamma|\ge k$ and $|\Mt(\Acal
	^{\frakM},\gamma)|>1$. Then there exist two different paths $\pi_1$ and $\pi_2$ starting at initial states and
	ending at different states such that $\tau(\pi_1)=\tau(\pi_2)\sqsubset w_k$, and after the last observable
	events of $\pi_1$ and $\pi_2$, all transitions are unobservable and instantaneous. By definition of $\CCa(\Acal
	^{\frakM})$, from $\pi_1$ and $\pi_2$ one can construct a transition sequence of $\CCa(\Acal^{\frakM})$
	as in \eqref{eqn2_3_det_MPautomata} by the Pigeonhole Principle, because $\CCa(\Acal^{\frakM})$ has at most
	$|Q|^2$ states. On the other hand, because $\Acal^{\frakM}$ has finitely many states, \eqref{eqn2_2_det_MPautomata}
	holds.
\end{proof}

\subsubsection{For strong periodic detectability:}

We first use the notion of observer to give an equivalent condition for strong periodic detectability
of labeled weighted automata, and furthermore represent the equivalent condition in terms of the notion of detector.

\begin{theorem}\label{thm8_det_MPautomata}
	A labeled weighted automaton $\Acal^{\frakM}$ \eqref{LWA_monoid_det_MPautomata} 
	is not strongly periodically detectable if and only if in any observer
	${\Acal}{^{\frakM}_{obs}}$
	\eqref{eqn_observer_sim_MPautomata}, at least one of the two following conditions holds.
	\begin{enumerate}[(i)]
		\item\label{item14_det_MPautomata}
			There is a reachable state $x\in X$ such that $|x|>1$
			and there exists a path $q\xrightarrow[]{s_1}q'\xrightarrow[]{s_2}q'$ in $\Acal^{\frakM}$,
			where $q\in x$, $s_1\in (E_{uo})^*$, $s_2\in(E_{uo})^+$, $q'\in Q$.
		\item\label{item15_det_MPautomata}
			There is a reachable cycle in ${\Acal}{^{\frakM}_{obs}}$
			such that no state in the cycle is a singleton.
	\end{enumerate}
\end{theorem}

\begin{proof}
	By Definition~\ref{def_SPD_MPautomata}, $\Acal^{\frakM}$ is not strongly periodically
	detectable if and only for all $k\in\Z_+$, there is an $\omega$-{weighted} word $w_k\in L^{\omega}(\Acal^\frakM)$
	and a prefix $w'\sqsubset w_k$ such that for all $w''\in (E\times T)^*$ satisfying
	$|\ell(w'')|<k$ and $w'w''\sqsubset w_k$, one has $|\Mt(\Acal^{\frakM},\ell(w'w''))|>1$.

	``if'': Assume \eqref{item14_det_MPautomata} holds. Then there exists a path
	$q_0\xrightarrow[]{s_{\gamma}}q\xrightarrow[]{s_1}q'\xrightarrow[]{s_2}q'$ in $\Acal^{\frakM}$
	such that $q_0\in Q_0$ and $\Mt(\Acal^{\frakM},\ell(\tau(q_0\xrightarrow[]{s_{\gamma}}q)))=x$.
	Denote $\tau(q_0\xrightarrow[]{s_{\gamma}}q)=:w_1\in L(\Acal^\frakM)$ and $\tau(q_0\xrightarrow[]{s_{\gamma}}q
	\xrightarrow[]{s_1}q'(\xrightarrow[]{s_2}q')^{\omega})=:w_1w_2\in L^{\omega}(\Acal^\frakM)$, then 
	for every $w\sqsubset w_2$, one has $\ell(w)=\ep$ and $|\Mt(\Acal^{\frakM},\ell(w_1w
	))|=|\Mt(\Acal^{\frakM},\ell(w_1))|>1$, which violates strong periodic detectability by definition.

	Assume \eqref{item15_det_MPautomata} holds. That is, there exist $\alpha\in(\Sig\times T)^*$, $\beta\in
	(\Sig\times T)^+$ such that $(x_0,\alpha,x),(x,\beta,x)\in\dt_{obs}$ for some $x\in X$ satisfying
	$|x|>1$, $\Mt(\Acal^{\frakM},\tau(\alpha))=\Mt(\Acal^{\frakM},\tau(\alpha\beta))=x$, and for all $\beta'\sqsubset\beta$,
	$|\Mt(\Acal^{\frakM},\tau(\alpha\beta'))|>1$. Then $\tau(\alpha\beta^{\omega})\in\LM^{\omega}(\Acal^{\frakM})$.
	Choose $w_{\alpha}w_{\beta}\in L^{\omega}(\Acal^\frakM)$
	such that $\ell(w_\alpha)=\tau(\alpha)$ and $\ell(w_{\alpha}w_\beta)=\tau(\alpha\beta^{\omega})$,
	then for every $w_{\beta}'
	\sqsubset w_{\beta}$, one has $|\Mt(\Acal^{\frakM},\ell(w_\alpha w_{\beta}'))|>1$,
	which also violates strong periodic detectability.

	``only if'':  Assume $\Acal^{\frakM}$ is not strongly periodically detectable and
	\eqref{item15_det_MPautomata} does not hold, next we prove \eqref{item14_det_MPautomata}
	holds.

	Since $\Acal^{\frakM}$ is not strongly periodically detectable, by definition,
	choose integer $k>|2^Q|$, $w_k\in L^{\omega}({\Acal^\frakM})$, and
	prefix $w'\sqsubset w_k$ such that for all $w''\in (\Sig\times T)^*$, $w'w''\sqsubset w_k$
	and $|\ell(w''))|<k$ imply $|\Mt(\Acal^{\frakM},\ell(w'w''))|>1$. Since \eqref{item15_det_MPautomata}
	does not hold, one has $\ell(w_k)\in(\Sig\times T)^*$ and $|\ell(w_k)|<k+|\ell(w')|$.
	Otherwise if $|\ell(w_k)|\ge k+|\ell(w')|$ or $\ell(w_k)\in(\Sig\times T)^{\omega}$,
	we can choose $\bar w''$ such that $w'\bar w''\sqsubset w_k$ and $|\ell(\bar w'')|=
	k$, then by the Pigeonhole Principle, there exist $\bar w_1'',\bar w_2''\sqsubset
	\bar w''$ such that $|\ell(\bar w_1'')|<|\ell(\bar w_2'')|$ and $\Mt(\Acal^{\frakM},\ell(w'\bar
	w_1''))=\Mt(\Acal^{\frakM},\ell(w'\bar w_2''))$, that is, there is a reachable cycle in pre-observer 
	$\tensor*[^{pre}]{\Acal}{_{obs}^{\frakM}}$ such that no state in the cycle is a singleton,
	then by definition, \eqref{item15_det_MPautomata} holds.
	Then $w_k=w'\hat w_1''\hat w_2''$, where $\hat w_1''\in (E\times T)^*$, $\hat w_2''\in (E_{uo}\times T)^
	{\omega}$. Moreover, one has $|\Mt(\Acal^{\frakM},\ell(w'
	\hat w_1''))|>1$, and also by the Pigeonhole Principle there exists a path $q_0\xrightarrow[]{w'
	\hat w_1''}q\xrightarrow[]{\tilde w_1''}q'\xrightarrow[]{\tilde w_2''}q'$ for some
	$q_0\in Q_0$, $q,q'\in Q$, $\tilde w_1''\in (E_{uo})^*$, and $\tilde w_2''\in
	(E_{uo})^+$, i.e., $\Mt(\Acal^{\frakM},\ell(w'\hat w_1''))$ is a reachable state of pre-observer 
	$\tensor*[^{pre}]{\Acal}{_{obs}^{\frakM}}$ such that $|\Mt(\Acal^{\frakM},\ell(w'\hat w_1''))|>1$ and $q\in
	\Mt(\Acal^{\frakM},\ell(w'\hat w_1''))$. By definition, $\Mt(\Acal^{\frakM},\ell(w'\hat w_1''))$ is also 
	a reachable state of observer $\Acal_{obs}^{\frakM}$, then \eqref{item14_det_MPautomata} holds.
\end{proof}

\begin{theorem}\label{thm9_det_MPautomata}
	A labeled weighted automaton $\Acal^{\frakM}$ \eqref{LWA_monoid_det_MPautomata} 
	is
	not strongly periodically detectable if and only if in any detector
	${\Acal}{_{det}^{\frakM}}$
	\eqref{eqn_detector_sim_MPautomata}, at least one of the two following conditions holds.
	\begin{enumerate}[(1)]
		\item\label{item16_det_MPautomata}
			There is a reachable state $x'\in X$ such that $|x'|>1$
			and there exists a path $q\xrightarrow[]{s_1}q'\xrightarrow[]{s_2}q'$ in $\Acal^{\frakM}$, where
			$q\in x'$, $s_1\in (E_{uo})^*$, $s_2\in(E_{uo})^+$, $q'\in Q$.
		\item\label{item17_det_MPautomata}
			There is a reachable cycle in ${\Acal}{_{det}^{\frakM}}$ such that
			all states in the cycle have cardinality $2$.
	\end{enumerate}
\end{theorem}

\begin{proof}
	We use Theorem~\ref{thm8_det_MPautomata} to prove this result.

	We firstly prove \eqref{item16_det_MPautomata} of this theorem is equivalent to \eqref{item14_det_MPautomata}
	of Theorem~\ref{thm8_det_MPautomata}.

	``\eqref{item16_det_MPautomata}$\Rightarrow$\eqref{item14_det_MPautomata}'':
	Assume \eqref{item16_det_MPautomata} holds.
	In ${\Acal}{_{det}^{\frakM}}$, choose a transition sequence $x_0\xrightarrow[]
	{\alpha}x'$. Then one has $x'\subset x$, where $(x_0,\alpha,x)\in\dt_{obs}$,
	hence \eqref{item14_det_MPautomata} of Theorem~\ref{thm8_det_MPautomata} holds.

	``\eqref{item16_det_MPautomata}$\Leftarrow$\eqref{item14_det_MPautomata}'':
	Assume \eqref{item14_det_MPautomata} holds.
	In ${\Acal}{_{obs}^{\frakM}}$, choose a transition sequence $x_0\xrightarrow[]
	{\alpha}x$. By Lemma~\ref{lem3_det_MPautomata}, moving backward on $x_0\xrightarrow[]
	{\alpha}x$ from $x$ to $x_0$, we can obtain a transition sequence
	$x_0\xrightarrow[]{\alpha}x'$ of ${\Acal}{_{det}^{\frakM}}$ such that $q\in x'\subset x$,
	hence \eqref{item16_det_MPautomata} of this theorem holds.

	We secondly prove \eqref{item17_det_MPautomata} of this theorem is equivalent to \eqref{item15_det_MPautomata}
	of Theorem~\ref{thm8_det_MPautomata}.

	``\eqref{item17_det_MPautomata}$\Rightarrow$\eqref{item15_det_MPautomata}'':
	Assume \eqref{item17_det_MPautomata} holds.
	In ${\Acal}{_{det}^{\frakM}}$, choose a transition sequence $x_0\xrightarrow[]{\alpha}
	x\xrightarrow[]{\beta}x$ such that in $x\xrightarrow[]{\beta}x$
	all states are of cardinality $2$ and $|\beta|>0$. Without loss of generality, we assume 
	$|\beta|>|2^{Q}|$, because otherwise we can repeat $x\xrightarrow[]{\beta}x$ for
	$|2^{Q}|+1$ times. By definition, one has for all
	$\beta'\sqsubset\beta$, for the $x_{\beta'}$ satisfying $(x_0,\alpha\beta',x_{\beta'})\in\dt_{obs}$,
	$|x_{\beta'}|>1$. Then by the Pigeonhole Principle,
	there exist $\beta_1,\beta_2\sqsubset\beta$ such that $|\beta_1|<|\beta_2|$
	and $x_{\beta_1}=x_{\beta_2}$, where $(x_0,\alpha\beta_1,x_{\beta_1}),(x_0,\alpha\beta_2,x_{\beta_2})
	\in\dt_{obs}$. Thus, \eqref{item15_det_MPautomata} of Theorem~\ref{thm8_det_MPautomata} holds.

	``\eqref{item17_det_MPautomata}$\Leftarrow$\eqref{item15_det_MPautomata}'':
	Assume \eqref{item15_det_MPautomata} holds.
	In ${\Acal}{_{obs}^{\frakM}}$, choose a transition sequence $x_0\xrightarrow[]{\alpha}
	x_1\xrightarrow[]{\beta_1}\cdots\xrightarrow[]{\beta_n}x_{n+1}$ such that $n\ge|Q|^2$, 
	$x_1=x_{n+1}$, $|x_1|,\dots,|x_{n+1}|>1$, and $\beta_1,\dots,\beta_n\in\Sig\times T$.
	By using Lemma~\ref{lem3_det_MPautomata} from $n+1$ to $2$,
	we obtain $x_i'\subset x_i$ for all $i\in\llb 1,n+1\rrb$
	such that $|x_1'|=\cdots=|x_{n+1}'|=2$ and a transition sequence
	$x_1'\xrightarrow[]{\beta_1}\cdots\xrightarrow[]{\beta_n}x_{n+1}'$ of ${\Acal}{_{det}^{\frakM}}$.
	Moreover, also by Lemma~\ref{lem3_det_MPautomata}, we obtain a transition sequence 
	$x_0\xrightarrow[]{\alpha}x_1'$ of ${\Acal}{_{det}^{\frakM}}$. By the Pigeonhole Principle,
	\eqref{item17_det_MPautomata} of this theorem holds.
\end{proof}

\subsubsection{For weak detectability and weak periodic detectability:}

We use the notion of observer to give equivalent conditions for weak detectability
and weak periodic detectability of labeled weighted automata.

\begin{theorem}\label{thm2_det_MPautomata}
	A labeled weighted automaton $\Acal^{\frakM}$ \eqref{LWA_monoid_det_MPautomata} 
	is 
	weakly detectable if and only if either one of the following three conditions holds.
	\begin{enumerate}[(i)]
		\item\label{item1_det_MPautomata}
			$L^{\omega}(\Acal^\frakM)=\emptyset$.
		\item\label{item2_det_MPautomata}
			$L^{\omega}(\Acal^\frakM)\ne\emptyset$ and
			there exists $w\in L^{\omega}(\Acal^\frakM)$ such that $\ell(w)\in (\Sig\times T)^*$.
		\item\label{item3_det_MPautomata}
			$L^{\omega}(\Acal^\frakM)\ne\emptyset$ and in any one of its observers, 
			there is a reachable cycle in which all states are singletons.
	\end{enumerate}
\end{theorem}

\begin{proof}
	``if'': \eqref{item1_det_MPautomata} implies that $\Acal^{\frakM}$ is weakly detectable vacuously.

	Assume \eqref{item2_det_MPautomata} holds. Choose integer $k>|\ell(w)|$, then one has $\Acal^{\frakM}$ is weakly
	detectable vacuously.

	Assume \eqref{item3_det_MPautomata} holds. Then in any observer ${\Acal}{^{\frakM}_{obs}}$,
	there is a
	transition sequence $x_0\xrightarrow[]{\gamma_1}x_1\xrightarrow[]{\gamma_2}x_1$ such that $\gamma_1\in
	(\Sig^T_{obs})^*$, $\gamma_2\in(\Sig^T_{obs})^+$, and in $x_1\xrightarrow[]{\gamma_2}x_1$,
	all states are singletons. Hence in $\Acal^{\frakM}$, there exists an infinite path $q_0\xrightarrow[]{s_1}q_1
	\xrightarrow[]{s_2}q_1(\xrightarrow[]{s_2}q_1)^{\omega}=:\pi$ such that $\tau(\pi)\in L^{\omega}(\Acal^\frakM)$,
	$q_0\in x_0$, $\{q_1\}=x_1$,
	$\ell(\tau(q_0\xrightarrow[]{s_1}q_1))=\tau(\gamma_1)$, $\ell(\tau(q_0\xrightarrow[]{s_1}
	q_1\xrightarrow[]{s_2}q_1))=\tau(\gamma_1\gamma_2)$, and $\ell(\tau(\pi))=\tau(\gamma_1(\gamma_2)^{\omega})$.
	For all prefixes $\gamma\sqsubset\gamma_1(\gamma_2)^{\omega}$
	such that $|\gamma|\ge|\gamma_1|$, one has $|\Mt(\Acal^{\frakM},\tau(\gamma))|=1$. Then $\Acal^{\frakM}$ is weakly detectable.

	``only if'': Assume $\Acal^{\frakM}$ is weakly detectable and neither \eqref{item1_det_MPautomata} 
	nor \eqref{item2_det_MPautomata} holds. Then by the finiteness of the number of states of $\Acal^{\frakM}$
	and the Pigeonhole Principle, in pre-observer $\tensor*[^{pre}]{\Acal}{_{obs}^{\frakM}}$,
	there is a reachable cycle in which all states are singletons, then by definition,
	\eqref{item3_det_MPautomata} holds.
\end{proof}

\begin{theorem}\label{thm7_det_MPautomata}
	A labeled weighted automaton $\Acal^{\frakM}$ \eqref{LWA_monoid_det_MPautomata} 
	is 
	weakly periodically detectable if and only if either one of the following three conditions holds.
	\begin{enumerate}[(i)]
		\item\label{item11_det_MPautomata}
			$L^{\omega}(\Acal^\frakM)=\emptyset$.
		\item\label{item12_det_MPautomata}
			$L^{\omega}(\Acal^\frakM)\ne\emptyset$,
			there exists $w\in L^{\omega}(\Acal^\frakM)$ such that $\ell(w)\in (\Sig\times T)^*$
			and $|\Mt(\Acal^{\frakM},\ell(w))|=1$.
		\item\label{item13_det_MPautomata}
			$L^{\omega}(\Acal^\frakM)\ne\emptyset$ and in any one of its observers, 
			there is a reachable cycle in which at least one state is a singleton.
	\end{enumerate}
\end{theorem}

We omit the proof of Theorem~\ref{thm7_det_MPautomata} that is similar to the proof of 
Theorem~\ref{thm2_det_MPautomata}.

	We have given equivalent conditions for the four notions of detectability of a labeled weighted automaton
	$\Acal^{\frakM}$, where these conditions are represented by its self-composition $\CCa(\Acal^{\frakM})$,
	any of its observers ${\Acal}{^{\frakM}_{obs}}$, and any of its detectors
	${\Acal}{^{\frakM}_{det}}$.
	Hence the decidability of these notions directly depends on whether the corresponding 
	$\CCa(\Acal^{\frakM})$, ${\Acal}{^{\frakM}_{obs}}$, and ${\Acal}{^{\frakM}_{det}}$
	are computable, which directly depends on $\frakM$. 
	In the following Section~\ref{subsec:verification_A_Q},
	we consider monoid $(\Q^k,+,0_k)$, and show that $\CCa(\Acal^{\Q^k})$, 
	$\Acal_{obs}^{\Q^k}$, and $\Acal_{det}^{\Q^k}$ are all
	computable with complexity upper bounds.

\subsection{Verification of notions of detectability for labeled weighted automata over the monoid 
$(\Q^k,+,0_k)$}
\label{subsec:verification_A_Q}

In this subsection, we show for labeled weighted automaton $\Acal^{\Q^k}$,
its self-composition $\CCa(\Acal^{\Q^k})$ \eqref{CC_MPautomata},
observers $\Acal_{obs}^{\Q^k}$ \eqref{eqn_observer_sim_MPautomata}, and detectors
$\Acal_{det}^{\Q^k}$ \eqref{eqn_detector_sim_MPautomata}
are computable in $\NP$, $2$-$\EXPTIME$, and $2$-$\EXPTIME$,
by using the EPL problem (Lemma~\ref{lem1_det_MPautomata})
and a subclass of Presburger arithmetic (Lemma~\ref{lem4_det_MPautomata}).
As a result, the problem of verifying strong detectability of 
$\Acal^{\Q^k}$ is proven to belong to $\coNP$, and the problems of verifying strong periodic detectability,
weak detectability, and weak periodic detectability of $\Acal^{\Q^k}$ are proven to belong to $2$-$\EXPTIME$.
Particularly, for $\Acal^{\Q^k}$ in which from every state, a distinct state can be
reached through some unobservable, instantaneous path, detectors $\Acal^{\Q^k}_{det}$ can be computed in 
$\NP$, and SPD can be verified in $\coNP$. In addition, we also prove that the problems of verifying strong
detectability and strong periodic detectability of deterministic, deadlock-free, and divergence-free $\Acal^{\N}$
are both $\coNP$-hard by constructing polynomial-time reductions from 
the $\NP$-complete subset sum problem (Lemma~\ref{lem2_det_MPautomata}).

\begin{remark}\label{rem4_det_MPautomata}
	We point out that in order to characterize detectability for automaton $\Acal^{\Q^k}$, one can 
	consider automaton $\Acal^{\Z^k}$ without loss of generality.
	Given an automaton $\Acal^{\Q^k}$, enumerate the entries of the weights of all its transitions 
	by $m_1/n_1,\dots,m_l/n_l$, where $m_i$ and $n_i$ are relatively prime integers, $i\in\llb 1,l
	\rrb$, then after multiplying the entries by a large positive integer $M$, the newly obtained automaton
	$\bar\Acal^{\Q^k}$ have the weights of all its transitions in $\Z^k$, where $M$ is the least common 
	multiple of $n_1,\dots,n_l$. One has $\Acal^{\Q^k}$ is detectable
	if and only if $\bar\Acal^{\Q^k}$ is detectable with respect to all the four definitions of detectability,
	because for every two paths $\pi_1$ and $\pi_2$ in $\Acal^{\Q^k}$, they have
	the same weight in $\Acal^{\Q^k}$ if and only if they have the same weight in $\bar\Acal^{\Q^k}$.
	In addition, such a modification does not change complexity class membership when verifying detectability
	of $\Acal^{\Q^k}$ and $\bar\Acal^{\Q^k}$. Later, we will show in order to verify detectability of
	$\Acal^{\Z^{k}}$, one can consider $\Acal^{\Z}$ without loss of generality.

	However, it is not known whether verification of detectability of $\Acal^{\Z^k}$ can be transformed to 
	verification of detectability of $\Acal^{\N^k}$ without loss of generality. One can try to find a function
	$f:\Z^k\to \N^k$ that maps the weights of all transitions of $\Acal^{\Z^k}$ to elements in $\N^k$
	(the newly obtained automaton is denoted by $f(\Acal^{\Z^k})$)
	and meanwhile preserves detectability, but it is not known whether such an $f$ exists. For simplicity,
	we try to find an $f$ such that
	\begin{align*}
		&\left( \forall x_1,\dots,x_m,y_1,\dots,y_n\in\Z^k \right)
		&\left[ \sum_{i=1}^{m}x_i=\sum_{i=1}^{n}y_i \iff \sum_{i=1}^{m}f(x_i)=\sum_{i=1}^{n}f(y_i) \right],
	\end{align*}
	because such an $f$ preserves detectability between $\Acal^{\Z^k}$ and $f(\Acal^{\Z^k})$
	($f$ guarantees that for every two paths $\pi_1$ and $\pi_2$ in $\Acal^{\Z^k}$, they have the same weight
	in $\Acal^{\Z^k}$ if and only if they have the same weight in $f(\Acal^{\Z^k})$). However, it is 
	easy to see that such an $f$ does not exist.
\end{remark}

\subsubsection{Computation of self-composition $\CCa(\Acal^{\Q^k})$ and verification of strong detectability}
\label{subsubsec:computation_self_composition}

As shown in Remark~\ref{rem4_det_MPautomata}, without loss of generality we compute 
$\CCa(\Acal^{\Z^k})$. Moreover, because we will reduce computation of $\CCa(\Acal^{\Z^k})$ to the multidimensional 
EPL problem which belongs to $\NP$, and the $1$-dimensional EPL problem is already $\NP$-hard 
(Lemma~\ref{lem1_det_MPautomata}), we compute $\CCa(\Acal^{\Z})$ without loss of generality.
Next we compute $\CCa(\Acal^{\Z})$ \eqref{CC_MPautomata}.
Given states $(q_1,q_2),(q_3,q_4)\in Q'$ and event $(e_1,e_2)\in E'$, we verify whether there is a
transition $$((q_1,q_2),(e_1,e_2),(q_3,q_4))\in\dt'$$ as follows:
\begin{enumerate}[(i)]
	\item\label{item4_det_MPautomata}
		Guess states $q_5,q_6,q_7,q_8\in Q$ such that there exist transitions 
		$q_5\xrightarrow[]{e_1}q_7$, $q_6\xrightarrow[]{e_2}q_8$ and unobservable, instantaneous paths
		$q_7\xrightarrow[]{s_3}q_3$, $q_8\xrightarrow[]{s_4}q_4$, where $s_3,s_4\in (E_{uo})^*$.
	\item\label{item5_det_MPautomata} 
		Check whether there exist unobservable paths $q_1\xrightarrow[]{s_1}q_5$, $q_2\xrightarrow[]{s_2}q_6$,
		where $s_1,s_2\in (E_{uo})^*$, such that the weights of paths $q_1\xrightarrow[]{s_1}q_5\xrightarrow[]{e_1}
		q_7$, $q_2\xrightarrow[]{s_2}q_6\xrightarrow[]{e_2}q_8$ are the same. If such paths $q_1\xrightarrow[]
		{s_1}q_5$, $q_2\xrightarrow[]{s_2}q_6$ exist, then one has $((q_1,q_2),(e_1,e_2),(q_3,q_4))\in\dt'$.
\end{enumerate}

Next we check the above \eqref{item5_det_MPautomata}. Firstly, compute subautomata $\Acal_{q_1}^{\Z}$ 
(resp., $\Acal_{q_2}^{\Z}$) of $\Acal^{\Z}$ starting at $q_1$ (resp., $q_2$) and passing through
exactly all possible unobservable transitions. Secondly, compute asynchronous product $\Acal_{q_1}^{\Z}
\otimes\Acal_{q_2}^{\Z}$ of $\Acal_{q_1}^{\Z}$ and $\Acal_{q_2}^{\Z}$, where the states
of the product are exactly pairs $(p_1,p_2)$ with $p_1$ and $p_2$ being states of
$\Acal_{q_1}^{\Z}$ and
$\Acal_{q_2}^{\Z}$, respectively; transitions are of the form 
\begin{equation}\label{eqn20_det_MPautomata}
	(p_1,p_2)\xrightarrow[]{(\ep,e)/-\mu(e)_{p_2p_3}}(p_1,p_3),
\end{equation}
where 
$p_2\xrightarrow[]{e/\mu(e)_{p_2p_3}}p_3$ is a transition of $\Acal_{q_2}^{\Z}$,
or of the form 
\begin{equation}\label{eqn21_det_MPautomata}
	(p_1,p_2)\xrightarrow[]{(e,\ep)/\mu(e)_{p_1p_3}}(p_3,p_2),
\end{equation}
where $p_1\xrightarrow[]{e/\mu(e)_{p_1p_3}}p_3$ is a transition of $\Acal_{q_1}^{\Z}$.
Regard $\Acal_{q_1}^{\Z}\otimes\Acal_{q_2}^{\Z}$ as a weighted directed graph, and the
above $-\mu(e)_{p_2p_3}$ and $\mu(e)_{p_1p_3}$ as the weights of transitions \eqref{eqn20_det_MPautomata} and
\eqref{eqn21_det_MPautomata}.
Finally, check in $\Acal_{q_1}^{\Z}\otimes
\Acal_{q_2}^{\Z}$, whether there is a path from $(q_1,q_2)$ to $(q_5,q_6)$ whose weight is equal to 
$\mu(e_2)_{q_6q_8}-\mu(e_1)_{q_5q_7}$, which is actually a $1$-dimensional EPL problem 
(Problem~\ref{prob1_det_MPautomata}).
Then since the EPL problem belongs to $\NP$ (Lemma~\ref{lem1_det_MPautomata}), the following 
result holds. 

\begin{theorem}\label{thm3_det_MPautomata}
	The self-composition $\CCa(\Acal^{\Q^k})$ of labeled weighted automaton $\Acal^{\Q^k}$ 
	can be computed in $\NP$ in the size of $\Acal^{\Q^k}$.
\end{theorem}

\begin{example}\label{exam6_MPautomata}
	We use automaton $\Acal_1^{\N}$ in Fig.~\ref{fig6_det_MPautomata} to illustrate how to 
	compute $\CCa(\Acal_1^{\N})$. Recall its self-composition $\CCa(\Acal_1^{\N})$ shown
	in Fig.~\ref{fig7_det_MPautomata}. We check whether there exists a transition $((q_1,q_2),(b,b),(q_3,q_3))$
	in $\CCa(\Acal_1^{\N})$ as follows: (1) Guess transitions $q_1\xrightarrow[]{b/2}q_3$ and
	$q_2\xrightarrow[]{b/1}q_3$ of $\Acal_1^{\N}$. Because the two transitions have different weights,
	now we do not know whether there exists a transition $((q_1,q_2),(b,b),(q_3,q_3))$ in $\CCa(\Acal_1^{\N})$.
	(2) Compute subautomata $\Acal_{1q_1}^{\N}$ and $\Acal_{1q_2}^{\N}$ and their asynchronous product
	$\Acal_{1q_1}^{\N}\otimes \Acal_{1q_2}^{\N}$ as in Fig.~\ref{fig9_det_MPautomata}.
	The rest is to check whether there exists a path from $(q_1,q_2)$ to $(q_1,q_2)$ in
	$\Acal_{1q_1}^{\N}\otimes \Acal_{1q_2}^{\N}$ with weight $\mu(b)_{q_2q_3}-\mu(b)_{q_1q_3}=1-2=-1$.
	The answer is YES: $(q_1,q_2)\xrightarrow[]{(\ep,u)}(q_1,q_2)$ is such a path. By these transitions and 
	paths we find two paths $q_1\xrightarrow[]{b}q_3=:\pi_1$ and $q_2\xrightarrow[]{u}q_2\xrightarrow[]{b}q_3=:\pi_2$
	such that they have the same weight. Note that $\pi_1$ and $\pi_2$ are exactly the $\pi_1$ and $\pi_2$ in 
	Example~\ref{exam3_MPautomata}. Then we conclude that there exists a transition $((q_1,q_2),(b,b),(q_3,q_3))$
	in $\CCa(\Acal_1^{\N})$. The other transitions of $\CCa(\Acal_1^{\N})$ can be computed similarly.
	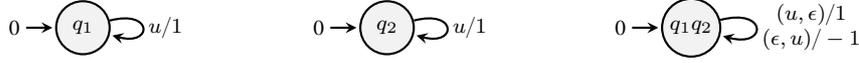
\begin{figure}[!htbp]
        \centering
	\begin{tikzpicture}
	[>=stealth',shorten >=1pt,thick,auto,node distance=2.0 cm, scale = 1.0, transform shape,
	->,>=stealth,inner sep=2pt, initial text = 0]

	\tikzstyle{emptynode}=[inner sep=0,outer sep=0]

	\node[initial, state, initial where = left] (q1) {$q_1$};

	\path [->]
	(q1) edge [loop right] node {$u/1$} (q1)
	;

	\node[emptynode, right of = q1] (empty1) {};

	\node[initial, state, initial where = left, right of = empty1] (q2) {$q_2$};

	\path [->]
	(q2) edge [loop right] node {$u/1$} (q2)
	;

	\node[emptynode, right of = q2] (empty2) {};

	\node[initial, state, initial where = left, right of = empty2] (q1q2) {$q_1q_2$};

	\path [->]
	(q1q2) edge [loop right] node {$\begin{matrix}(u,\ep)/1\\(\ep,u)/-1\end{matrix}$} (q1q2)
	;

        \end{tikzpicture}
		\caption{Subautomaton $\Acal_{1q_1}^{\N}$ (left) and subautomaton $\Acal_{1q_2}^{\N}$ (middle)
		of labeled weighted automaton $\Acal_1^{\N}$ in Fig.~\ref{fig6_det_MPautomata}, and their
		asynchronous product $\Acal_{1q_1}^{\N}\otimes \Acal_{1q_2}^{\N}$ (right).} 
		\label{fig9_det_MPautomata} 
	\end{figure}
\end{example}

One can see that the condition in Theorem~\ref{thm1_det_MPautomata} can be verified in time linear in the 
size of $\CCa(\Acal^{\Q^k})$ by computing its strongly connected components
(a similar check is referred to \cite[Theorem 3]{Zhang2020DetPNFA}),
then the following result holds.

\begin{theorem}\label{thm4_det_MPautomata}
	The problem of verifying strong detectability of labeled weighted automaton $\Acal^{\Q^k}$
	belongs to $\coNP$.
\end{theorem}

Particularly, one directly sees from the process of computing $\CCa(\Acal^{\Z})$ that,
if all transitions of $\Acal^{\Z}$ are observable, then its self-composition $\CCa(\Acal^{\Z})$ 
can be computed in polynomial time. Hence we have
the following direct corollary. 

\begin{corollary}\label{cor1_det_MPautomata}
	Consider a labeled weighted automaton $\Acal^{\Q^k}$ all of whose transitions are observable.
	Its self-composition $\CCa(\Acal^{\Q^k})$ can be computed in polynomial time,
	and its strong detectability can be verified also in polynomial time.
\end{corollary}

\subsubsection{Computation of observer $\Acal_{obs}^{\Q^k}$ and verification of
weak detectability and weak periodic detectability}
\label{subsubsec:observer}

Also as shown in Remark~\ref{rem4_det_MPautomata}, without loss of generality we compute observer 
$\Acal_{obs}^{\Z^k}$ \eqref{eqn_observer_sim_MPautomata} of automaton $\Acal^{\Z^k}$.
To this end, we will repetitively use the EPL problem (Lemma~\ref{lem1_det_MPautomata}) and
the subcalss of Presburger arithmetic in Lemma~\ref{lem4_det_MPautomata} with at most $2$ quantifier
alternations. The final result in this part is that $\Acal^{\Z^k}_{obs}$ can be computed in 
$2$-$\EXPTIME$ in the size of $\Acal^{\Z^k}$. Since as a special case, $\Acal^{\Z}_{obs}$ is computed
in $2$-$\EXPTIME$ in the size of $\Acal^{\Z}$, to simplify the notation in the following statement,
we consider $\Acal^{\Z}$.
The initial state $x_0=\Mt(\Acal^{\Z},\ep)$ can be directly computed by starting at an initial state of $\Acal^{\Z}$
and passing through all possible unobservable, instantaneous paths. We then start from $x_0$,
find all reachable states step by step together with the corresponding transitions.

\begin{enumerate}[(i)]
	\item \label{item18_det_MPautomata} 
		Choose a state $x_1\in X$ that we have just computed. Choose $\s\in\Sig$. 
		For each $q\in x_1$, compute subautomaton $\Acal_{q}^{\Z}$ that consists of all paths 
		of the form 
		\begin{equation}\label{eqn3_det_MPautomata}
			q\xrightarrow[]{s}q^1\xrightarrow[]{e}q^2
		\end{equation}
		of $\Acal^{\Z}$ such that $s\in(E_{uo})^*$, $e\in E_{\s}$ (i.e., $e$ is observable and $\ell(e)=\s$).
		Denote the set of all such 
		$q^2$ by $x_2$. Note that one may have $x_2=\emptyset$, $|x_2|>|x_1|$, $|x_2|=|x_1|$, or $|x_2|<|x_1|$.


	\item\label{item19_det_MPautomata}
		Choose an arbitrary $\bar x_2=\{\bar q_1^2,\dots,\bar q_n^2\}\subset x_2$, where $|\bar x_2|=n$.
		For every $\bar q_i^2$, $i\in\llb 1,n\rrb$, \textbf{nondeterministically} choose a path
 		\begin{align}\label{eqn22_det_MPautomata}
			\bar q_i\xrightarrow[]{\bar s_i}\bar q_i^1\xrightarrow[]{\bar e_i}\bar q_i^2
		\end{align} as in \eqref{eqn3_det_MPautomata}, where $\bar q_i\in x_1$.

	\item\label{item9_det_MPautomata}
		Compute asynchronous product (which will be regarded as a weighted automaton and a weighted
		directed graph)
		\begin{equation}\label{eqn4_det_MPautomata}
			\Acal_{\bar q_{1}}^{\Z}\otimes\cdots\otimes\Acal_{\bar q_{n}}^{\Z},
		\end{equation}
		where the states of the asynchronous product/automaton are $(\hat q_1,\dots,\hat q_n)$,
		where $\hat q_{i}$ are states of $\Acal_{\bar q_i}^{\Z}$, $i\in\llb 1,n\rrb$; there is a transition
		\begin{align}\label{eqn23_det_MPautomata} 
			(\hat q_1,\dots,\hat q_n)
			\xrightarrow[]{(\hat e_1,\dots,\hat e_n)\left/\left(\mu(\hat e_1)_{\hat q_1\hat q_1'},\dots,
			\mu(\hat e_n)_{\hat q_n\hat q_n'}\right)\right.} 
			(\hat q_1',\dots,\hat q_n')
		\end{align}
		in automaton \eqref{eqn4_det_MPautomata} if and only if one of the following two conditions holds.
		\begin{enumerate}
			\item\label{item6_det_MPautomata}
				For some $i\in\llb 1,n\rrb$, $\hat q_i\xrightarrow[]{\hat e_i}\hat q_i'$
				is an unobservable transition of $\Acal_{\bar q_i}^{\Z}$, for all other $j\in\llb 1,n\rrb$, 
				$\hat e_j=\ep$ and $\hat q_j=\hat q_j'$. In this case, \eqref{eqn23_det_MPautomata} is called 
				\emph{unobservable}. 
			\item\label{item7_det_MPautomata}
				For every $i\in\llb 1,n\rrb$, $\hat q_i\xrightarrow[]{\hat e_i}\hat q'_i$ is an observable
				transition of $\Acal^{\Z}_{\bar q_i}$, $\hat q_i=\bar q_i^1$, $\hat e_i=\bar e_i$, and
				$\hat q'_i=\bar q_i^2$. In this case, \eqref{eqn23_det_MPautomata} is called \emph{observable}.
		\end{enumerate}
		Automaton \eqref{eqn4_det_MPautomata} has at most $|Q|^{n}$ states, at most $n|\Dt_{uo}||Q|^{n-1}$
		unobservable transitions (recall $\Dt_{uo}=\{(q,e,q')\in\Dt|\ell(e)=\ep,e\ne\ep\}$),
		and exactly $1$ observable transition. 

		If in automaton \eqref{eqn4_det_MPautomata},
		\begin{equation*}
			\tag{A}\label{quoteA_det_MPautomata}
			\parbox{\dimexpr\linewidth-4em}{%
			\strut
				there exists a path $\pi$ from the initial state $(\bar q_1,\dots,\bar q_n)$ to the state
				$(\bar q_1^2,\dots,\bar q_n^2)$ 
				such that only the last transition is observable and all components of the path
				have the same weight, which is denoted by $t\in\Z$,
			\strut
			}
		\end{equation*}
		then $\Mt(\Acal^{\Z},\ep|\bar x_2)\subset \Mt(\Acal^{\Z},(\s,t)|x_1)$ (defined in \eqref{ISE_det_MPautomata} 
		and \eqref{CSE_x_det_MPautomata}).
		We say a subset $\bar x_2\subset x_2$ is a \emph{pre-successor} of $x_1$ if
		for every $\bar q_i^2\in\bar x_2$, there exists a path $\bar q_i\xrightarrow[]{\bar s_i}\bar q_i^1
		\xrightarrow[]{\bar e_i}\bar q_i^2$ as in \eqref{eqn22_det_MPautomata} such that the corresponding
		automaton \eqref{eqn4_det_MPautomata} 
		satisfies condition \eqref{quoteA_det_MPautomata}. Then by definition, if $\bar x_2$
		is a pre-successor of $x_1$ and no other $\tilde x_2$ satisfying $\bar x_2\subsetneq \tilde x_2\subset x_2$
		is a pre-successor of $x_1$, then in $\Acal_{obs}^{\Z}$, there is a transition $x_1\xrightarrow[]{
		(\s,t)}\Mt(\Acal^{\Z},\ep|\bar x_2)$, and there is no transition from $x_1$ to $\Mt(\Acal^{\Z},\ep|\tilde x_2)$
		for any of such $\tilde x_2$ in case $\Mt(\Acal^{\Z},\ep|\bar x_2)\subsetneq \Mt(\Acal^{\Z},\ep|\tilde x_2)$,
		where $t$ is the weight of any components of a path $\pi$ 
		as in \eqref{quoteA_det_MPautomata}. Such special pre-successors $\bar x_2$ are called \emph{successors}
		of $x_1$.

		In order to find all successors of $x_1$, we first check whether $x_2$ is a pre-successor of $x_1$.
		If the answer is YES, then $x_2$ is a successor of $x_1$ and we obtain a transition of $\Acal_{obs}
		^{\Z}$ from $x_1$ to $\Mt(\Acal^{\Z},\ep|x_2)$, we also know no strict subset of $x_2$ will be a successor
		of $x_1$; otherwise, we check whether subsets of $x_2$ are successors of $x_1$ in a decreasing order of
		cardinality. Once we find a successor of $x_1$, we do not need to check the strict subsets of the successor 
		because its strict subsets will not be successors of $x_1$.

		A decision procedure for checking condition \eqref{quoteA_det_MPautomata} is 
		given in \eqref{item10_det_MPautomata}.
	\item\label{item10_det_MPautomata}
		In \eqref{eqn4_det_MPautomata}, we delete the observable transition, and
		replace the weight $$\left(\mu(\hat e_1)_{\hat q_1\hat q_1'},\dots,
		\mu(\hat e_n)_{\hat q_n\hat q_n'}\right)$$ (see \eqref{eqn23_det_MPautomata}) of each transition by
		$$\left(\mu(\hat e_2)_{\hat q_2\hat q_2'}-\mu(\hat e_1)_{\hat q_1\hat q_1'},\dots,
		\mu(\hat e_n)_{\hat q_n\hat q_n'}-\mu(\hat e_1)_{\hat q_1\hat q_1'}\right),$$
		then we obtain a new automaton 
		\begin{align}\label{eqn24_det_MPautomata}
			\reallywidehat{\Acal_{\bar q_{1}}^{\Z}\otimes\cdots\otimes\Acal_{\bar q_{n}}^{\Z}}.
		\end{align}
		Then we check whether
		\begin{equation*}
			\tag{B}\label{quoteB_det_MPautomata}
			\parbox{\dimexpr\linewidth-4em}{%
			\strut
				there is a path $\bar \pi$ from $(\bar q_1,\dots,\bar q_n)$ to $(\bar q_1^1,\dots,\bar q_n^1)$ in
				\eqref{eqn24_det_MPautomata} with weight
					$(\mu(\bar e_1)_{\bar q_1^1\bar q_1^2}-\mu(\bar e_2)_{\bar q_2^1\bar q_2^2},\dots,\mu(\bar e_1)
					_{\bar q_1^1\bar q_1^2}-\mu(\bar e_n)_{\bar q_n^1\bar q_n^2}).$
			\strut
			}
		\end{equation*}
		
		If the answer
		is YES, then $\bar\pi\xrightarrow[]{(\bar e_1,\dots,\bar e_n)}(\bar q_1^2,\dots,\bar q_n^2)$ is a path of 
		\eqref{eqn4_det_MPautomata} as shown in \eqref{quoteA_det_MPautomata}, i.e.,
		the weight $(w_1,\dots,w_n)$ of $\bar\pi$ in \eqref{eqn4_det_MPautomata} satisfies 
		$$(w_1,\dots,w_n)+(\mu(\bar e_1)_{\bar q_1^1\bar q_1^2},\dots,\mu(\bar e_n)_{\bar q_n^1\bar q_n^2})$$ 
		has equal components, and $\Mt(\Acal^{\Z}, \ep|\bar x_2)\subset\Mt(\Acal^{\Z},(\s,w_1+\mu(\bar e_1)
		_{\bar q_1^1\bar q_1^2})|x_1)$. That is, $\bar x_2$ is a pre-successor of $x_1$.
		If additionally $\bar x_2$ is a successor of $x_1$, then 
		\begin{align}\label{eqn30_det_MPautomata}
			x_1\xrightarrow[]{\left(\s,w_1+\mu(\bar e_1)_{\bar q_1^1\bar q_1^2}\right)}\Mt(\Acal^{\Z},\ep|\bar x_2)
		\end{align}
		is a transition of $\Acal^{\Z}_{obs}$. 

		\eqref{quoteB_det_MPautomata} is an EPL problem of weighted directed graph \eqref{eqn24_det_MPautomata},
		and hence can be
		checked in $\NP$ in the size of \eqref{eqn24_det_MPautomata} by Lemma~\ref{lem1_det_MPautomata}.
	\item\label{item23_det_MPautomata} 
		In the following, we check for every successor $\bar x_2$ of $x_1$, for every $\emptyset\ne
		\hat x_2\subsetneq\bar x_2$ satisfying $\Mt(\Acal^{\Z},\ep|\hat x_2)\subsetneq \Mt(\Acal^{\Z},\ep|\bar x_2)$, 
		whether
		\begin{equation*}
			\tag{C}\label{quoteC_det_MPautomata}
			\parbox{\dimexpr\linewidth-4em}{%
			\strut
				there is a transition from $x_1$ to $\Mt(\Acal^{\Z},\ep|\hat x_2)$ in $\Acal^{\Z}_{obs}$, 
			\strut
			}
		\end{equation*}
		which is equivalent to whether
		\begin{equation*}
			\tag{D}\label{quoteD_det_MPautomata}
			\parbox{\dimexpr\linewidth-4em}{%
			\strut
				there exists $t\in\Z$ such that for every $\hat q_i^2\in\hat x_2$, there exists
				a path $\hat q_i\xrightarrow[]{\hat s_i}\hat q_i^1\xrightarrow[]{\hat e_i}\hat q_i^2$ as 
				in \eqref{eqn22_det_MPautomata} with weight $t$, and for any $q\in x_1$ and any
				$q^2\in x_2\setminus \hat x_2$, there exists no path $q\xrightarrow[]{s}q^1\xrightarrow
				[]{e}q^2$ as in \eqref{eqn3_det_MPautomata} with weight $t$.
			\strut
			}
		\end{equation*}
		We will use the subclass of Presburger arithmetic as in Lemma~\ref{lem4_det_MPautomata} to check
		\eqref{quoteD_det_MPautomata}.

		Denote $\hat x_2=\{\hat q_1^2,\dots,\hat q_m^2\}$, where $|\hat x_2|=m$. For every $i\in\llb 1,m\rrb$,
		\textbf{nondeterministically} choose a path $\hat q_i\xrightarrow[]{\hat s_i}\hat q_i^1\xrightarrow[]{\hat e_i}
		\hat q_i^2$	as in \eqref{eqn22_det_MPautomata}, then $\hat q_i\in x_1$; consider a copy 
		of $\Acal^{\Z}_{\hat q_i}$
		as in \eqref{item18_det_MPautomata}, and use $\Acal^{\Z}_{\frakB_i}$ to denote the subautomaton 
		of $\Acal^{\Z}_{\hat q_i}$ obtained by deleting all observable transitions and all states that do not
		belong to any unobservable path from $\hat q_i$ to $\hat q_i^1$, hence $\Acal^{\Z}_{\frakB_i}$ consists
		of all unobservable paths from $\hat q_i$ to $\hat q_i^1$ in $\Acal^{\Z}_{\hat q_i}$. We call $\hat q_i$
		the \emph{source} and $\hat q_i^1$ the \emph{sink}.

		Denote 
		\begin{equation}\label{eqn38_det_MPautomata}
			\begin{split}
				\frakE =:  \{ (q,q^1,e,q^2)| &q\in x_1,e\in E_\s,q^2\in x_2\setminus\hat x_2,\\
				&\text{there exists a path }q\xrightarrow[]{s}q^1\xrightarrow[]{e}q^2\text{ as in 
				\eqref{eqn3_det_MPautomata}}.
			\end{split}
		\end{equation}
		Then $|\frakE|\le |x_1||\Dt_{\s}|(|x_2|-|\hat x_2|)$, where recall $\Dt_{\s}=\{(q,e,q')\in\Dt|\ell(e)=\s\}$.
		Rewrite $\frakE$ as $\{(q_j,q_j^1,e_j,q_j^2)|j\in\llb 1,|\frakE|\rrb\}$ and denote $(q_j,q_j^1,e_j,q_j^2)
		=:\frakE_j$, compute subautomaton $\Acal^{\Z}_{\frakE_j}$ from $\Acal^{\Z}_{q_j}$ by deleting 
		all observable transitions and all states that do not belong to any unobservable path from $q_j$ to $q_j^1$,
		hence $\Acal^{\Z}_{\frakE_j}$ consists of all unobservable paths from $q_j$ to $q_j^1$ in $\Acal^{\Z}
		_{q_j}$. We call $q_j$ the \emph{source} and $q_j^1$ the \emph{sink}.

		Then \eqref{quoteD_det_MPautomata} is satisfied if and only if the following \eqref{quoteE_det_MPautomata} 
		is satisfied.
		\begin{equation*}
			\tag{E}\label{quoteE_det_MPautomata}
			\parbox{\dimexpr\linewidth-4em}{%
			\strut
				There exists $t\in\Z$ such that in $\Acal^{\Z}_{\frakB_i}$, $i\in\llb 1,m\rrb$, 
				there exists a path from $\hat q_i$ to $\hat q_i^1$ with weight $t-\mu(\hat e_i)
				_{\hat q_i^1\hat q_i^2}$; and in $\Acal^{\Z}_{\frakE_j}$, $j\in\llb 1,|\frakE|\rrb$,
				there exists no path from $q_j$ to $q_j^1$ with weight $t-\mu(e_j)_{q_j^1q_j^2}$.
			\strut
			}
		\end{equation*}

		If \eqref{quoteE_det_MPautomata} is satisfied, then there is a transition $x_1
		\xrightarrow[]{(\s,t)}\Mt(\Acal^{\Z},\ep|\hat x_2)$ in $\Acal^{\Z}_{obs}$.

		We will equivalently transform \eqref{quoteE_det_MPautomata} to a Presburger formula as in 
		Lemma~\ref{lem4_det_MPautomata} and check satisfiability of the formula. 
		Rewrite $\Acal^{\Z}_{\frakB_i}$ as a weighted directed graph $G_{\frakB_i}
		=(\Z,V_{\frakB_i},A_{\frakB_i})$, $i\in\llb 1,m\rrb$; also rewrite $\Acal^{\Z}_{\frakE_j}$ as 
		a weighted directed graph $G_{\frakE_j}=(\Z,V_{\frakE_j},A_{\frakE_j})$, $j\in\llb 1,|\frakE|\rrb$.
		Then \eqref{quoteE_det_MPautomata} is equivalent to the following \eqref{quoteF_det_MPautomata}.
		\begin{equation*}
			\tag{F}\label{quoteF_det_MPautomata}
			\parbox{\dimexpr\linewidth-4em}{%
			\strut
				There exists $t\in\Z$ such that in graph $G_{\frakB_i}$, $i\in\llb 1,m\rrb$, 
				there exists a path $\hat\pi_i$ from $\hat q_i$ to $\hat q_i^1$ with weight $t-\mu(\hat e_i)
				_{\hat q_i^1\hat q_i^2}$; and in graph $G_{\frakE_j}$, $j\in\llb 1,|\frakE|\rrb$,
				there exists no path from $q_j$ to $q_j^1$ with weight $t-\mu(e_j)_{q_j^1q_j^2}$.
			\strut
			}
		\end{equation*}

		In graph $G_{\frakB_i}$, $i\in\llb 1,m\rrb$, for each edge $a_i\in A_{\frakB_i}$, denote its
		weight by $w_{a_i}$ and define a variable $y_{a_i}\in\N$ which indicates how many times $a_i$ is included
		in a path. In graph $G_{\frakE_j}$, $j\in\llb 1,|\frakE|\rrb$, for each edge $a_j\in A_{\frakE_j}$, also
		denote its weight by $w_{a_j}$ and define a variable $y_{a_j}\in\N$ as above. In these graphs, an edge
		$a$ such that $y_a>0$ is called \emph{realizable}, a path is called \emph{realizable} if all its edges are 
		realizable. We furthermore have \eqref{quoteF_det_MPautomata} is satisfied if and only if
		\begin{subequations}\label{eqn39_det_MPautomata}
			\begin{align}
				& ( \exists_{a_1\in A_{\frakB_1}} y_{a_1}\in \N ) \dots 
				( \exists_{a_m\in A_{\frakB_m}} y_{a_m}\in \N ) \label{eqn39_a_det_MPautomata} \\
				& ( \forall_{a_1\in A_{\frakE_1}} y_{a_1}\in \N ) \dots 
				( \forall_{a_{|\frakE|}\in A_{\frakE_{|\frakE|}}} y_{a_{|\frakE|}}\in \N ) \label{eqn39_b_det_MPautomata} \\
				& \left[
				\bigwedge_{i\in\llb 1,m\rrb} \bigwedge_{v_i\in V_{\frakB_i}\setminus\{\hat q_i,\hat q_i^1\}}
				\left( \sum_{\substack{a_i\in A_{\frakB_i}\\ \tail(a_i)=v_i}}y_{a_i} =
				\sum_{\substack{a_i\in A_{\frakB_i}\\ \head(a_i)=v_i}}y_{a_i}\right) \wedge
				\right.\label{eqn39_c_det_MPautomata}\\
				& \bigwedge_{i\in\llb 1,m\rrb} 
				\left( \sum_{\substack{a_i\in A_{\frakB_i}\\ \tail(a_i)=\hat q_i }}y_{a_i} =
				\sum_{\substack{a_i\in A_{\frakB_i}\\ \head(a_i)=\hat q_i}}y_{a_i}+1\right) \wedge\label{eqn39_d_det_MPautomata}\\
				& \bigwedge_{i\in\llb 1,m\rrb} 
				\left( \sum_{\substack{a_i\in A_{\frakB_i}\\ \tail(a_i)=\hat q_i^1}}y_{a_i} + 1 =
				\sum_{\substack{a_i\in A_{\frakB_i}\\ \head(a_i)=\hat q_i^1}}y_{a_i}\right) \wedge \label{eqn39_e_det_MPautomata} \\ 
				& \bigwedge_{i\in\llb 1,m \rrb} \bigwedge_{a_i\in A_{\frakB_i}}
				\left( \begin{array}{c} y_{a_i} > 0 \wedge \\ \head(a_i)\ne\hat q_i^1 \end{array} \implies
				\begin{array}{c} \exists \text{ a realizable path} \\\text{from }\head(a_i)\text{ to }\hat q_i^1 
				\end{array}
				\right) \wedge \label{eqn39_f_det_MPautomata}\\
				& \bigwedge_{i\in\llb 2,m \rrb}
				\left( \sum_{a_1\in A_{\frakB_1}}w_{a_1}y_{a_1} + \mu(\hat e_1)_{\hat q_1^1\hat q_1^2}
				= \sum_{a_i\in A_{\frakB_i}}w_{a_i}y_{a_i} + \mu(\hat e_i)_{\hat q_i^1\hat q_i^2} \right)\wedge\label{eqn39_g_det_MPautomata}\\
				& \bigwedge_{j\in\llb 1,|\frakE|\rrb} \left(  \bigwedge_{v_j\in V_{\frakE_j}\setminus\{q_j,q_j^1\}}
				\left( \sum_{\substack{a_j\in A_{\frakE_j}\\ \tail(a_j)=v_j}}y_{a_j} =
				\sum_{\substack{a_j\in A_{\frakE_j}\\ \head(a_j)=v_j}}y_{a_j} \right) \wedge \right. \label{eqn39_h_det_MPautomata} \\
				& 
				\left( \sum_{\substack{a_j\in A_{\frakE_j}\\ \tail(a_j)=q_j }}y_{a_j} =
				\sum_{\substack{a_j\in A_{\frakE_j}\\ \head(a_j)=q_j }}y_{a_j}+1 \right) \wedge \label{eqn39_i_det_MPautomata}\\ 
				& 
				\left( \sum_{\substack{a_j\in A_{\frakE_j}\\ \tail(a_j)=q_j^1 }}y_{a_j} + 1 =
				\sum_{\substack{a_j\in A_{\frakE_j}\\ \head(a_j)=q_j^1 }}y_{a_j}\right) \wedge\label{eqn39_j_det_MPautomata} \\
				& 
				\bigwedge_{a_j\in A_{\frakE_j}}
				\left( 
				\begin{array}{c} y_{a_j} > 0 \wedge \\ \head(a_j)\ne q_j^1 \end{array} \implies
				\begin{array}{c} \exists \text{ a realizable path} \\\text{from }\head(a_j)\text{ to }q_j^1 
				\end{array}
				\right) \label{eqn39_k_det_MPautomata} \\
				&\implies \left. 
				\left. \sum_{a_1\in A_{\frakB_1}}w_{a_1}y_{a_1} + \mu(\hat e_1)_{\hat q_1^1\hat q_1^2}
				\ne \sum_{a_j\in A_{\frakE_j}}w_{a_j}y_{a_j} + \mu(e_j)_{q_j^1q_j^2} \right) \label{eqn39_l_det_MPautomata}
				\right]
			\end{align}
		\end{subequations}
		is satisfied.
		In \eqref{eqn39_a_det_MPautomata}, $( \exists_{a_1\in A_{\frakB_1}} y_{a_1}\in \N ) $ is 
		short for $(\exists y_{a_1^1}\in\N)\dots(\exists y_{a_1^{\alpha}}\in\N)$, where $a_1^1,\dots,a_1^{\alpha}$
		is an arbitrary arrangement of the edges of $A_{\frakB_1}$. The other expressions in 
		\eqref{eqn39_a_det_MPautomata} and the expressions in \eqref{eqn39_b_det_MPautomata} have analogous meanings.
		Hence \eqref{eqn39_a_det_MPautomata} contains $\sum_{i=1}^{m}|A_{\frakB_i}|$ existential
		quantifiers and \eqref{eqn39_b_det_MPautomata} contains $\sum_{j=1}^{|\frakE|}|A_{\frakE_j}|$
		universal quantifiers. One can see \eqref{eqn39_c_det_MPautomata},
		\eqref{eqn39_d_det_MPautomata}, and \eqref{eqn39_e_det_MPautomata} are all satisfied if and only if
		in every graph $G_{\frakB_i}$,
		$i\in\llb 1,m\rrb$, all edges $a$ satisfying $y_a>0$ form one path from $\hat q_i$ to $\hat q_i^1$ and
		possibly several disjoint cycles, where these cycles do not intersect with the path. Then 
		\eqref{eqn39_c_det_MPautomata}, \eqref{eqn39_d_det_MPautomata}, \eqref{eqn39_e_det_MPautomata}, and
		\eqref{eqn39_f_det_MPautomata} are all satisfied if and only if in every graph $G_{\frakB_i}$,
		$i\in\llb 1,m\rrb$, all edges $a$ satisfying $y_a>0$ form exactly one path $\hat\pi_i$ from
		$\hat q_i$ to $\hat q_i^1$; then together with \eqref{eqn39_g_det_MPautomata}, these paths additionally
		satisfy $\WEIGHT_{\hat\pi_k}+\mu(\hat e_k)_{\hat q_k^1\hat q_k^2}=\WEIGHT_{\hat\pi_l}+\mu(\hat e_l)
		_{\hat q_l^1\hat q_l^2}=:t$ for all different $k,l$ in $\llb 1,m \rrb$ (note that in
		\eqref{eqn39_g_det_MPautomata},
		$\sum_{a_i\in A_{\frakB_i}}w_{a_i}y_{a_i}$ is the weight $\WEIGHT_{\hat \pi_i}$ of path $\hat \pi_i$,
		$i\in\llb 1,m \rrb$).
		Analogously, \eqref{eqn39_h_det_MPautomata}, \eqref{eqn39_i_det_MPautomata}, \eqref{eqn39_j_det_MPautomata}, 
		\eqref{eqn39_k_det_MPautomata}, and \eqref{eqn39_l_det_MPautomata} are all satisfied if and only if
		in every graph $G_{\frakE_j}$, $j\in\llb 1,|\frakE|\rrb$, if all edges $a$ satisfying $y_a>0$ form
		exactly one path $\pi_j$ from
		$q_j$ to $q_j^1$ then $\WEIGHT_{\pi_j}+\mu(e_j)_{q_j^1q_j^2}$ is not equal to $t$. Hence
		\begin{equation*}
			\parbox{\dimexpr\linewidth-4em}{%
			\strut
				the whole \eqref{eqn39_det_MPautomata} is satisfied only if there exists a transition $x_1
				\xrightarrow[]{(\s,t)}\Mt(\Acal^{\Z},\ep|\hat x_2)$ in $\Acal^{\Z}_{obs}$.
			\strut
			}
		\end{equation*}

	\item\label{item24_det_MPautomata} 
		Now we specify \eqref{eqn39_f_det_MPautomata} and \eqref{eqn39_k_det_MPautomata} to make the whole 
		\eqref{eqn39_det_MPautomata} become a Presburger formula. 

		Consider an arbitrary $i\in\llb 1,m \rrb$ and an arbitrary edge $a_i\in A_{\frakB_i}$
		such that $\head(a_i)\ne \hat q_i^1$. Compute a subgraph $G_{\frakB_i}^{a_i}=(\Z,V_{\frakB_i}^{a_i},
		A_{\frakB_i}^{a_i})$ of graph $G_{\frakB_i}$ that consists of $a_i$ and all 
		paths from $\head(a_i)$ to $\hat q_i^1$ in time polynomial in the size of $G_{\frakB_i}$.
		We next show that \eqref{eqn39_f_det_MPautomata} can be equivalently specified as
		\begin{align}\label{eqn41_det_MPautomata} 
			\bigwedge_{i\in\llb 1,m \rrb} \bigwedge_{a_i\in A_{\frakB_i}}
			\left( \begin{array}{c} y_{a_i} > 0\wedge \\ \head(a_i)\ne\hat q_i^1\end{array}
				\implies \sum_{ \substack{a\in A_{\frakB_i}^{a_i}\\ \tail(a)=\head(a_i)
			}}y_{a}>0 \right)\wedge.
		\end{align}

		For every $l\in\llb 1,m \rrb$, if $G_{\frakB_l}$ satisfies \eqref{eqn39_c_det_MPautomata},
		\eqref{eqn39_d_det_MPautomata}, and \eqref{eqn39_e_det_MPautomata} (after $\bigwedge_{i\in\llb 1,m \rrb}$ 
		was deleted), then for every $a_l\in A_{\frakB_l}$, subgraph $G_{\frakB_l}^{a_l}$
		also satisfies \eqref{eqn39_c_det_MPautomata}, \eqref{eqn39_d_det_MPautomata}, and 
		\eqref{eqn39_e_det_MPautomata} if $\tail(a_l)$ and
		$\hat q_l^1$ are regarded as the source and the sink, respectively.
		Before proceeding, we give an illustrative example. 
		\begin{example}
			Consider the directed graph in Fig.~\ref{fig20_det_MPautomata} showing a graph $G_{\frakB_i}$ for
			some $i\in\llb 1,m \rrb$.
			\begin{figure}[!htbp]
                \centering
                \begin{tikzpicture}
                	[>=stealth',shorten >=1pt,thick,auto,node distance=2.0 cm, scale = 1.0, transform shape,
                	->,>=stealth,inner sep=2pt, initial text = 0]
                
                	\tikzstyle{emptynode}=[inner sep=0,outer sep=0]
                
                	\node[state] (0) {$\hat q_i$};
                	\node[state] (1) [below of = 0] {1};
					\node[emptynode] (empty1) [below of = 1] {};
                	\node[state] (2) [left of = empty1] {2};
                	\node[state] (3) [right of = empty1] {3};
					\node[emptynode] (empty2) [left of = 2] {};
					\node[state] (4) [below of = empty2] {4};
					\node[state] (5) [right of = 4] {5};
					\node[state] (6) [right of = 5] {6};
					\node[state] (7) [right of = 6] {7};
					\node[state] (8) [right of = 7] {8};
					\node[state] (9) [below of = 6] {9};
					\node[emptynode] (empty3) [below of = 4] {};
                	\node[state] (qi1) [left of = empty3] {$\hat q_i^1$};

                	\path [->]
					(0) edge node {$b_0,1$} (1)
					(1) edge node [sloped, above] {$b_3,3$} (2)
					(1) edge node [sloped, above] {$b_4,0$} (3)
					(4) edge node [sloped, above] {$b_5,1$} (2)
					(5) edge node [sloped, above]  {$b_{9},0$} (4)
					(2) edge node [sloped, above]  {$b_{6},4$} (5)
					(5) edge node [sloped, above]  {$b_{10},4$} (6)
					(6) edge node [sloped, above]  {$b_{7},4$} (1)
					(3) edge node [sloped, above]  {$b_{8},0$} (7)
					(8) edge node [sloped, above]  {$b_{13},1$} (7)
					(1) edge [bend left] node [sloped, above] {$b_2,1$} (8)
					(1) edge [bend right] node [sloped, above] {$b_1,1$} (4)
					(5) edge node [sloped, above]  {$b_{11},0$} (9)
					(7) edge node [sloped, above]  {$b_{12},1$} (9)
					(9) edge node [sloped, above]  {$b_{14},1$} (qi1)
					;

                \end{tikzpicture}
				\caption{A directed graph illustrating a graph $G_{\frakB_i}$, where $b_k$, $0\le k\le 14$,
				denote the corresponding edges, the number after each $b_k$ denotes the valuation of $y_{b_k}$.}
				\label{fig20_det_MPautomata} 
           \end{figure}
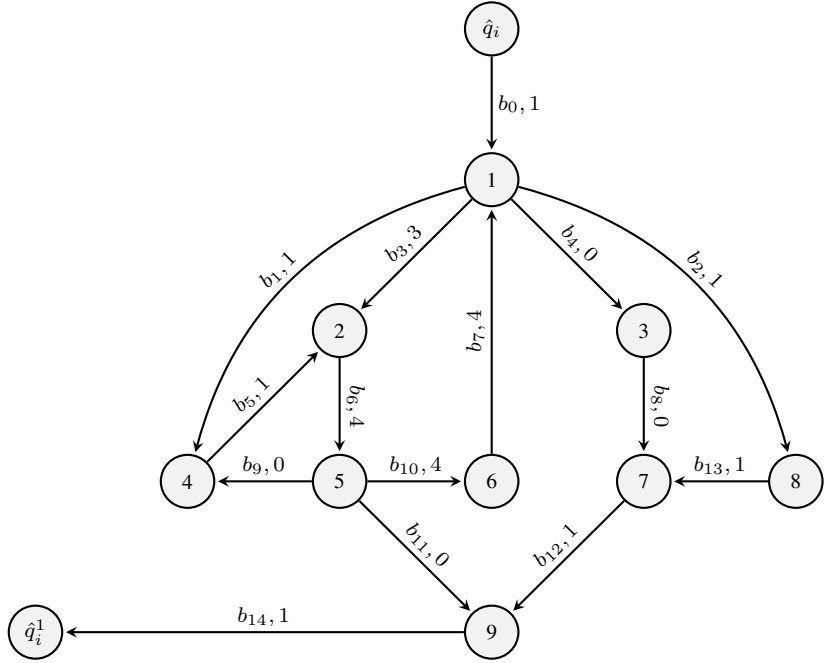
		   Note that the graph contains a cycle $1\rightarrow 2\rightarrow 5\rightarrow 6\rightarrow 1$. 
		   One easily has $G_{\frakB_i}^{b_0}=G_{\frakB_i}$. In subgraph
		   $G_{\frakB_i}^{b_0}$, one has $$\sum_{ \substack{a\in A_{\frakB_i}^{b_0}\\ \tail(a)=\head(b_0)\\
		   }} y_{a}=y_{b_1}+y_{b_3}+y_{b_4}+y_{b_2}=5>0.$$
		   Subgraph $G_{\frakB_i}^{b_1}$ can be obtained from $G_{\frakB_i}$ by only deleting $b_0$.
		   One also has $$\sum_{ \substack{a\in A_{\frakB_i}^{b_1}\\ \tail(a)=\head(b_1)\\
		   }} y_{a}=y_{b_5}=1>0.$$
		   After checking all other $b_k$ such that $y_{b_k}>0$ ($k=2,3,5,6,7,10,12,13,14$), we have
		   graph $G_{\frakB_i}$ satisfies formula \eqref{eqn41_det_MPautomata} (after $\bigwedge_{i\in\llb 1,m \rrb}$
		   was deleted). In addition, graph $G_{\frakB_i}$ satisfies \eqref{eqn39_c_det_MPautomata}, 
		   \eqref{eqn39_d_det_MPautomata}, and \eqref{eqn39_e_det_MPautomata} (after $\bigwedge_{i\in\llb 1,
		   m \rrb}$ was deleted). Moreover, for each $k$, $1\le k \le 14$, subgraph $G_{\frakB_i}^{b_k}$
		   satisfies \eqref{eqn39_c_det_MPautomata}, \eqref{eqn39_d_det_MPautomata}, and \eqref{eqn39_e_det_MPautomata}
		   if $\tail(b_k)$ and $\hat q_i^1$ are regarded as the source and the sink, respectively.
		\end{example}
		
		Now we come back to the poof of the property that \eqref{eqn39_f_det_MPautomata} can be equivalently
		specified by \eqref{eqn41_det_MPautomata}. Equivalently, we need to prove that \eqref{eqn39_c_det_MPautomata}, 
		\eqref{eqn39_d_det_MPautomata}, \eqref{eqn39_e_det_MPautomata}, and \eqref{eqn39_f_det_MPautomata}
		are all satisfied if and only if \eqref{eqn39_c_det_MPautomata}, 
		\eqref{eqn39_d_det_MPautomata}, \eqref{eqn39_e_det_MPautomata}, and \eqref{eqn41_det_MPautomata}
		are all satisfied. The ``only if'' part trivially holds. 
		To prove the ``{\bf if}'' part, we need to, for every $i\in\llb 1,m \rrb$, for
		every edge $a_i\in A_{\frakB_i}$ such that $y_{a_i}>0$ and $\head(a_i)\ne\hat q_i^1$,
		find a realizable path $\hat\pi_i$ from $\head(a_i)$
		to $\hat q_i^1$. We \textbf{claim} that for a graph
		$G_{\frakB_i}$ satisfying \eqref{eqn39_c_det_MPautomata}, \eqref{eqn39_d_det_MPautomata},
		\eqref{eqn39_e_det_MPautomata}, and \eqref{eqn41_det_MPautomata} (after $\bigwedge _{i\in\llb 1,m \rrb}$
		was deleted), $i\in\llb 1,m \rrb$, in every realizable cycle, if we subtract $y_a$ for each edge $a$
		in the cycle by $1$,
		then the rest of $G_{\frakB_i}$ (denoted by $G_{\frakB_i}'$) still satisfies \eqref{eqn39_c_det_MPautomata},
		\eqref{eqn39_d_det_MPautomata}, \eqref{eqn39_e_det_MPautomata}, and \eqref{eqn41_det_MPautomata}.
		One directly sees that if $G_{\frakB_i}$ satisfies \eqref{eqn39_c_det_MPautomata},
		\eqref{eqn39_d_det_MPautomata}, \eqref{eqn39_e_det_MPautomata}, then after doing such a subtraction,
		$G_{\frakB_i}'$ still satisfies \eqref{eqn39_c_det_MPautomata},
		\eqref{eqn39_d_det_MPautomata}, \eqref{eqn39_e_det_MPautomata}.
		As one can see, for every realizable cycle in $G_{\frakB_i}$, for every edge $a_i\in A_{\frakB_i}$, 
		either the cycle is contained in subgraph $G_{\frakB_i}^{a_i}$ or it does not intersect with 
		$G_{\frakB_i}^{a_i}$.
		
		Fix a graph $G_{\frakB_i}$, where $i\in \llb 1,m \rrb$, that
		satisfies \eqref{eqn39_c_det_MPautomata}, \eqref{eqn39_d_det_MPautomata}, \eqref{eqn39_e_det_MPautomata},
		and \eqref{eqn41_det_MPautomata} (after $\bigwedge_{i\in\llb 1,m \rrb}$ was deleted).
		Fix a realizable simple cycle $C$ and an edge $a_i\in A_{\frakB_i}^{a_i}$ such that $y_{a_i}>0$ and
		$\head(a_i)\ne\hat q_i^1$. Without loss of generality, assume $C$ is contained in $G_{\frakB_i}^{a_i}$.
		Note that $\tail(a_i)$ is the source of subgraph $G_{\frakB_i}^{a_i}$, so it satisfies
		\begin{align}\label{eqn44_det_MPautomata}
			\sum_{\substack{a\in A_{\frakB_i}^{a_i}\\ \head(a)=\tail(a_i)}} y_a+1 =
			\sum_{\substack{a\in A_{\frakB_i}^{a_i}\\ \tail(a)=\tail(a_i)}} y_a.
		\end{align}
		
		\begin{itemize}
			\item If $a_i$ does not intersect with $C$, then after doing the subtraction (on $C$), $a_i$ trivially
				satisfies \eqref{eqn41_det_MPautomata}.
			\item If $a_i$ intersects with $C$ at $\head(a_i)$ but $a_i$ is not contained in $C$, then after
				doing the subtraction, $a_i$ still satisfies \eqref{eqn41_det_MPautomata}, because 
				$\head(a_i)$ satisfies \eqref{eqn39_c_det_MPautomata}.
			\item If $a_i$ is contained in $C$ and $a_i$ is a self-loop, then $C$ coincides with $a_i$. After doing
				the subtraction we may have $y_{a_i}>0$ or $y_{a_i}=0$. In the former case, 
				$a_i$ naturally satisfies \eqref{eqn41_det_MPautomata}, in the latter case, $a_i$ satisfies
				\eqref{eqn41_det_MPautomata} vacuously, and on the other hand, because $\head(a_i)$
				is the source (satisfying \eqref{eqn44_det_MPautomata}), there is $a\in A_{\frakB_i}^{a_i}$ 
				such that $\tail(a)=\head(a_i)$ and $y_a>0$.
			\item If $a_i$ is contained in $C$ and $a_i$ is not a self-loop, after doing the subtraction,
				if $y_{a_i}>0$ then $a_i$ still satisfies \eqref{eqn41_det_MPautomata} because $\head(a_i)$ 
				satisfies \eqref{eqn39_c_det_MPautomata}, if $y_{a_i}=0$, then $a_i$ satisfies 
				\eqref{eqn41_det_MPautomata} vacuously, and there is $a\in A_{\frakB_i}^{a_i}$
				such that $\tail(a)=\tail(a_i)$ and $y_a>0$, because $\tail(a_i)$ is the source
				(satisfying \eqref{eqn44_det_MPautomata}).
		\end{itemize}
		
		We have proved the above {\bf claim}.
		We next repetitively trim $G_{\frakB_i}^{a_i}$ by doing the above subtractions on 
		realizable, simple cycles until we finally obtain a subgraph $\trim(G_{\frakB_i}^{a_i})$ in which
		there is no realized cycle. By the claim, we still have $\trim(G_{\frakB_i}^{a_i})$ satisfies 
		\eqref{eqn39_c_det_MPautomata},	\eqref{eqn39_d_det_MPautomata}, \eqref{eqn39_e_det_MPautomata}, and
		\eqref{eqn41_det_MPautomata}. In addition, we have either (1) there is an edge $b_1$ in 
		$\trim(G_{\frakB_i}^{a_i})$ such that $\tail(b_1)=\head(a_i)$ and $y_{b_1}>0$, or, (2) 
		there is an edge $b_2$ in $\trim(G_{\frakB_i}^{a_i})$ such that $\tail(b_2)=\tail(a_i)$, $y_{b_2}>0$,
		and in the original graph $G_{\frakB_i}^{a_i}$ there is a realizable path from $\head(a_i)$ to
		$\tail(a_i)$ (hence there is a realizable path from $\head(a_i)$ to $\head(b_2)$). Since in
		$\trim(G_{\frakB_i}^{a_i})$ there is no realizable cycle and $\trim(G_{\frakB_i}^{a_i})$ satisfies 
		\eqref{eqn41_det_MPautomata}, if we start from $\head(b_1)$ or $\head(b_2)$ and traverse
		realizable edges one by one, finally we will reach $\hat q_i^1$, thus we find a realizable path from
		$\head(a_i)$ to $\hat q_i^1$ in $G_{\frakB_i}^{a_i}$. The proof of the ``{\bf if}'' part has been finished.

		\begin{example}\label{exam9_MPautomata}
			Reconsider the subgraph $G_{\frakB_i}^{b_0}$ of graph $G_{\frakB_i}$ shown in
			Fig.~\ref{fig20_det_MPautomata}.
			After doing the above subtractions on its realizable cycles, we finally obtain 
			$\trim(G_{\frakB_i}^{b_0})$ in Fig.~\ref{fig21_det_MPautomata}.
			\begin{figure}[!htbp]
                \centering
                \begin{tikzpicture}
                	[>=stealth',shorten >=1pt,thick,auto,node distance=2.0 cm, scale = 1.0, transform shape,
                	->,>=stealth,inner sep=2pt, initial text = 0]
                
                	\tikzstyle{emptynode}=[inner sep=0,outer sep=0]
                
                	\node[state] (0) {$\hat q_i$};
                	\node[state] (1) [below of = 0] {1};
					\node[emptynode] (empty1) [below of = 1] {};
                	\node[state] (2) [left of = empty1] {2};
                	\node[state] (3) [right of = empty1] {3};
					\node[emptynode] (empty2) [left of = 2] {};
					\node[state] (4) [below of = empty2] {4};
					\node[state] (5) [right of = 4] {5};
					\node[state] (6) [right of = 5] {6};
					\node[state] (7) [right of = 6] {7};
					\node[state] (8) [right of = 7] {8};
					\node[state] (9) [below of = 6] {9};
					\node[emptynode] (empty3) [below of = 4] {};
                	\node[state] (qi1) [left of = empty3] {$\hat q_i^1$};

                	\path [->]
					(0) edge node {$b_0,1$} (1)
					(1) edge node [sloped, above] {$b_3,0$} (2)
					(1) edge node [sloped, above] {$b_4,0$} (3)
					(4) edge node [sloped, above] {$b_5,1$} (2)
					(5) edge node [sloped, above]  {$b_{9},0$} (4)
					(2) edge node [sloped, above]  {$b_{6},1$} (5)
					(5) edge node [sloped, above]  {$b_{10},1$} (6)
					(6) edge node [sloped, above]  {$b_{7},1$} (1)
					(3) edge node [sloped, above]  {$b_{8},0$} (7)
					(8) edge node [sloped, above]  {$b_{13},1$} (7)
					(1) edge [bend left] node [sloped, above] {$b_2,1$} (8)
					(1) edge [bend right] node [sloped, above] {$b_1,1$} (4)
					(5) edge node [sloped, above]  {$b_{11},0$} (9)
					(7) edge node [sloped, above]  {$b_{12},1$} (9)
					(9) edge node [sloped, above]  {$b_{14},1$} (qi1)
					;

                \end{tikzpicture}
				\caption{$\trim(G_{\frakB_i}^{b_0})$.}
				\label{fig21_det_MPautomata} 
           \end{figure}

		We find two realizable paths from $1$ (i.e., $\head(b_0)$) to $\hat q_i^1$. They are
		\begin{align*}
			&b_1b_5b_6b_{10}b_7b_2b_{13}b_{12}b_{14},\\
			&b_2b_{13}b_{12}b_{14}.
		\end{align*}
		It is easy to see that $\trim(G_{\frakB_i}^{b_0})$ satisfies \eqref{eqn39_c_det_MPautomata},
		\eqref{eqn39_d_det_MPautomata}, \eqref{eqn39_e_det_MPautomata}, and \eqref{eqn41_det_MPautomata}
		(after $\bigwedge_{i\in\llb 1,m \rrb}$ was deleted).
		\end{example}
		
		Now consider an arbitrary $j\in\llb 1,|\frakE| \rrb$ and an arbitrary edge $a_j\in A_{\frakE_j}$
		such that $y_{a_j}>0$ and $\head(a_j)\ne q_j^1$.
		Compute a subgraph $G_{\frakE_j}^{a_j}$ of graph $G_{\frakE_j}$ that consists of 
		$a_j$ and all 
		paths from $\head(a_j)$ to $q_j^1$ in time polynomial in the size of $G_{\frakE_j}$. Analogously, we 
		have \eqref{eqn39_h_det_MPautomata}, \eqref{eqn39_i_det_MPautomata}, \eqref{eqn39_j_det_MPautomata},
		and \eqref{eqn39_k_det_MPautomata} (after $\bigwedge_{j\in\llb 1,|\frakE| \rrb}$ was deleted)
		are all satisfied if and only if
		\eqref{eqn39_h_det_MPautomata}, \eqref{eqn39_i_det_MPautomata}, \eqref{eqn39_j_det_MPautomata}, and
		\begin{align}\label{eqn42_det_MPautomata} 
			\bigwedge_{a_j\in A_{\frakE_j}}
			\left( \begin{array}{c} y_{a_j} > 0 \wedge \\ \head(a_j)\ne q_j^1 \end{array} \implies
				\sum_{ \substack{a\in A_{\frakE_j}^{a_j}\\ \tail(a)=\head(a_j)
			}} y_{a}>0 \right) 
		\end{align}
		are all satisfied.

	\item\label{item25_det_MPautomata} 
		Now we have obtained a Presburger formula \eqref{eqn39_det_MPautomata} (after \eqref{eqn39_f_det_MPautomata} 
		and \eqref{eqn39_k_det_MPautomata} were replaced by \eqref{eqn41_det_MPautomata} and
		\eqref{eqn42_det_MPautomata}, respectively) in the form of \eqref{eqn34_det_MPautomata}. Next we show 
		the formula has length polynomial in the size of $\Acal^{\Z}$. By \eqref{eqn39_a_det_MPautomata} and
		\eqref{eqn39_b_det_MPautomata}, the formula has at most $2$ quantifier alternations. It has $2$
		quantifier alternations if and only if $\sum_{i=1}^{m}|A_{\frakB_i}|>0$ and
		$\sum_{j=1}^{|\frakE|}|A_{\frakE_j}|>0$. In \eqref{eqn39_c_det_MPautomata}, \eqref{eqn39_d_det_MPautomata}, 
		and \eqref{eqn39_e_det_MPautomata}, the number of equations is bounded from above by $\sum_{i=1}^m|V_{\frakB_i}|
		\le |\hat x_2||Q|\le |Q|^2$ (see \eqref{item23_det_MPautomata}); in \eqref{eqn39_h_det_MPautomata},
		\eqref{eqn39_i_det_MPautomata}, and \eqref{eqn39_j_det_MPautomata}, the number of equations has an upper 
		bound $\sum_{i=1}^{|\frakE|}|V_{\frakE_j}| \le |\frakE| |Q| \le |x_1||\Dt_\s|(|x_2|-|\hat x_2|) |Q|
		\le |Q|^3|\Dt|$ (also see \eqref{item23_det_MPautomata}); in \eqref{eqn39_g_det_MPautomata}, the number
		of equations is bounded from above by $m-1\le |Q|$; in \eqref{eqn39_l_det_MPautomata}, the number of
		inequalities
		is bounded from above by $|\frakE|\le |x_1||\Dt_\s|(|x_2|-|\hat x_2|) \le |Q|^2|\Dt|$; in 
		\eqref{eqn41_det_MPautomata}, the number of inequalities is no greater than $2\sum_{i=1}^{m}|A_{\frakB_i}|
		\le 2m|\Dt|\le 2|Q||\Dt|$; in \eqref{eqn42_det_MPautomata}, the number of inequalities is no greater than
		$2\sum_{j=1}^{|\frakE|}|A_{\frakE_j}|\le 2|\frakE||\Dt|\le 2|x_1||\Dt_\s|(|x_2|-|\hat x_2|) |\Dt|
		\le 2|Q|^2|\Dt|^2$. All these equations and inequalities have length polynomial in the size of $\Acal^{\Z}$.
		Hence {\bf\eqref{eqn39_det_MPautomata} has length polynomial in the size of $\Acal^{\Z}$}.

		\begin{itemize}
				\item
		When \eqref{eqn39_det_MPautomata} has $2$ quantifier alternations,
		by Lemma~\ref{lem4_det_MPautomata}, in the worst case we need to check all 
		$\sum_{i=1}^{m}|A_{\frakB_i}|+\sum_{j=1}^{|\frakE|}|A_{\frakE_j}|$ variables 
		between $0$ and $w=2^{cr^{(s+3)^3}}$ to make sure whether
		\eqref{eqn39_det_MPautomata} is satisfied, where $c$ is a constant, $r$ is the length of 
		\eqref{eqn39_det_MPautomata}, $s=\sum_{i=1}^{m}|A_{\frakB_i}|+\sum_{j=1}^{|\frakE|}|A_{\frakE_j}|$ 
		is the number of quantifiers. The logarithm of $w$ is exponential
		in the size of $\Acal^{\Z}$, and each check can be done in $\EXPTIME$, hence the satisfiability of
		\eqref{eqn39_det_MPautomata} can be checked in $2$-$\EXPTIME$ in the size of $\Acal^{\Z}$.

				\item
		When \eqref{eqn39_det_MPautomata} has $1$ quantifier alternation (in this case exactly one of 
		$\sum_{i=1}^{m}|A_{\frakB_i}|$ and $\sum_{j=1}^{|\frakE|}|A_{\frakE_j}|$ is equal to $0$),
		one has $w=2^{cr^{(s+3)^2}}$.
		Hence the satisfiability of \eqref{eqn39_det_MPautomata} can be checked also in $2$-$\EXPTIME$.

				\item
		When \eqref{eqn39_det_MPautomata} has $0$ quantifier alternation, i.e., all graphs $G_{\frakB_i}$,
		$i\in\llb 1,m \rrb$, $G_{\frakE_j}$, $j\in \llb 1,|\frakE|\rrb$, are singletons, the satisfiability of
		\eqref{eqn39_det_MPautomata} can be checked in polynomial time.
		\end{itemize}
\end{enumerate}

\begin{remark}\label{rem5_det_MPautomata}
	In \eqref{item24_det_MPautomata}, the specifications of \eqref{eqn39_f_det_MPautomata} and
	\eqref{eqn39_k_det_MPautomata} to \eqref{eqn41_det_MPautomata} and \eqref{eqn42_det_MPautomata}
	are crucial steps. If we did not do this but directly specify \eqref{eqn39_f_det_MPautomata} and
	\eqref{eqn39_k_det_MPautomata} as they are, then \eqref{eqn39_det_MPautomata} would have length 
	exponential in the size of $\Acal^{\Z}$, and then the satisfiability of \eqref{eqn39_det_MPautomata} 
	would be checked in $3$-$\EXPTIME$ in the size of $\Acal^{\Z}$ by using Lemma~\ref{lem4_det_MPautomata}. 

	Taking graph $G_{\frakB_i}^{a_i}$ for some $i\in\llb 1,m \rrb$ for example, in order to directly specify
	\eqref{eqn39_f_det_MPautomata}, we need to enumerate
	all $\alpha$ distinct (at most exponentially many) simple paths from $\head(a_i)$ to $\hat q_i^1$ as 
	\begin{align*}
		&a_i^{1,1}\dots a_i^{1,1_\alpha},\\ & \vdots \\
		&a_i^{\alpha,1}\dots a_i^{\alpha,\alpha_\alpha},
	\end{align*}
	where $\tail(a_i^{l,1})=\head(a_i)$, $\head(a_i^{l,l_\alpha})=
	\hat q_i^1$, $\head(a_i^{l,k_l})=\tail(a_i^{l,k_l+1})$, $1\le l\le \alpha$,  
	$1\le l_\alpha\le |V_{\frakB_i}^{a_i}|-1$, $1\le k_l\le l_\alpha-1$.
	\eqref{eqn39_f_det_MPautomata} is directly specified as $$\bigvee_{l=1}^{\alpha}\left(\bigwedge
	_{k_l=1}^{l_{\alpha}} \left(y_{a_i^{l,k_l}}>0\right)\right),$$
	which has length exponential in the size of $\Acal^{\Z}$.
\end{remark}

\begin{remark}\label{rem6_det_MPautomata}
	Using the above procedure of computing observer $\Acal^{\Q^k}_{obs}$ as in Items \eqref{item18_det_MPautomata}
	through \eqref{item25_det_MPautomata}, one can do state estimation based on a given weighted label sequence
	$\gamma=(\s_1,t_1)\dots(\s_n,t_n)$, where $\s_i\in \Sig$, $t_i\in\Q^k$, $i\in\llb 1,n
	\rrb$: first, compute all states $q_1$ that are reachable from some state of $Q_0$ through
	some unobservable path followed by some observable transition with an event $e_1$ in $E_{\s_1}$,
	where the entire path has weight $t$, in nondeterministic polynomial time by Lemma~\ref{lem1_det_MPautomata},
	denote the set of all such
	$q_1$ by $x_1$, and call $x_1$ the \emph{$(\s_1,t_1)$-successor} of $Q_0$; second, compute the 
	$(\s_2,t_2-t_1)$-successor $x_2$ of $x_1$; \dots; finally compute the $(\s_n,t_n-t_{n-1})$-successor
	$x_n$ of $x_{n-1}$, then one has $\Mt(\Acal^{\Q^k},\ep|x_n)=\Mt(\Acal^{\Q^k},\gamma)$.
\end{remark}

Now we analyze the complexity of computing $\Acal_{obs}^{\Z}$. For every $x_1\subset Q$ (see 
\eqref{item18_det_MPautomata}), we compute the corresponding $x_2$ in polynomial time.
For every $\bar x_2\subset x_2$ (see \eqref{item19_det_MPautomata}), we check whether $\bar x_2$ is a
successor of $x_1$ and compute the corresponding transitions starting from $x_1$ 
in $\Acal_{obs}^{\Z}$ as follows: we compute automaton
\eqref{eqn4_det_MPautomata} for at most $(|x_1||\Dt_{\s}|)^{|\bar x_2|}$ times
(see \eqref{item19_det_MPautomata} and \eqref{item9_det_MPautomata}).
Then in each automaton \eqref{eqn24_det_MPautomata} (trivially obtained from \eqref{eqn4_det_MPautomata})
(see \eqref{item10_det_MPautomata}), we check the existence of path $\bar\pi$ (see \eqref{quoteB_det_MPautomata})
and compute transitions of $\Acal_{obs}^{\Z}$ starting at $x_1$ in $\NP$ in the size of 
\eqref{eqn24_det_MPautomata} by Lemma~\ref{lem1_det_MPautomata}, where the size of \eqref{eqn24_det_MPautomata}
is exponential in the size of $\Acal^{\Z}$. Totally we check the existence of $\bar\pi$ in 
\eqref{eqn24_det_MPautomata} for at most exponentially many times. Hence this part can be finished in 
$\NEXPTIME$ in the size of $\Acal^{\Z}$.

Consider the above $x_1$ and $x_2$, for every successor $\bar x_2\subset x_2$ of $x_1$ computed as above,
for every subset $\hat x_2\subsetneq \bar x_2$, we check whether there is a transition from $x_1$ to 
$\Mt(\Acal^{\Z},\ep|\hat x_2)$ as in \eqref{item23_det_MPautomata}, \eqref{item24_det_MPautomata},
\eqref{item25_det_MPautomata}. When $x_1,x_2,\hat x_2$ are given, $\frakE$ (see \eqref{eqn38_det_MPautomata})
is also given, we need to check the satisfiability of Presburger formula \eqref{eqn39_det_MPautomata} 
by Lemma~\ref{lem4_det_MPautomata} for at most 
$(|x_1||\Dt_{\s}|)^{|\hat x_2|}$ times. Totally we do the check for at most exponentially many times.
Hence this part can be finished in $2$-$\EXPTIME$ in the size of $\Acal^{\Z}$.


\begin{theorem}\label{thm5_det_MPautomata}
	An observer ${\Acal}_{obs}^{\Q^k}$ \eqref{eqn_observer_sim_MPautomata}
	of a labeled weighted automaton $\Acal^{\Q^k}$ 
	can be computed in 
	$2$-$\EXPTIME$ in the size of $\Acal^{\Q^k}$.
\end{theorem}

Particularly, for deadlock-free and divergence-free automaton $\Acal^{\Q^k}$, in which there exists
no unobservable cycle, from the above procedure of computing 
${\Acal}{^{\Q^k}_{obs}}$, one directly sees that the EPL problem and Presburger arithmetic are not needed, but 
one only needs to enumerate all (exponentially many) unobservable paths. Hence the next direct
corollary follows.
\begin{corollary}\label{cor5_det_MPautomata}
	The observer ${\Acal}_{obs}^{\Q^k}$ \eqref{eqn_observer_sim_MPautomata}
	of a divergence-free $\Acal^{\Q^k}$ 
	can be computed in $\EXPTIME$ in the size of $\Acal^{\Q^k}$.
\end{corollary}

In Theorem~\ref{thm2_det_MPautomata}, conditions~\eqref{item1_det_MPautomata} and \eqref{item2_det_MPautomata}
can be verified in time linear in the size of $\Acal^{\Q^k}$ by computing its strongly connected
components, and condition~\eqref{item3_det_MPautomata} can be verified in time linear in the size
of any observer ${\Acal}{^{\Q^k}_{obs}}$. Then the following result holds.

\begin{theorem}\label{thm6_det_MPautomata}
	The weak detectability of a labeled weighted automaton $\Acal^{\Q^k}$ 
	can be verified in $2$-$\EXPTIME$ in the size of $\Acal^{\Q^k}$.
\end{theorem}

Similarly, by Theorem~\ref{thm7_det_MPautomata} and Theorem~\ref{thm5_det_MPautomata}, the following result
holds.

\begin{theorem}\label{thm10_det_MPautomata}
	The weak periodic detectability of a labeled weighted automaton $\Acal^{\Q^k}$ 
	can be verified in $2$-$\EXPTIME$ in the size of $\Acal^{\Q^k}$.
\end{theorem}

By Theorems~\ref{thm2_det_MPautomata} and \ref{thm7_det_MPautomata}, and Corollary~\ref{cor5_det_MPautomata},
the following result holds.
\begin{corollary}\label{cor6_det_MPautomata}
	The weak detectability and weak periodic detectability of a labeled deadlock-free, divergence-free
	weighted automaton $\Acal^{\Q^k}$ 
	can be verified in $\EXPTIME$ in the size of $\Acal^{\Q^k}$. The upper bounds also apply to
	deadlock-free and divergence-free unambiguous 
	weighted automata over semiring $\underQ$. 
\end{corollary}

When all transitions of $\Acal^{\Q^k}$ are observable, 
we have the following direct corollary.

\begin{corollary}\label{cor2_det_MPautomata}
	Consider a labeled weighted automaton $\Acal^{\Q^k}$ all of whose transitions are observable.
	Its observer ${\Acal}{^{\Q^k}_{obs}}$ can be computed in $\EXPTIME$,
	and its weak detectability and weak periodic detectability can be verified also
	in $\EXPTIME$.
\end{corollary}

\subsubsection{Computation of detector $\Acal_{det}^{\Q^k}$ and
verification of strong periodic detectability}
\label{subsubsec:detector}

One sees detector ${\Acal}_{det}^{\Q^k}$ \eqref{eqn_detector_sim_MPautomata}
is a simplified version of
observer ${\Acal}_{obs}^{\Q^k}$ \eqref{eqn_observer_sim_MPautomata},
hence ${\Acal}_{det}^{\Q^k}$ can be
computed similarly by starting from the initial state $x_0=\Mt(\Acal^{\Q^k},\ep)$,
and find all reachable states and transitions step by step, where 
the states (except for $x_0$) are of cardinality $\le2$. Similarly to \eqref{item18_det_MPautomata}
in Section~\ref{subsubsec:observer},
choose a state $x_1\in X$ of $\Acal_{det}^{\Q^k}$ that we have just computed (here $1\le |x_1|\le 2$ if $x_1\ne x_0$), 
choose $\s\in\Sig$, for each $q\in x_1$, compute subautomaton $\Acal_{q}^{\Q^k}$ that consists of all paths 
of the form 
\begin{equation}\label{eqn43_det_MPautomata}
	q\xrightarrow[]{s}q^1\xrightarrow[]{e}q^2
\end{equation}
of $\Acal^{\Q^k}$ such that $s\in(E_{uo})^*$, $e\in E_{\s}$. Denote the set of all such 
$q^2$ by $x_2$. Note that one may have $|x_2|>2$.
\begin{enumerate}[(1)]
	\item If $|\Mt(\Acal^{\Q^k},\ep|x_2)|=1$ (i.e., $\Mt(\Acal^{\Q^k},\ep|x_2)=x_2$ and $x_2$ is a singleton),
		then we find a transition $x_1\xrightarrow[]{(\s,t)} x_2$
		of $\Acal^{\Q^k}_{det}$, where $t$ can be the weight of any path from any
		$q$ in $x_1$ to the unique $q^2$ in $x_2$ as in \eqref{eqn43_det_MPautomata}.
	\item If $|\Mt(\Acal^{\Q^k},\ep|x_2)|>1$, for every $\bar x_2\subset x_2$ with $|\bar x_2|=2$,
		we check whether there exists two paths
		\begin{align*}
			 q_i \xrightarrow[]{s_i}q^1_i\xrightarrow[]{e_i}q^2_i,\quad i=1,2,
		\end{align*}
		as in \eqref{eqn43_det_MPautomata}, such that $q_i\in x_1$, $\{q_1^2,q_2^2\}=\bar x_2$,
		and the two paths have the same weight $t\in \Q^k$.
		If the answer is YES, then we find transitions $x_1\xrightarrow[]{(\s,t)}\bar x_2'$ for any 
		$\bar x_2'\subset \Mt(\Acal^{\Q^k},\ep|\bar x_2)$ with $|\bar x_2'|=2$.
		In each of these checks, we need to compute the synchronous product (see \eqref{eqn4_det_MPautomata}) 
		of two subautomata of $\Acal^{\Q^k}$ (in polynomial time), and solve a $1$-dimensional EPL problem
		in the product (by Lemma~\ref{lem1_det_MPautomata}, in $\NP$).
	\item If $|\Mt(\Acal^{\Q^k},\ep|x_2)|>1$, for every $\bar x_2\subset x_2$ with $|\bar x_2|=1$
		and $|\Mt(\Acal^{\Q^k},\ep|\bar x_2)|>1$, we find transitions $x_1\xrightarrow[]{(\s,t)} \bar x_2'$
		for any $\bar x_2'\subset\Mt(\Acal^{\Q^k},\ep|\bar x_2)$ with $|\bar x_2'|=2$,
		where $t$ can be the weight of any path from any
		$q$ in $x_1$ to the unique $q^2$ in $\bar x_2$ as in \eqref{eqn43_det_MPautomata}.
	\item If $|\Mt(\Acal^{\Q^k},\ep|x_2)|>1$, for every $\bar x_2\subset x_2$ with
		$|\Mt(\Acal^{\Q^k},\ep|\bar x_2)|=1$, we check whether 
		\begin{equation*}
			\tag{G}\label{quoteG_det_MPautomata} 
			\parbox{\dimexpr\linewidth-4em}{%
			\strut
				there exists a path
				\begin{align*}
					\hat q \xrightarrow[]{\hat s}\hat q^1\xrightarrow[]{\hat e}\hat q^2
				\end{align*}
				as in \eqref{eqn43_det_MPautomata}, such that $\hat q\in x_1$, $\{\hat q^2\}=\bar x_2$, the weight
				of the path is denoted by $t\in \Q^k$; and for any $\bar q\in x_1$, for any $q^2\in x_2\setminus 
				\bar x_2$, there is no
				path 
				\begin{align*}
					q \xrightarrow[]{s}q^1\xrightarrow[]{e}q^2
				\end{align*}
				as in \eqref{eqn43_det_MPautomata} with weight $t$.
			\strut
			}
		\end{equation*}
		If the answer is YES, then we find a transition $x_1\xrightarrow[]{(\s,t)}\bar x_2$ of $\Acal^{\Q^k}_{det}$.
		We need to transform the satisfiability of \eqref{quoteG_det_MPautomata} to satisfiability of some
		Presburger formula
		as in \eqref{item23_det_MPautomata}, \eqref{item24_det_MPautomata}, \eqref{item25_det_MPautomata}.
		By Lemma~\ref{lem4_det_MPautomata}, such checks can be done in $2$-$\EXPTIME$.
\end{enumerate}

\begin{theorem}\label{thm13_det_MPautomata}
	Consider a labeled weighted automaton $\Acal^{\Q^k}$. An detector $\Acal^{\Q^k}_{det}$ can be computed in 
	$2$-$\EXPTIME$ (in the size of $\Acal^{\Q^k}$); particularly if for each state $q$ of $\Acal^{\Q^k}$,
	$|\Mt(\Acal^{\Q^k},\ep|\{q\})|>1$ (in this case, one need not check the satisfiability
	of \eqref{quoteG_det_MPautomata}), then $\Acal^{\Q^k}_{det}$ can be computed in $\NP$.
\end{theorem}

By Theorem~\ref{thm9_det_MPautomata}, the following result holds.

\begin{theorem}\label{thm11_det_MPautomata}
	Consider a labeled weighted automaton $\Acal^{\Q^k}$.
	Its strong periodic detectability  can be verified in $2$-$\EXPTIME$. Particularly,
	if for each state $q$ of $\Acal^{\Q^k}$,
	$|\Mt(\Acal^{\Q^k},\ep|\{q\})|>1$, then its strong periodic detectability can be
	checked in $\NP$.
\end{theorem}

By Theorem~\ref{thm8_det_MPautomata} and Corollary~\ref{cor5_det_MPautomata}, the following two results
hold.
\begin{theorem}\label{thm14_det_MPautomata}
	The problem of verifying strong periodic detectability of a deadlock-free and divergence-free
	$\Acal^{\Q^k}$ belongs to $\EXPTIME$.
\end{theorem}

\begin{corollary}\label{cor8_det_MPautomata}
	The detector ${\Acal}_{det}^{\Q^k}$ \eqref{eqn_detector_sim_MPautomata}
	of a divergence-free $\Acal^{\Q^k}$ 
	can be computed in $\EXPTIME$ in the size of $\Acal^{\Q^k}$.
\end{corollary}

Similarly to the case of $\CCa(\Acal^{\Q^k})$, we have the following direct corollary.

\begin{corollary}\label{cor3_det_MPautomata}
	Consider a labeled weighted automaton $\Acal^{\Q^k}$ all of whose transitions are observable.
	Its detector ${\Acal}{^{\Q^k}_{det}}$ can be computed in polynomial time,
	and its strong periodic detectability can be verified also in polynomial time.
\end{corollary}

For a labeled finite-state automaton $\Acal$, one directly sees from the above procedure of computing 
${\Acal}{^{\Q^k}_{det}}$ that detector $\Acal_{det}$ can be computed in polynomial time,
because all unobservable transitions of $\Acal$ have weight $0$ as shown in Remark~\ref{rem1_det_MPautomata}.
Then the following corollary holds. 
\begin{corollary}\label{cor7_det_MPautomata}
	The strong periodic detectability of a labeled finite-state automaton $\Acal$ can be verified in polynomial
	time.
\end{corollary}

\subsubsection{The complexity lower bounds on verifying strong (periodic) detectability of labeled weighted
automaton $\Acal^{\N}$ and $\Acal^{\underN}$}

In this subsection, we prove $\coNP$ lower bounds on verifying strong detectability 
and strong periodic detectability of labeled deterministic weighted automata over monoid $(\N,+,0)$.
These complexity lower bounds also apply to strong detectability 
and strong periodic detectability of labeled deterministic weighted automata over 
semiring $\underN$, because deterministic automata are unambiguous, and
the definitions of detectability 
for unambiguous $\Acal^{\underQ}$ in \cite{Lai2021DetUnambiguousWAutomata} coincide with
those for unambiguous $\Acal^{\Q}$ in the current paper (except for minor and neglectable differences,
see Remark~\ref{rem2_det_MPautomata}).
As a result, deterministic $\Acal^{\N}$ and $\Acal^{\underN}$ are fundamentally more
complicated than labeled finite-state automata, because it is known that strong (periodic) detectability of
labeled finite-state automata can be verified in polynomial time \cite{Shu2011GDetectabilityDES,Zhang2019KDelayStrDetDES}.

\begin{theorem}\label{thm12_det_MPautomata}
	The problems of verifying strong detectability and strong periodic detectability of labeled 
	deterministic, deadlock-free, and divergence-free weighted automaton $\Acal^{\N}$ and
	$\Acal^{\underN}$ 
	are both $\coNP$-hard.
\end{theorem}

\begin{proof}
	We reduce the $\NP$-complete subset sum problem (Problem~\ref{prob2_det_MPautomata}) to
	negation of strong detectability and strong periodic detectability of labeled deterministic
	weighted automata over $\N$ (hence also over $\underN$).

	Given positive integers $n_1,\dots,n_m$, and $N$, next we construct in polynomial time a
	labeled deterministic weighted automaton $\Acal_2^{\N}$
	as illustrated in Fig.~\ref{fig1_det_MPautomata}. Apparently, $\Acal_2^{\N}$ is deadlock-free and
	divergence-free. $q_0$ is the unique
	initial state and has initial time delay $0$. Events $u_1,u_2$ are unobservable.
	Event $e$ is observable
	and $\ell(e)=e$. For all $i\in\llb 0,m-1\rrb$, there exist two unobservable transitions 
	$q_{i}\xrightarrow[]{u_1/n_{i+1}}q_{i+1}$ and $q_{i}\xrightarrow[]{u_2/0}q_{i+1}$.
	The observable transitions are $q_{m}\xrightarrow[]{e/1}q_{m+1}^1$,
	$q_{0}\xrightarrow[]{e/N+1}q_{m+1}^2$, and two self-loops $q_{m+1}^1\xrightarrow[]{e/1}q_{m+1}^1$
	and $q_{m+1}^2\xrightarrow[]{e/1}q_{m+1}^2$.

	\begin{figure}[!htbp]
        \centering
\begin{tikzpicture}
	[>=stealth',shorten >=1pt,thick,auto,node distance=1.5 cm, scale = 1.0, transform shape,
	->,>=stealth,inner sep=2pt, initial text = 0]

	\tikzstyle{emptynode}=[inner sep=0,outer sep=0]

	\node[initial, state] (q0) {$q_0$};
	\node[state] (q1) [above right of = q0] {$q_1$};
	\node[state] (q2) [right of = q1] {$q_2$};
	\node[emptynode] (e1) [right of = q2] {$\cdots$};
	\node[state] (qm-1) [right of = e1] {$q_{m-1}$};
	\node[state] (qm) [right of = qm-1] {$q_m$};
	\node[state] (qm+11) [right of = qm] {$q_{m+1}^1$};
	\node[state] (qm+12) [below of = qm+11] {$q_{m+1}^2$};

	\path [->]
	(q0) edge [bend left] node [above, sloped] {$u{_1}/n_1$} (q1)
	(q0) edge [bend right] node [below, sloped] {$u{_2}/0$} (q1)
	(q1) edge [bend left] node [above, sloped] {$u{_1}/n_2$} (q2)
	(q1) edge [bend right] node [below, sloped] {$u{_2}/0$} (q2)
	(qm-1) edge [bend left] node [above, sloped] {$u{_1}/n_m$} (qm)
	(qm-1) edge [bend right] node [below, sloped] {$u{_2}/0$} (qm)
	(qm) edge node [above, sloped] {$e/1$} (qm+11)
	(qm+11) edge [loop right] node {$e/1$} (qm+11)
	(qm+12) edge [loop right] node {$e/1$} (qm+12)
	;
	
	\draw [->]
	(q0) .. controls (1,-0.5) .. node [above, sloped] {$e/N+1$} (qm+12)
	;

        \end{tikzpicture}
		\caption{Sketch of the reduction in the proof of Theorem \ref{thm12_det_MPautomata}.}
		\label{fig1_det_MPautomata}
	\end{figure}
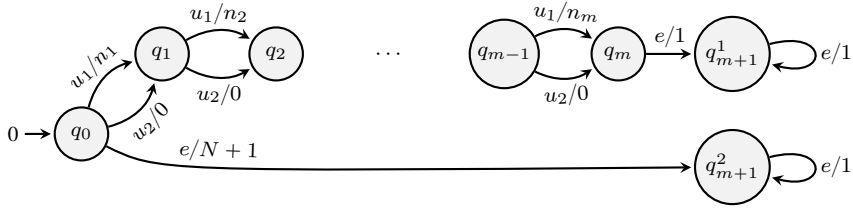

	Suppose there exists $I\subset\llb 1,m\rrb$ such that $N=\sum_{i\in I}n_i$.
	Then there is an unobservable path
	$\pi\in q_0\rightsquigarrow q_m$ whose weight is equal to $N$. Then we have 
	\begin{align}
		\ell(\tau(\pi\xrightarrow[]{e}q_{m+1}^1)) &= (e,N+1),\label{eqn9_det_MPautomata}\\
		\Mt(\Acal_2^{\N},(e,N+1)\dots(e,N+k)) &= \{q_{m+1}^1,q_{m+1}^2\}\label{eqn10_det_MPautomata}
	\end{align}
	for all $k\in\Z_+$.

	Choose \[w=\tau(\pi\xrightarrow[]{e}q_{m+1}^1(\xrightarrow[]{e}q_{m+1}^1)^{\omega})\in
	L^{\omega}(\Acal_2^\N).\]
	Then $$\ell(w)=(e,N+1)(e,N+2)\dots.$$

	Choose prefix $\gamma_k=(e,N+1)\dots(e,N+k)\sqsubset\ell(w)$. Then we have $|\gamma_k|\ge k$ and
	$|\Mt(\Acal_2^{\N},
	\gamma_k)|>1$ by \eqref{eqn10_det_MPautomata}. Hence $\Acal_2^{\N}$ is not strongly detectable.

	For all $k\in\N$, choose the above $w$, choose $w'=\tau(\pi\xrightarrow[]{e}q_{m+1}^1)\sqsubset w$,
	for all $w''$ such that $w'w''\sqsubset w$ and $|\ell(w'')|<k$, 
	we have $|\Mt(\Acal_2^{\N},\ell(w'w''))|>1$ by \eqref{eqn10_det_MPautomata}.
	Hence $\Acal_2^{\N}$ is not strongly periodically detectable.

	Suppose for all $I\subset\llb 1,m\rrb$, $N\ne\sum_{i\in I}n_i$. Then for all $\pi'\in q_0\rightsquigarrow
	q_m$,
	one has 
	\begin{align*}
		\ell(\tau(\pi'\xrightarrow[]{e}q_{m+1}^1)) &= (e,N'+1)\text{ for some }N'\ne N,\\
		\Mt(\Acal_2^{\N},\ell(\tau(\pi'\xrightarrow[]{e}q_{m+1}^1))) &= \{q_{m+1}^1\},\\
		\Mt(\Acal_2^{\N},(e,N'+1)\dots(e,N'+k)) &= \{q_{m+1}^1\},\\
		\Mt(\Acal_2^{\N},(e,N+1)\dots(e,N+k)) &= \{q_{m+1}^2\}
	\end{align*}
	for all $k\in\Z_+$. Hence $\Acal_2^{\N}$ is strongly detectable and strongly periodically detectable.
\end{proof}

The next corollary directly 
follows from Theorem~\ref{thm4_det_MPautomata} and Theorem~\ref{thm12_det_MPautomata}.

\begin{corollary}\label{cor4_det_MPautomata}
	The problems of verifying strong detectability of labeled unambiguous
	weighted automaton $\Acal^{\Q^k}$ and $\Acal^{\underQ}$ 
	are both $\coNP$-complete. 
\end{corollary}

\subsection{Illustrative examples}

In this subsection, we illustrate how to use Theorem~\ref{thm1_det_MPautomata}, 
Theorem~\ref{thm9_det_MPautomata}, Theorem~\ref{thm2_det_MPautomata}, and Theorem~\ref{thm7_det_MPautomata}
to verify strong (periodic) detectability and weak (periodic) detectability of labeled weighted
automata over monoid $(\N,+,0)$ and labeled unambiguous weighted automata over semiring $\underN$. 

\begin{example}\label{exam8_MPautomata} 
	Reconsider labeled unambiguous weighted automaton $\Acal_0^{\N}$ (the same as $\Acal_0^{\underN}$)
	in Fig.~\ref{fig10_det_MPautomata}.
	Its detectability cannot be verified by using the method developed in 
	\cite{Lai2021DetUnambiguousWAutomata} (see Remark~\ref{rem7_det_MPautomata}). Next we show how to verify
	its detectability by using the methods proposed in the current paper. 

	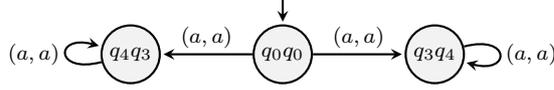
\begin{figure}[!htbp]
        \centering
	\begin{tikzpicture}
	[>=stealth',shorten >=1pt,thick,auto,node distance=2.0 cm, scale = 1.0, transform shape,
	->,>=stealth,inner sep=2pt]

	\tikzstyle{emptynode}=[inner sep=0,outer sep=0]

	\node[initial, state, initial where = above] (00) {$q_0q_0$};
	\node[state] (34) [right of = 00] {$q_3q_4$};
	\node[state] (43) [left of = 00] {$q_4q_3$};

	\path [->]
	(00) edge node [above, sloped] {$(a,a)$} (34)
	(00) edge node [above, sloped] {$(a,a)$} (43)
	(34) edge [loop right] node {$(a,a)$} (34)
	(43) edge [loop left] node {$(a,a)$} (43)
	;

    \end{tikzpicture}
	\caption{Self-composition $\CCa(\Acal_0^{\N})$ of the automaton $\Acal_0^{\N}$ in 
	Fig.~\ref{fig10_det_MPautomata}.}
	\label{fig12_det_MPautomata} 
	\end{figure} 

	Its self-composition $\CCa(\Acal_0^{\N})$ is shown in Fig.~\ref{fig12_det_MPautomata}.
	The self-loops on $(q_3,q_4)$ and on $(q_4,q_3)$ are easy to find. We use the method developed
	in Section~\ref{subsubsec:computation_self_composition} to check whether there exists a transition
	$((q_0,q_0),(a,a),(q_3,q_4))$: (1) Guess transitions $q_1\xrightarrow[]{a/1}q_3$ and $q_2\xrightarrow[]
	{a/1}q_4$ of $\Acal_0^{\N}$. (2) Compute subautomaton $\Acal_{q_0}^{\N}$ as in 
	Fig.~\ref{fig13_det_MPautomata} and asynchronous product $\Acal_{q_0}^{\N}\otimes\Acal_{q_0}^{\N}$
	as in Fig.~\ref{fig14_det_MPautomata} and check in weighted directed graph $\Acal_{q_0}^{\N}\otimes
	\Acal_{q_0}^{\N}$, whether there exists a path from $(q_0,q_0)$ to $(q_1,q_2)$ with weight $0$.
	By using the solution to the $1$-dimensional EPL problem given in \cite{Nykanen2002ExactPathLength}, 
	we can find several such paths, e.g., $(q_0,q_0)\left( \xrightarrow[]{-1} (q_0,q_2) \right)
	^{9}\xrightarrow[]{10} (q_1,q_2)$.

	In $\CCa(\Acal_0^{\N})$, there exists a self-loop on the reachable state $(q_3,q_4)$,
	then a transition sequence as in \eqref{eqn2_3_det_MPautomata} exists. In addition, in 
	$\Acal_0^{\N}$, there exists a self-loop on $q_3$, then \eqref{eqn2_2_det_MPautomata}
	of Theorem~\ref{thm1_det_MPautomata}
	holds. By Theorem~\ref{thm1_det_MPautomata}, $\Acal_0^{\N}$ is not strongly detectable.
	By definition, we choose infinite path $q_0\xrightarrow[]{u}q_1\left( \xrightarrow[]{a}q_3 \right)^{\omega}$,
	and choose path $\pi_n=q_0\xrightarrow[]{u}q_1\left( \xrightarrow[]{a}q_3 \right)^{n}$, then
	we have $\ell(\tau(\pi_n))=(a,11)\dots(a,10+n)$, and $\Mt(\Acal_0^{\N},\ell(\tau(\pi_n)))=
	\{q_3,q_4\}$, hence we also have $\Acal_0^{\N}$ is not strongly detectable.

	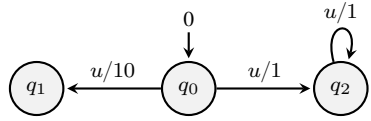
\begin{figure}[!htbp]
        \centering
	\begin{tikzpicture}
	[>=stealth',shorten >=1pt,thick,auto,node distance=2.0 cm, scale = 1.0, transform shape,
	->,>=stealth,inner sep=2pt, initial text = 0]

	\tikzstyle{emptynode}=[inner sep=0,outer sep=0]

	\node[initial, state, initial where = above] (q0) {$q_0$};
	\node[state] (q2) [right of = q0] {$q_2$};
	\node[state] (q1) [left of = q0] {$q_1$};

	\path [->]
	(q0) edge node [above, sloped] {$u/10$} (q1)
	(q0) edge node [above, sloped] {$u/1$} (q2)
	(q2) edge [loop above] node [above, sloped] {$u/1$} (q2)
	;

    \end{tikzpicture}
	\caption{Subautomaton $\Acal_{q_0}^{\N}$ of the automaton $\Acal_0^{\N}$ in 
	Fig.~\ref{fig10_det_MPautomata}.}
	\label{fig13_det_MPautomata} 
	\end{figure}

	\begin{figure}[!htbp]
        \centering
	\begin{tikzpicture}
	[>=stealth',shorten >=1pt,thick,auto,node distance=2.7 cm, scale = 1.0, transform shape,
	->,>=stealth,inner sep=2pt, initial text = 0]

	\tikzstyle{emptynode}=[inner sep=0,outer sep=0]

	\node[initial, state, initial where = left] (00) {$q_0q_0$};
	\node[state] (10) [above left of = 00] {$q_1q_0$};
	\node[state] (01) [above right of = 00] {$q_0q_1$};
	\node[state] (11) [above right of = 10] {$q_1q_1$};
	\node[state] (02) [below left of = 00] {$q_0q_2$};
	\node[state] (20) [below right of = 00] {$q_2q_0$};
	\node[state] (12) [above left of = 02] {$q_1q_2$};
	\node[state] (22) [below right of = 02] {$q_2q_2$};
	\node[state] (21) [above right of = 20] {$q_2q_1$};

	\path [->]
	(00) edge node [above, sloped] {$(\ep,u)/-10$} (01)
	(00) edge node [above, sloped] {$(u,\ep)/10$} (10)
	
	(00) edge node [above, sloped] {$(u,\ep)/1$} (20)
	(10) edge node [above, sloped] {$(\ep,u)/-10$} (11)
	(01) edge node [above, sloped] {$(u,\ep)/10$} (11)
	
	(20) edge [loop below] node {$(u,\ep)/1$} (20)
	
	(02) edge node [above, sloped] {$(u,\ep)/1$} (22)
	(20) edge node [above, sloped] {$(\ep,u)/-1$} (22)
	(20) edge node [above, sloped] {$(\ep,u)/-10$} (21)
	(10) edge node [above, sloped] {$(\ep,u)/-1$} (12)
	(01) edge node [above, sloped] {$(u,\ep)/1$} (21)
	(12) edge [loop below] node {$(\ep,u)/-1$} (12)
	(22) edge [loop below] node {$\begin{matrix}(u,\ep)/1\\(\ep,u)/-1\end{matrix}$} (22)
	(21) edge [loop below] node {$(u,\ep)/1$} (21)
	;

	\path [->, very thick]
	(00) edge node [above, sloped] {$(\ep,u)/-1$} (02)
	(02) edge [loop below] node {$(\ep,u)/-1$} (02)
	(02) edge node [above, sloped] {$(u,\ep)/10$} (12)
	;

    \end{tikzpicture}
	\caption{Asynchronous product $\Acal_{q_0}^{\N}\otimes\Acal_{q_0}^{\N}$ of subautomaton
	$\Acal_{q_0}^{\N}$ (in Fig.~\ref{fig13_det_MPautomata})
	of the automaton $\Acal_0^{\N}$ in Fig.~\ref{fig10_det_MPautomata}.}
	\label{fig14_det_MPautomata} 
	\end{figure}
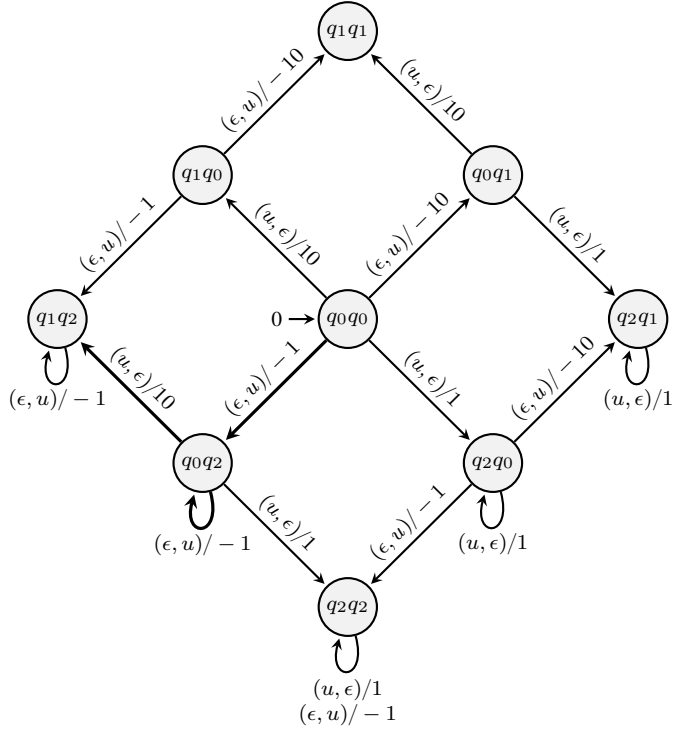

	One observer ${\Acal}{_{0obs}^{\N}}$ (also detector ${\Acal}{_{0det}^{\N}}$) is shown in 
	Fig.~\ref{fig15_det_MPautomata}. We use the method developed in Section~\ref{subsubsec:observer} to check
	whether $\{q_0\}\xrightarrow[]{(a,11)}\{q_3,q_4\}$ and $\{q_0\}\xrightarrow[]{(a,2)}\{q_4\}$ are
	transitions of ${\Acal}{_{0obs}^{\N}}$. For the former, by the path $(q_0,q_0)\left( \xrightarrow[]{-1}
	(q_0,q_2) \right)^{9}\xrightarrow[]{10} (q_1,q_2)$ with weight $0$ in Fig.~\ref{fig14_det_MPautomata} 
	and transitions $q_1\xrightarrow[]{a/1}q_3$ and $q_2\xrightarrow[]{a/1}q_4$, we know that $\{q_3,q_4\}$
	is a successor of $\{q_0\}$ (see \eqref{item9_det_MPautomata} in Section~\ref{subsubsec:observer}), 
	and then obtain the transition $\{q_0\}\xrightarrow[]{(a,11)}\Mt(\Acal_0^{\N},\ep|\{q_3,q_4\})$, 
	where $\Mt(\Acal_0^{\N},\ep|\{q_3,q_4\})=\{q_3,q_4\}$; for the latter, one sees
	a transition sequence $q_0\xrightarrow[]{u/1}q_2\xrightarrow[]{a/1}q_4$ with weight $2$, but the unique
	transition sequence from $q_0$ to $q_3$ is $q_0\xrightarrow[]{u/10}q_1\xrightarrow[]{a/1}q_3$ having weight
	$11$ that is not equal to $2$, so $\{q_0\}\xrightarrow[]{(a,2)}\{q_4\}$ is also a transition.
	On the other hand, since there is 
	a transition sequence $q_0\left(\xrightarrow[]{u/1}q_2\right)^{10}\xrightarrow[]{a/1}q_4$ having the same weight as
	$q_0\xrightarrow[]{u/10}q_1\xrightarrow[]{a/1}q_3$, there is no transition from $\{q_0\}$ to $\Mt(\Acal_0^{\N},
	\ep|\{q_3\})=\{q_3\}$.
	In observer ${\Acal}{_{0obs}^{\N}}$ (also detector ${\Acal}{_{0det}
	^{\N}}$), there is a self-loop on a reachable state 
	$\{q_3,q_4\}$ of cardinality $2$, then \eqref{item17_det_MPautomata} of Theorem~\ref{thm9_det_MPautomata} holds,
	hence $\Acal_0^{\N}$ is not strongly periodically detectable. Consider infinite path 
	$\pi=q_0\left( \xrightarrow[]{u}q_2 \right)^{\omega}$, one has $\tau(\pi)=(u,1)(u,2)\dots$, and $\ell(\tau
	(\pi))=\ep\in(\{a\}\times\N)^*$, then \eqref{item2_det_MPautomata} of Theorem~\ref{thm2_det_MPautomata}
	holds, and $\Acal_0^{\N}$ is weakly detectable. Moreover, one has $\Mt(\Acal_0^{\N},\ep)=\{q_0\}$,
	then \eqref{item12_det_MPautomata} of Theorem~\ref{thm7_det_MPautomata} holds, and $\Acal_0^{\N}$
	is weakly periodically detectable.

	\begin{figure}[!htbp]
        \centering
	\begin{tikzpicture}
	[>=stealth',shorten >=1pt,thick,auto,node distance=2.0 cm, scale = 1.0, transform shape,
	->,>=stealth,inner sep=2pt]

	\tikzstyle{emptynode}=[inner sep=0,outer sep=0]

	\node[initial, state, initial where = above] (0) {$q_0$};
	\node[state] (34) [right of = 00] {$q_3,q_4$};
	\node[state] (4) [left of = 00] {$q_4$};

	\path [->]
	(00) edge node [above, sloped] {$(a,11)$} (34)
	(34) edge [loop right] node {$(a,1)$} (34)
	(00) edge node [above, sloped] {$(a,2)$} (4)
	(4) edge [loop left] node {$(a,1)$} (4)
	;

    \end{tikzpicture}
	\caption{One observer ${\Acal}{_{0obs}^{\N}}$ (also detector ${\Acal}{_{0det}
	^{\N}}$) of the automaton $\Acal_0^{\N}$ in Fig.~\ref{fig10_det_MPautomata}.}
	\label{fig15_det_MPautomata} 
	\end{figure}
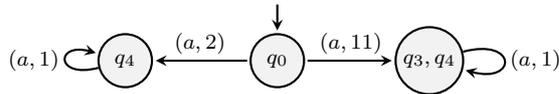 
\end{example}

\begin{example}\label{exam7_MPautomata}
	Reconsider automaton $\Acal_1^{\N}$ shown in Fig.~\ref{fig6_det_MPautomata}.
	
	In its self-composition
	shown in Fig.~\ref{fig7_det_MPautomata}, there exists a unique cycle, i.e., a self-loop on state
	$(q_4,q_4)$; there exists no state of the form $(q,q')$ satisfying $q\ne q'$ reachable from the 
	unique cycle. Then there exists no transition sequence shown in \eqref{eqn2_3_det_MPautomata}.
	By Theorem~\ref{thm1_det_MPautomata}, $\Acal_1^{\N}$ is strongly detectable. By definition,
	$\Acal_1^{\N}$ is strongly detectable by choosing $k=2$.

	In one of its detectors obtained from Fig.~\ref{fig8_det_MPautomata} by changing $\Z_{+}$ to $1$,
	there exists a reachable state
	$\{q_1,q_2\}$ satisfying $|\{q_1,q_2\}|>1$, and in $\Acal_1^{\N}$, there exists an unobservable
	self-loop on $q_1$. Then \eqref{item16_det_MPautomata} in Theorem~\ref{thm9_det_MPautomata} is satisfied,
	so $\Acal_1^{\N}$ is not strongly periodically detectable. By definition, for all $k\in\N$,
	choose $w_k=\tau\left(q_0\xrightarrow[]{a}q_1\left( \xrightarrow[]{u}q_1 \right)^{\omega}\right)=(a,1)(u,2)
	(u,3)\dots$, $w'=(a,1)\sqsubset w_k$; for all $w''=(u,2)\dots(u,i)\sqsubset (u,2)(u,3)\\\dots$, one 
	has $w'w''\sqsubset w_k$, $\ell(w'')=\ep$, and $\Mt(\Acal_1^{\N},\ell(w'w''))=\Mt(\Acal_1^{\N},
	(a,1))=\{q_1,q_2\}$, so $\Acal_1^{\N}$ is not strongly periodically detectable. 

	In one of its observers also obtained from Fig.~\ref{fig8_det_MPautomata} by changing $\Z_{+}$ to $1$,
	one sees that 
	\eqref{item3_det_MPautomata} in Theorem~\ref{thm2_det_MPautomata} is satisfied, so $\Acal_1^{\N}$
	is weakly detectable. On the other hand, we have $(a,1)(u,2)(u,3)\dots$$\in L^{\omega}(\Acal_1^\N)$
	such that $\ell((a,1)(u,2)(u,3)\dots)=(a,1)\in(\Sig\times\N)^+$,
	\eqref{item2_det_MPautomata} of Theorem~\ref{thm2_det_MPautomata} is also satisfied, one then also has 
	$\Acal_1^{\N}$ is weakly detectable. Similarly, \eqref{item13_det_MPautomata} in 
	Theorem~\ref{thm7_det_MPautomata} is satisfied, so $\Acal_1^{\N}$ is weakly periodically detectable.
\end{example}

\begin{example}
Reconsider automaton $\Acal^{\N}_2$ in the proof of
Theorem~\ref{thm12_det_MPautomata} (in Fig.~\ref{fig1_det_MPautomata}).

Assume that the subset sum problem has a solution, that is, there exists $I\subset\llb 1,m\rrb$ such that 
$N=\sum_{i\in I}n_i$. We consider $m>1$. 

The self-composition of $\Acal^{\N}_2$ is shown in Fig.~\ref{fig2_det_MPautomata}. The initial
state $(q_0,q_0)$ transitions to a self-loop on $(q_{m+1}^1,q_{m+1}^2)$ and then also to 
the state $(q_{m+1}^1,q_{m+1}^2)$ such that $q_{m+1}^1\ne q_{m+1}^2$. In $\Acal^{\N}_2$,
there is a self-loop on $q_{m+1}^1$. Then by Theorem~\ref{thm1_det_MPautomata}, $\Acal^{\N}_2$ is 
not strongly detectable.

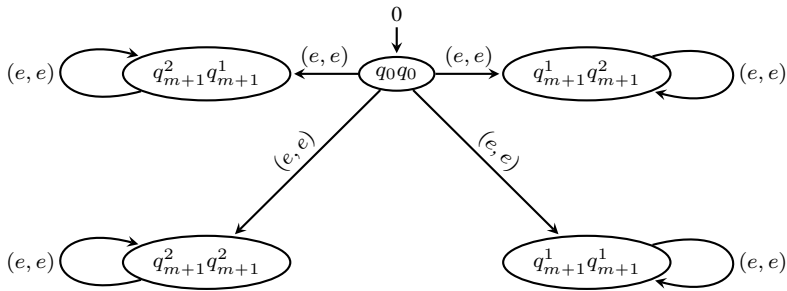
\begin{figure}[!htbp]
        \centering
	\begin{tikzpicture}
	[>=stealth',shorten >=1pt,thick,auto,node distance=2.5 cm, scale = 1.0, transform shape,
	->,>=stealth,inner sep=2pt, initial text = 0]

	\tikzstyle{emptynode}=[inner sep=0,outer sep=0]
	
	\node[initial, initial where = above, elliptic state] (q0) {$q_0q_0$};
	\node[elliptic state] (q12) [right of = q0] {$q_{m+1}^1q_{m+1}^2$};
	\node[elliptic state] (q21) [left of = q0] {$q_{m+1}^2q_{m+1}^1$};
	\node[elliptic state] (q11) [below of = q12] {$q_{m+1}^1q_{m+1}^1$};
	\node[elliptic state] (q22) [below of = q21] {$q_{m+1}^2q_{m+1}^2$};

	\path [->]
	(q0) edge node [above, sloped] {$(e,e)$} (q12)
	(q0) edge node [above, sloped] {$(e,e)$} (q21)
	(q12) edge [loop right] node [right] {$(e,e)$} (q12)
	(q21) edge [loop left] node [left] {$(e,e)$} (q21)
	(q0) edge node [above, sloped] {$(e,e)$} (q11)
	(q0) edge node [above, sloped] {$(e,e)$} (q22)
	(q11) edge [loop right] node [right] {$(e,e)$} (q11)
	(q22) edge [loop left] node [left] {$(e,e)$} (q22)
	
	;
	\end{tikzpicture}
	\caption{Self-composition of $\Acal^{\N}_2$ in Fig.~\ref{fig1_det_MPautomata} when the 
	subset sum problem has a solution.}
	\label{fig2_det_MPautomata}
\end{figure}

The observers (also the detectors) of $\Acal^{\N}_2$ are shown in Fig.~\ref{fig4_det_MPautomata}. 
The initial state $x_0$ transitions to a self-loop on $\{q_{m+1}^1,q_{m+1}^2\}$, where
$\{q_{m+1}^1,q_{m+1}^2\}$ is of cardinality $2$,
then \eqref{item15_det_MPautomata} of Theorem~\ref{thm8_det_MPautomata} and \eqref{item17_det_MPautomata}
of Theorem~\ref{thm9_det_MPautomata} are satisfied.
By Theorem~\ref{thm8_det_MPautomata} or Theorem~\ref{thm9_det_MPautomata}, $\Acal^{\N}_2$
is not strongly periodically detectable. On the other hand, $x_0$ transitions to a self-loop on
$\{q_{m+1}^1\}$, where $\{q_{m+1}^1\}$ is of cardinality $1$, then \eqref{item3_det_MPautomata}
of Theorem~\ref{thm2_det_MPautomata} is satisfied. 
By Theorem~\ref{thm2_det_MPautomata}, $\Acal^{\N}_2$
is weakly detectable. In addition, \eqref{item13_det_MPautomata} of Theorem~\ref{thm7_det_MPautomata}
is satisfied. By Theorem~\ref{thm7_det_MPautomata}, $\Acal^{\N}_2$
is weakly periodically detectable.

\begin{figure}[!htbp]
        \centering
	\begin{tikzpicture}
	[>=stealth',shorten >=1pt,thick,auto,node distance=2.8 cm, scale = 1.0, transform shape,
	->,>=stealth,inner sep=2pt, initial text = 0]
	\tikzstyle{emptynode}=[inner sep=0,outer sep=0]
	
	\node[initial, initial where = above, elliptic state] (q0) {$x_0$};
	\node[elliptic state] (q1) [right of = q0] {$q_{m+1}^1$};
	\node[elliptic state] (q12) [left of = q0] {$q_{m+1}^1,q_{m+1}^2$};

	\path [->]
	(q0) edge node [above, sloped] {$(e,U+1)$} (q1)
	(q0) edge node [above, sloped] {$(e,N+1)$} (q12)
	(q1) edge [loop right] node [right] {$(e,1)$} (q1)
	(q12) edge [loop left] node [left] {$(e,1)$} (q12)
	
	;
	\end{tikzpicture}
	\caption{Observers (also detectors) of $\Acal^{\N}_2$ in Fig.~\ref{fig1_det_MPautomata} when the 
	subset sum problem has a solution, where $x_0=\Mt(\Acal^{\N}_2,\ep)=\{q_0,\dots,q_m\}$, $U$ can be the sum
	of elements of any subset of $\{n_1,\dots,n_m\}$ different from $N$.}
	\label{fig4_det_MPautomata}
\end{figure}
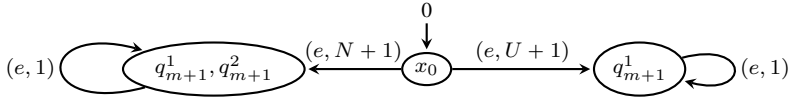

Assume that the subset sum problem has no solution, 
that is, for all $I\subset\llb 1,m\rrb$, $N\ne\sum_{i\in I}n_i$.

The self-composition of $\Acal^{\N}_2$ is shown in Fig.~\ref{fig3_det_MPautomata}, in which 
there is no reachable state of the form $(q,q')$ such that $q\ne q'$. 
Then by Theorem~\ref{thm1_det_MPautomata}, $\Acal^{\N}_2$ is strongly detectable.

\begin{figure}[!htbp]
        \centering
	\begin{tikzpicture}
	[>=stealth',shorten >=1pt,thick,auto,node distance=2.8 cm, scale = 1.0, transform shape,
	->,>=stealth,inner sep=2pt, initial text = 0]
	\tikzstyle{emptynode}=[inner sep=0,outer sep=0]
	
	\node[initial, initial where = above, elliptic state] (q0) {$q_0q_0$};
	\node[elliptic state] (q11) [right of = q0] {$q_{m+1}^1q_{m+1}^1$};
	\node[elliptic state] (q22) [left of = q0] {$q_{m+1}^2q_{m+1}^2$};

	\path [->]
	(q0) edge node [above, sloped] {$(e,e)$} (q11)
	(q0) edge node [above, sloped] {$(e,e)$} (q22)
	(q11) edge [loop right] node [right] {$(e,e)$} (q11)
	(q22) edge [loop left] node [left] {$(e,e)$} (q22)
	
	;
	\end{tikzpicture}
	\caption{Self-composition of $\Acal^{\N}_2$ in Fig.~\ref{fig1_det_MPautomata} when the 
	subset sum problem has no solution.}
	\label{fig3_det_MPautomata}
\end{figure}

The observers (also the detectors) of $\Acal^{\N}_2$ are shown in Fig.~\ref{fig5_det_MPautomata}. 
The initial state $x_0$ satisfies that $|x_0|>1$, but there is no unobservable cycle in $\Acal^{\N}_2$,
that is, \eqref{item16_det_MPautomata} of Theorem~\ref{thm9_det_MPautomata} is not satisfied.
Apparently, \eqref{item17_det_MPautomata} of Theorem~\ref{thm9_det_MPautomata} is not satisfied
either. Hence $\Acal^{\N}_2$ is strongly periodically detectable. Here $x_0$ also transitions to 
a self-loop in which the state is of cardinality $1$ (e.g., $\{q_{m+1}^1\}$), then by
Theorem~\ref{thm2_det_MPautomata} and Theorem~\ref{thm7_det_MPautomata}, $\Acal^{\N}_2$
is weakly detectable and weakly periodically detectable.

\begin{figure}[!htbp]
        \centering
	\begin{tikzpicture}
	a[>=stealth',shorten >=1pt,thick,auto,node distance=2.8 cm, scale = 1.0, transform shape,
	->,>=stealth,inner sep=2pt, initial text = 0]
	\tikzstyle{emptynode}=[inner sep=0,outer sep=0]
	
	\node[initial, initial where = above, elliptic state] (q0) {$x_0$};
	\node[elliptic state] (q1) [right of = q0] {$q_{m+1}^1$};
	\node[elliptic state] (q2) [left of = q0] {$q_{m+1}^2$};

	\path [->]
	(q0) edge node [above, sloped] {$(e,U+1)$} (q1)
	(q0) edge node [above, sloped] {$(e,N+1)$} (q2)
	(q1) edge [loop right] node [right] {$(e,1)$} (q1)
	(q2) edge [loop left] node [left] {$(e,1)$} (q2)
	
	;
	\end{tikzpicture}
	\caption{Observers (also detectors) of $\Acal^{\N}_2$ in Fig.~\ref{fig1_det_MPautomata} when the 
	subset sum problem has no solution, where $x_0=\Mt(\Acal^{\N}_2,\ep)=\{q_0,\dots,q_m\}$, $U$ can be the sum
	of elements of any subset of $\{n_1,\dots,n_m\}$.}
	\label{fig5_det_MPautomata}
\end{figure}

\end{example}

\section{Initial exploration of detectability in labeled timed automata}
\label{sec:detLTA}

As mentioned in the last paragraph of Section~\ref{subsec:background}, a labeled weighted automaton
$\Acal^{\Q_{\ge0}}$ is actually
a one-clock labeled timed automaton in which the clock is reset along with every occurrence of every event and all
clock constraints are singletons. In this section, we initially explore detectability in general labeled
timed automata, and show some relations between a labeled weighted automaton $\Acal^{\frakM}$ and a labeled
timed automaton $\Acal^{\frakT}$. 

\subsection{Notation}

Let $\Sig$ be an alphabet, a \emph{timed word} over $\Sig$ is a finite sequence $(a_1,t_1)\dots(a_n,t_n)$,
where $n\in\N$, $a_1,\dots,a_n\in\Sig$, $t_1,\dots,t_n\in\R_{\ge0}$, $t_1\le \cdots \le t_n$.
The set of timed words over $\Sig$ is denoted by $\TW^*(\Sig)$. An \emph{$\omega$-timed word} over $\Sig$
is an infinite sequence $(a_1,t_1)(a_2,t_2)\dots$, where $a_i\in\Sig$, $t_i\in\R_{\ge0}$, $t_i\le t_{i+1}$,
$i\in\Z_+$. The set of $\omega$-timed words over $\Sig$ is denoted by $\TW^{\omega}(\Sig)$.

A labeled timed automaton\footnote{In order to
study detectability, the model~\eqref{eqn_LTA} is obtained by adding a labeling function to the original
model of timed automata proposed in \cite{Alur1994TimedAutomaton}, and in addition, the final states in
the original model are
omitted. Adding a labeling function brings in essential difficulties. For example, event-recording automata
are a special class of timed automata which are determinizable, so the deterministic timed automaton obtained
by determinizing an event-recording automaton can be regarded as the observer of the latter and hence can be 
used to verify current-state opacity of the latter, where current-state opacity means that for a generated 
event sequence, if its last state is secret, then there is another generated event sequence whose last state
is not secret such that the two event sequences generate the same label sequence.
However, current-state opacity (called L-opacity in \cite{Cassez2009TimedOpacityUndecidable}) is undecidable
in deterministic labeled event-recording automata \cite{Cassez2009TimedOpacityUndecidable}.}
\cite{Tripakis2002DiagnosisTimedAutomata,Cassez2009TimedOpacityUndecidable} is formulated as 
\begin{equation}\label{eqn_LTA}
	\Acal^{\frakT}=(Q,E,Q_0,C,\Delta,\Sig,\ell),
\end{equation}
where $Q$ is a finite set of \emph{states}, $E$ a finite \emph{alphabet} (elements of $E$ are
called \emph{events}), $Q_0\subset Q$ a set of \emph{initial states}, $C$ a finite set of \emph{clocks}
(i.e., real variables), $\Delta\subset Q\times E\times Q\times 2^C\times \Phi(C)$ the set of \emph{edges},
$\ell:E\to\Sig\cup\{\ep\}$ a labeling function. The observable event set $E_o$ and unobservable event set
$E_{uo}$ are defined in the same way as those for $\Acal^{\frakM}$~\eqref{LWA_monoid_det_MPautomata}.
Labeling function $\ell$ is also recursively extended to $E^*\cup E^{\omega}\to\Sig^*\cup\Sig^{\omega}$.
A timed automaton is actually an $\Acal^{\frakT}$ in which $\ell$ is the identity mapping.
In the sequel, by ``a timed automaton $\bar\Acal^{\frakT}$'' we mean the automaton obtained from $\Acal^{\frakT}$
by removing all its labels, including $\ep$. In this case, one directly observes every occurrence of every 
event.

An edge $(q,e,q',\lambda,\zeta)$ represents a transition from state $q$ to state $q'$ when event $e$ occurs
and the current \emph{clock interpretation} satisfies the \emph{clock constraint} $\zeta$ with the clocks in the 
\emph{reset subset}
$\lambda\subset C$ reset (to $0$). A clock interpretation is a mapping $v:C\to\R_{\ge0}$ which assigns to each 
clock a nonnegative real number. A clock constraint $\zeta$ over $C$ is defined inductively by
$$\zeta:=x\le c|c\le x|\neg \zeta|\zeta_1\wedge\zeta_2,$$
where $x$ is a clock in $C$ and $c$ is a nonnegative rational constant.  A clock interpretation $\lambda$
satisfies a clock constraint $\zeta$ if and only if $\zeta$ evaluates to true using the values given by $v$.
For a clock interpretation $v$, a constant $t\in\R_{\ge0}$, and a reset set $\lambda\subset C$,
$v+t$ denotes the clock
interpretation that assigns to each clock $x$ the value $v(x)+t$, $[\lambda\mapsto t]v$ denotes the clock
interpretation that assigns $t$ to each clock in $\lambda$ and agrees with $v$ over the rest of the clocks.

A pair $(q,v)$ is called a \emph{configuration}, where $q\in Q$, $v\in(\R_{\ge0})^{C}$ (i.e., $v$ is 
a clock interpretation). An \emph{infinite run} of $\Acal^{\frakT}$ is defined as an alternating 
sequence
\begin{equation}\label{eqn_LTA_inf_run}
	\begin{array}{ll}
		\pi:= &(q_0,v_0)\xrightarrow[]{t_1}(q_0,v_0+t_1)\xrightarrow[\lambda_1,\zeta_1]{e_1}(q_1,v_1)\cdots\\
		&(q_i,v_i)\xrightarrow[]{t_{i+1}}(q_i,v_i+t_{i+1})\xrightarrow[\lambda_{i+1},\zeta_{i+1}]{e_{i+1}}(q_{i+1},v_{i+1})\cdots,
	\end{array}
\end{equation}
or briefly as
\begin{equation}\label{eqn_LTA_inf_run_brief}
	\pi:= (q_0,v_0)\xrightarrow[\lambda_1,\zeta_1]{t_1,e_1}(q_1,v_1)\cdots
	(q_i,v_i)\xrightarrow[\lambda_{i+1},\zeta_{i+1}]{t_{i+1},e_{i+1}}(q_{i+1},v_{i+1})\cdots,
\end{equation}
where $q_0\in Q_0$, $v_0(x)=0$ for all $x\in C$, for all $i\in\N$, one has $t_i\ge0$, $v_i+t_{i+1}$ satisfies
$\zeta_{i+1}$, $v_{i+1}=[\lambda_{i+1}\mapsto 0](v_i+t_{i+1})$, $(q_i,e_{i+1},q_{i+1},\lambda_{i+1},
\zeta_{i+1})\in\Delta$. 

When $\Acal^{\frakT}$ was in a state $q$ at time instant a nonnegative real number $\tau$ with clock
interpretation $v$, as time elapsed, $v$ might become $v + t$ with $t$ a nonnegative real number. If there was
a transition $(q,e,q',\lambda,\zeta)$ in $\Delta$ and $v + t$ satisfies $\zeta$, then the automaton could transition
to state $q'$ at time instant $\tau + t$ with the occurrence of $e$, and meanwhile $v + t$ was changed to 
$[\lambda\mapsto0](v + t)$, i.e., all clocks in $\lambda$ were reset and all clocks outside $\lambda$ remained 
invariant.

For an infinite run $\pi$ as in \eqref{eqn_LTA_inf_run}, the sequence $(e_1,t_1)\dots(e_i,\sum_{j=1}^i t_j)\dots$
is called its \emph{$\omega$-timed word} and denoted by $\tau(\pi)$. We use $L^{\omega}(\Acal^{\frakT})$ 
to denote the set of the
$\omega$-timed words  of infinite runs of $\Acal^{\frakT}$. Similarly, a \emph{finite run} is defined as
a prefix of an infinite run $\pi$ \eqref{eqn_LTA_inf_run} ending with some $(q_i,v_i+t_{i+1})\xrightarrow[\lambda_{i+1},
\zeta_{i+1}]{e_{i+1}}(q_{i+1},v_{i+1})$, then its \emph{timed word} is defined by $(e_1,t_1)\dots(e_{i+1},\sum_{j=1}^{i+1}
t_j)$ and its \emph{time}, denoted by $\tim(\pi)$, is defined by $\sum_{j=1}^{i+1}t_j$.
We use $L(\Acal^{\frakT})$ to denote the set of the timed words of finite runs of $\Acal^{\frakT}$.

A labeled timed automaton $\Acal^{\frakT}$ is \emph{deterministic} if
\begin{itemize}
	\item $|Q_0|=1$,
	\item for every two different edges of the form $(q,e,-,-,\zeta_1)$ and $(q,e,-,-,\zeta_2)$,
		$\zeta_1\wedge\zeta_2$ is unsatisfiable.
\end{itemize}
In a deterministic $\Acal^{\frakT}$, for an edge $(q,e,q',\lambda,\zeta)$, once $q,e,\zeta$ are fixed,
$q'$ and $\lambda$ are uniquely determined. And for an edge of the form $(q,e,-,-,\zeta)$ and a clock interpretation
$v$, once $q,e,v$ are fixed, there is at most one $\zeta$ such that $v$ satisfies $\zeta$. Hence for an edge
$(q,e,q',\lambda,\zeta)$ and a clock interpretation $v$, one has $q,e,v$ uniquely determine $q',\lambda,\zeta$.
Consequently, for a deterministic $\Acal^{\frakT}$, consider an infinite run $\pi$~\eqref{eqn_LTA_inf_run},
if $q_0,v_0$ and its $\omega$-timed word are fixed, then the whole run is uniquely determined. 

Labeling function
$\ell$ is also extended to $E\times \R_{\ge0}$ and recursively extended 
to $(E\times\R_{\ge0})^*\cup(E\times\R_{\ge0})^{\omega}$ in the same way as those for 
$\Acal^{\frakM}$~\eqref{LWA_monoid_det_MPautomata}. The \emph{observation} to a (finite or infinite) run
$\pi$ is defined by
timed label sequence $\ell(\tau(\pi))$ based on a conventional setting in which there is a clock
outside $\Acal^{\frakT}$ that is never reset and can record global time instants of the occurrences
of all events, where the setting has been widely used in the studies of fault diagnosis of $\Acal^{\frakT}$ 
\cite{Tripakis2002DiagnosisTimedAutomata} and opacity of $\Acal^{\frakT}$ \cite{Cassez2009TimedOpacityUndecidable}.
For a timed label sequence $\gamma\in(\Sig\times\R_{\ge0})^*$, the \emph{current-state estimate}
$\Mt(\Acal^{\frakT},\gamma)$\footnote{corresponding to $\Mt(\Acal^{\frakM},\gamma)$~\eqref{CSE_det_MPautomata}}
is defined by the set of states in $Q$ that are reachable from $Q_0$ via some finite
run the observation to which is $\gamma$, where after the last observable event in the run, no time elapses. 
Formally, 
\begin{equation}
	\begin{split}
		\Mt(\Acal^{\frakT},\ep) = Q_0\cup \{q\in Q| &(\exists \text{ run } (q_0,v_0)\xrightarrow[\lambda_1,\zeta_1]{0,e_1}\cdots
		\xrightarrow[\lambda_{i+1},\zeta_{i+1}]{0,e_{i+1}}(q_{i+1},v_{i+1})) \\
		&[(e_1,\dots,e_{i+1}\in E_{uo})\wedge(i\in\N)]\},
	\end{split}
\end{equation}
for all $\gamma\in(\Sig\times\R_{\ge0})^+$,
\begin{equation}
	\begin{split}
		\Mt(\Acal^{\frakT},\gamma) = \{ & (\exists\text{ run } 
		(q_0,v_0)\xrightarrow[\lambda_1,\zeta_1]{t_1,e_1}\cdots
		\xrightarrow[\lambda_{i+1},\zeta_{i+1}]{t_{i+1},e_{i+1}}(q_{i+1},v_{i+1})\\
		&\xrightarrow[\lambda_{i+2},\zeta_{i+2}]{{t_{i+2}},e_{i+2}}\cdots
		\xrightarrow[\lambda_{i+j+1},\zeta_{i+j+1}]{{t_{i+j+1}},e_{i+j+1}}(q_{i+j+1},v_{i+j+1}))\\
		&[(e_{i+1}\in E_{o})\wedge (e_{i+2},\dots,e_{i+j+1}\in E_{uo})\wedge(i,j\in\N)\wedge\\
		&\ell(\tau((q_0,v_0)\xrightarrow[\lambda_1,\zeta_1]{t_1,e_1}\cdots
		\xrightarrow[\lambda_{i+1},\zeta_{i+1}]{t_{i+1},e_{i+1}}(q_{i+1},v_{i+1})))=\gamma]\}.
	\end{split}
\end{equation}

\subsection{Relation between $\Acal^{\Q_{\ge0}}$ and $\Acal^{\frakT}$}

Consider labeled weighted automaton $\Acal^{\Q_{\ge0}}$~\eqref{LWA_monoid_det_MPautomata}, if we remove the 
weights on the initial states, and replace each transition $q\xrightarrow[]{e/\mu(e)_{qq'}}q'$ by an edge
$(q,e,q',\{x\},x=\mu(e)_{qq'})$, then $\Acal^{\Q_{\ge0}}$ becomes an equivalent labeled timed automaton
with a single clock $x$, where $x$ is reset along with every occurrence of every event.
The paths (starting from initial states) of the former correspond to the runs of the latter,
and the timed words of the paths coincide 
with the timed words of the runs. Hence by definition, a deterministic $\Acal^{\Q_{\ge0}}$ is a
deterministic labeled timed automaton, but not vice versa.

\begin{example}\label{exam1_LTA}
	Reconsider the labeled weighted automaton~$\Acal_1^{\N}$ in Fig.~\ref{fig6_det_MPautomata}. 

	\begin{figure}[!htbp]
        \centering
	\begin{tikzpicture}
	[>=stealth',shorten >=1pt,thick,auto,node distance=3.0 cm, scale = 1.0, transform shape,
	->,>=stealth,inner sep=2pt]

	\tikzstyle{emptynode}=[inner sep=0,outer sep=0]

	\node[initial, state, initial where = left] (q0) {$q_0$};
	\node[state] (q1) [above right of = q0] {$q_1$};
	\node[state] (q2) [below right of = q0] {$q_2$};
	\node[state] (q3) [above right of = q2] {$q_3$};
	\node[state] (q4) [right of = q3] {$q_{4}$};

	\path [->]
	(q0) edge node [above, sloped] {$a$} node [below, sloped] {$\{x\},x=1$} (q1)
	(q0) edge node [above, sloped] {$a$} node [below, sloped] {$\{x\},x=1$} (q2)
	(q1) edge [loop above] node [above, sloped] {$u,\{x\},x=1$} (q1)
	(q2) edge [loop below] node [below, sloped] {$u,\{x\},x=1$} (q2)
	(q1) edge node [above, sloped] {$b$} node [below, sloped] {$\{x\},x=2$} (q3)
	(q2) edge node [above, sloped] {$b$} node [below, sloped] {$\{x\},x=1$}  (q3)
	(q3) edge node [above, sloped] {$u$} node [below, sloped] {$\{x\},x=1$} (q4)
	(q4) edge [loop right] node {$a,\{x\},x=1$} (q4)
	;

        \end{tikzpicture}
		\caption{Labeled timed automaton $\Acal_1^{\frakT}$ that is equivalent to the labeled weighted automaton
		$\Acal_1^{\N}$ in Fig.~\ref{fig6_det_MPautomata}, where $x$ is the unique clock and reset whenever
		an event occurs, $\ell(u)=\ep$, $\ell(a)=\ell(b)=\rho$.}
		\label{fig1_det_LTA}
	\end{figure}
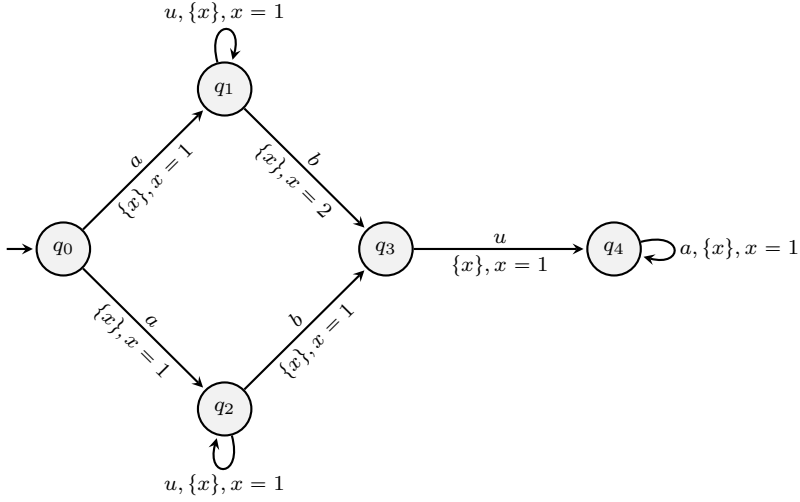

	Automaton $\Acal_1^{\N}$ can be equivalently represented as the labeled timed automaton $\Acal_1^{\frakT}$ 
	shown in Fig.~\ref{fig1_det_LTA}. Consider the finite run
	\begin{align}\label{eqn1_det_LTA}
		\pi_1 = (q_0,0) \xrightarrow[]{1} (q_0,1) \xrightarrow[\{x\},x=1]{a} (q_1,0)
		\xrightarrow[]{2} (q_1,2) \xrightarrow[\{x\},x=2]{b} (q_3,0),
	\end{align}
	its timed word and timed label sequence are
	\begin{subequations}\label{eqn2_det_LTA}
		\begin{align}
			\tau(\pi_1) &= (a,1)(b,3)\text{ and}\\
			\ell(\tau(\pi_1)) &= (\rho,1)(\rho,3),\text{ respectively}.
		\end{align}
	\end{subequations}
	Note that the run $\pi_1$ in \eqref{eqn1_det_LTA} is consistent with the path $\pi_1$ in
	\eqref{eqn12_1_det_MPautomata}, the timed word $\tau(\pi_1)$ and timed label sequence $\ell(\tau(\pi_1))$ in 
	\eqref{eqn2_det_LTA} are the same as those in \eqref{eqn13_1_det_MPautomata}.
\end{example}

\subsection{The definitions of detectability}

We reformulate the four definitions of detectability for labeled timed automata.

\begin{definition}[SD]\label{def_SD_LTA}
	A labeled timed automaton $\Acal^{\frakT}$ \eqref{eqn_LTA}
	is called \emph{strongly detectable} if
	there is $t\in\blue\Z_+$, for every $\omega$-timed word $w\in L^{\omega}(\Acal^\frakT)$, for each
	prefix $\gamma$ of $\ell(w)$, if $|\gamma|\ge t$, then $|\Mt(\Acal^{\frakT},\gamma)|=1$.
\end{definition}

\begin{definition}[SPD]\label{def_SPD_LTA}
	A labeled timed automaton $\Acal^{\frakT}$ \eqref{eqn_LTA}
	is called \emph{strongly periodically
	detectable} if there is $t\in\blue\Z_+$, for every $\omega$-timed word $w\in L^{\omega}(\Acal^\frakT)$, 
	for every prefix $w'\sqsubset w$, there is $w''\in (E\times \R_{\ge0})^*$ such that 
	$|\ell(w'')|<t$, $w'w''\sqsubset w$, and $|\Mt(\Acal^{\frakT},\ell(w'w''))|=1$.
\end{definition}

\begin{definition}[WD]\label{def_WD_LTA}
	A labeled timed automaton $\Acal^{\frakT}$ \eqref{eqn_LTA}
	is called \emph{weakly detectable} if
	$L^{\omega}(\Acal^\frakT)\ne\emptyset$ implies that
	there is $t\in\blue\Z_+$, for some $\omega$-timed word $w\in L^{\omega}(\Acal^\frakT)$, for each
	prefix $\gamma$ of $\ell(w)$, if $|\gamma|\ge t$, then $|\Mt(\Acal^{\frakT},\gamma)|=1$.
\end{definition}

\begin{definition}[WPD]\label{def_WPD_LTA}
	A labeled timed automaton $\Acal^{\frakT}$ \eqref{eqn_LTA}
	is called \emph{weakly periodically
	detectable} if $L^{\omega}(\Acal^\frakT)\ne\emptyset$ implies that
	there is $t\in\blue\Z_+$, for some $\omega$-timed word $w\in L^{\omega}(\Acal^\frakM)$, for each prefix
	$w'\sqsubset w$, there is $w''\in (E\times \R_{\ge0})^*$ such that $|\ell(w'')|<t$, $w'w''\sqsubset w$,
	and $|\Mt(\Acal^{\frakM},\ell(w'w''))|=1$.
\end{definition}

By definition, a deterministic timed automaton satisfies the four definitions of detectability.

\subsection{Decidability and undecidability of detectability}

In this section, we prove that in labeled timed automata, the strong detectability verification problem
is $\PSPACE$-complete, while weak (periodic) detectability is undecidable, which remarkably differentiates a 
labeled timed automaton $\Acal^{\frakT}$ from a labeled weighted automaton $\Acal^{\Q^k}$.

In order to prove $\PSPACE$-hardness of verifying strong detectability, we use the following reduction
that is almost the same as the one used to prove $\PSPACE$-hardness of verifying diagnosability of labeled
timed automata \cite{Tripakis2002DiagnosisTimedAutomata}.
Given a deterministic timed automaton $\bar\Acal^{\frakT}$ (which is always detectable as mentioned above)
and a target state $q_f$,
add two transitions from $q_f$ to two new states $q_1$ and $q_2$, and add two self-loops
on $q_1$ and $q_2$, all the four transitions have the same event and the same clock
constraint \textbf{true}. Call the new automaton $\bar\Acal^{\frakT}_{ext}$. All the events of 
$\bar\Acal^{\frakT}_{ext}$
are set to be observable with themselves as their labels. Then, it can be seen that $q_f$ is reachable in 
$\bar\Acal^{\frakT}$ if and only if $\bar\Acal^{\frakT}_{ext}$ is not strongly (periodically) detectable.
By the $\PSPACE$-hardness of the reachability problem in deterministic timed automata \cite{Alur1994TimedAutomaton},
we have it is $\PSPACE$-hard to verify strong  (periodic) detectability of labeled timed automata.

In order to verify strong detectability of $\Acal^{\frakT}$, we use the 
standard \emph{parallel composition} $\Acal^{\frakT}||\Acal^{\frakT}$ 
\cite{Tripakis2002DiagnosisTimedAutomata} in which observable edges of $\Acal^{\frakT}$ with the same label
are synchronized but unobservable edges interleave.
We make two copies of $\Acal^{\frakT}$, $\Acal^{\frakT}_i=(Q_i,E_i,Q_0^i,
C_i,\Delta_i,\Sig,\ell_i)$, $i=1,2$, by renaming states, events, and clocks of $\Acal^{\frakT}$:
\begin{itemize}
	\item Each state $q$ of $\Acal^{\frakT}$ is renamed $q_1$ in $\Acal^{\frakT}_1$ and $q_2$ in $\Acal^{\frakT}_2$.
	\item Each event $e$ of $\Acal^{\frakT}$ is renamed $e_1$ in $\Acal^{\frakT}_1$ and $e_2$ in $\Acal^{\frakT}_2$.
	\item Each clock $x$ of $\Acal^{\frakT}$ is renamed $x_1$ in $\Acal^{\frakT}_1$ and $x_2$ in $\Acal^{\frakT}_2$.
	\item The edges are copied and renamed correspondingly.
	\item For $i=1,2$, for all events $e_i\in E_i$, $\ell_i(e_i):=\ell(e)$.
\end{itemize}
$\Acal^{\frakT}||\Acal^{\frakT}$ is defined by 
\begin{equation}\label{eqn_PC_LTA} 
	(Q',E',Q_0',C',\Delta',\Sig,\ell'),
\end{equation}
where $Q'=Q_1\times Q_2$, $E'=\{(e_1,e_2)|e_1\in E_1,e_2\in E_2,\ell_1(e_1)=\ell_2(e_2)\in\Sig\}\cup\{(e_1,\ep)
|e_1\in E_1,\ell_1(e_1)
=\ep\}\cup\{(\ep,e_2)|e_2\in E_2,\ell_2(e_2)=\ep\}$, $Q_0'=Q_0^1\times Q_0^2$, $C'=C_1\cup C_2$,
for every two observable edges $(q_i,e_i,q_i',\lambda_i,\zeta_i)$ of $A^{\frakT}_i$, $i=1,2$, with
$\ell_1(e_1)=\ell_2(e_2)\in\Sig$, construct an observable edge $((q_1,q_2),(e_1,e_2),(q_1',q_2'),\lambda_1\cup
\lambda_2,\zeta_1\wedge \zeta_2)$ of $\Acal^{\frakT}||\Acal^{\frakT}$, for every unobservable edge 
$(q_1,e_1,q_1',\lambda_1,\zeta_1)$ of $\Acal^{\frakT}
_1$ and every state $q_2$ of $\Acal^{\frakT}_2$, construct an unobservable edge $((q_1,q_2),(e_1,\ep),(q_1',q_2),
\lambda_1,\zeta_1)$ of $\Acal^{\frakT}||\Acal^{\frakT}$, for every state $q_1$ of $\Acal^{\frakT}_1$ and every
unobservable edge $(q_2,e_2,q_2',\lambda_2, \zeta_2)$ of $\Acal^{\frakT}_2$, construct an unobservable edge
$((q_1,q_2),(\ep,e_2),(q_1,q_2'),\lambda_2,\zeta_2)$ of $\Acal^{\frakT}||\Acal^{\frakT}$, for every event
$(e_1,e_2)\in E'$, $\ell'((e_1,e_2)):=\ell_1(e_1)=\ell_2(e_2)$.

$\Acal^{\frakT}||\Acal^{\frakT}$ can be computed in time polynomial in the size of $\Acal^{\frakT}$.
Let $\rho$ be a finite run of $\Acal^{\frakT}||\Acal^{\frakT}$, for $i=1,2$, let $\rho^i$ be obtained by erasing
all elements of $\Acal^{\frakT}_{3-i}$ from $\rho$ and aggregating all consecutive time delays without events
between them. Then one has $\rho$ is a finite run of $\Acal^{\frakT}||\Acal^{\frakT}$ if and only if $\rho^1$
and $\rho^2$ are finite runs of $\Acal_1^{\frakT}$ and $\Acal_2^{\frakT}$, respectively, and for such
$\rho,\rho^1,\rho^2$, one has
$\tim(\rho)=\tim(\rho^1)=\tim(\rho^2)$ \cite{Tripakis2002DiagnosisTimedAutomata}.
For example, consider finite run 
\begin{align*}
	\rho = &((q_0^1,q_0^2),(v_0^1,v_0^2))\xrightarrow[\lambda_1^1\cup\lambda_1^2,\zeta_1^1\wedge\zeta_1^2]{t_1,
	(e_1^1,e_1^2)}  ((q_1^1,q_1^2),(v_1^1,v_1^2)) \xrightarrow[\lambda_2^2,\zeta_2^2]{t_2,(\ep,e_2^2)}\\
	&((q_1^1,q_2^2),(v_2^1,v_2^2)) \xrightarrow[\lambda_3^1\cup\lambda_3^2,\zeta_3^1\wedge\zeta_3^2]{t_3,(e_3^1,e_3^2)}
	((q_3^1,q_3^2),(v_3^1,v_3^2)),
\end{align*}
the corresponding $\rho^1$ and $\rho^2$ are as follows:
\begin{align*}
	\rho^1 = (q_0^1,v_0^1)\xrightarrow[\lambda_1^1,\zeta_1^1]{t_1,e_1^1} & (q_1^1,v_1^1) \xrightarrow[\lambda_3^1,
	\zeta_3^1]{t_2+t_3,e_3^1}(q_3^1,v_3^1),\\
	\rho^2 = (q_0^2,v_0^2)\xrightarrow[\lambda_1^2,\zeta_1^2]{t_1,e_1^2} & (q_1^2,v_1^2) \xrightarrow[\lambda_2^2,\zeta_2^2]{t_2,e_2^2}(q_2^2,v_2^2)
	\xrightarrow[\lambda_3^2,\zeta_3^2]{t_3,e_3^2}(q_3^2,v_3^2).
\end{align*}
Then 
\begin{align*}
	\tim(\rho) &= \tim(\rho^1)=\tim(\rho^2)=t_1+t_2+t_3,\\
	\ell(\tau(\rho)) &= \ell(\tau(\rho^1)) = \ell(\tau(\rho^2)) = (\ell(e_1^1),t_1)(\ell(e_3^1),t_1+t_2+t_3).
\end{align*}

\begin{example}\label{exam2_LTA}
	Reconsider the labeled timed automaton $\Acal_1^{\frakT}$ in Example~\ref{exam1_LTA} (shown in 
	Fig.~\ref{fig1_det_LTA}). Part of the parallel composition $\Acal_1^{\frakT}||\Acal_1^{\frakT}$ is
	shown in Fig.~\ref{fig2_det_LTA}.
	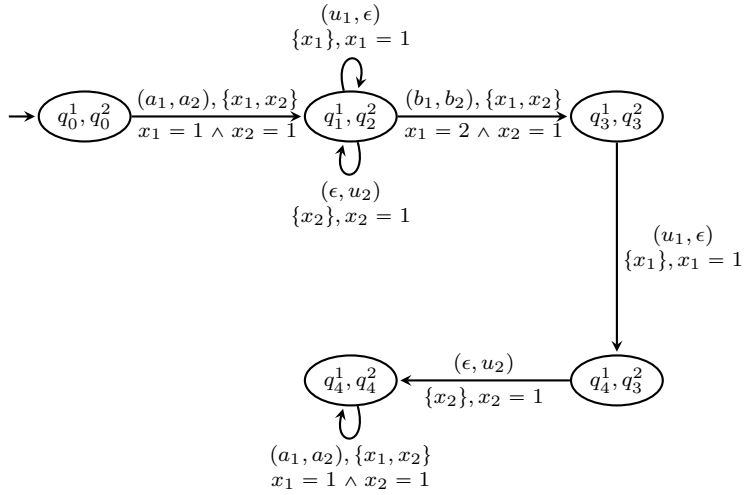
\begin{figure}[!htbp]
        \centering
	\begin{tikzpicture}
	[>=stealth',shorten >=1pt,thick,auto,node distance=3.5 cm, scale = 1.0, transform shape,
	->,>=stealth,inner sep=2pt]

	\tikzstyle{emptynode}=[inner sep=0,outer sep=0]

	\node[initial, elliptic state, initial where = left] (00) {$q_0^1,q_0^2$};
	\node[elliptic state] (12) [right of = 00] {$q_1^1,q_2^2$};
	\node[elliptic state] (33) [right of = 12] {$q_3^1,q_3^2$};
	\node[elliptic state] (43) [below of = 33] {$q_4^1,q_3^2$};
	\node[elliptic state] (44) [left of = 43] {$q_4^1,q_4^2$};

	\path [->]
	(00) edge node [above, sloped] {$(a_1,a_2),\{x_1,x_2\}$} node [below, sloped] {$x_1=1\wedge x_2=1$} (12)
	(12) edge [loop above] node [above, sloped] {$\begin{matrix}(u_1,\ep)\\\{x_1\},x_1=1\end{matrix}$} (12)
	(12) edge [loop below] node [below, sloped] {$\begin{matrix}(\ep,u_2)\\\{x_2\},x_2=1\end{matrix}$} (12)
	(12) edge node [above, sloped] {$(b_1,b_2),\{x_1,x_2\}$} node [below, sloped] {$x_1=2\wedge x_2=1$} (33)
	(33) edge node [right] {$\begin{matrix}(u_1,\ep)\\\{x_1\},x_1=1\end{matrix}$} (43)
	(43) edge node [above, sloped] {$(\ep,u_2)$} node [below, sloped] {$\{x_2\},x_2=1$} (44)
	(44) edge [loop below] node [below] {$\begin{matrix}(a_1,a_2),\{x_1,x_2\}\\x_1=1\wedge x_2=1\end{matrix}$} (44)
	;

    \end{tikzpicture}
	\caption{Part of the parallel composition $\Acal_1^{\frakT}||\Acal_1^{\frakT}$, where $\Acal_1^{\frakT}$
	is in Fig~\ref{fig1_det_LTA}.}
	\label{fig2_det_LTA}
	\end{figure}
\end{example}

By using the parallel composition $\Acal^{\frakT}||\Acal^{\frakT}$, we give an equivalent condition for 
strong detectability of $\Acal^{\frakT}$ (analogous to the equivalent condition for
strong detectability of labeled weighted automaton $\Acal^{\frakM}$ in Theorem~\ref{thm1_det_MPautomata} 
based on its self-composition $\CCa(\Acal^{\frakM})$). Differently
from Theorem~\ref{thm1_det_MPautomata} in which $\CCa(\Acal^{\frakM})$ was directly used, here
$\Acal^{\frakT}||\Acal^{\frakT}$ cannot be directly used because $\Acal^{\frakT}$ is time-variant,
i.e., when $\Acal^{\frakT}$ was in the same state at different time instants, the possible next transitions 
may vary. We use the region automaton \cite{Alur1994TimedAutomaton} of $\Acal^{\frakT}||\Acal^{\frakT}$
to obtain the equivalent condition. 

The region automaton $\RA(\bar\Acal^{\frakT})$ of a timed automaton $\bar\Acal^{\frakT}$ (here labels are 
not useful) is actually a finite-state automaton each of whose states is a pair of a state 
of $\bar\Acal^{\frakT}$ and a clock region (a special subset of $(\mathbb{R}_{\ge0})^C$).
The set $(\mathbb{R}_{\ge0})^C$ is partitioned into a finite number of subsets (called clock regions)
based on a special equivalence relation $\sim$ on $\R_{\ge0}^C$, where each clock region is an equivalence
class and each clock of all vectors
in the same clock region will be reset in the same order as time advances (if possible). We refer the reader 
to \cite{Alur1994TimedAutomaton} for details of constructing such a finite partition and a region automaton.
The size of a region automaton is exponential in the size of $\bar\Acal^{\frakT}$. The runs of an $\bar\Acal^{\frakT}$ 
correspond to the runs of its region automaton $\RA(\bar\Acal^{\frakT})$. In detail, for every run 
$\pi$~\eqref{eqn_LTA_inf_run_brief}, after replacing each $v_i$ by the clock region $[v_i]_{\sim}$ generated by
$v_i$ and removing each $t_i,\lambda_i,\zeta_i$, then a run of $\RA(\bar\Acal^{\frakT})$ is obtained. Finding a
run of $\bar\Acal^{\frakT}$ that corresponds to a given run of $\RA(\bar\Acal^{\frakT})$ can also be done (although
more complicatedly). Then based on the correspondence and argument similar
to that in the proof of Theorem~\ref{thm1_det_MPautomata}, the following result holds.
\begin{theorem}\label{thm1_det_LTA}
	A labeled timed automaton $\Acal^{\frakT}$~\eqref{eqn_LTA}
	is not strongly detectable if and only if in the region automaton $\RA(\Acal^{\frakT}||\Acal^{\frakT})$
	of the parallel composition $\Acal^{\frakT}||\Acal^{\frakT}$ \eqref{eqn_PC_LTA},
	\begin{enumerate}[(i)]
		\item there exists a transition sequence
			\begin{align*}
				((q_0^1,q_0^2),R_0) \xrightarrow[]{s_1'} ((q_1^1,q_1^2),R_1) \xrightarrow[]{s_2'}
				((q_1^1,q_1^2),R_1) \xrightarrow[]{s_3'} ((q_2^1,q_2^2),R_2)
			\end{align*}
			such that $((q_0^1,q_0^2),R_0)$ is initial,
			$s_2'$ contains at least one observable event of $\Acal^{\frakT}||\Acal^{\frakT}$, $q_2^1\ne q_2^2$,
		\item
			$(q_2^1,R_2|_{C_1})$ is reachable in the region automaton $\RA(\Acal^{\frakT})$ of $\Acal^{\frakT}$
			and there is a cycle reachable from $(q_2^1,R_2|_{C_1})$ in $\RA(\Acal^{\frakT})$, where $R_2|_{C_1}$
			is the projection of clock region $R_2$ to the left component of $\Acal^{\frakT}||\Acal^{\frakT}$.
	\end{enumerate}
\end{theorem}

By using nondeterministic search, the conditions in Theorem~\ref{thm1_det_LTA} can be checked in 
$\NPSPACE$, without computing the whole $\RA(\Acal^{\frakT}||\Acal^{\frakT})$ and
$\RA(\Acal^{\frakT})$. Hence by $\coNPSPACE=\NPSPACE=\PSPACE$, the following complexity result follows.

\begin{theorem}\label{thm2_det_LTA}
	The strong detectability verification problem is $\PSPACE$-complete in labeled timed automata.
\end{theorem}

In labeled weighted automaton $\Acal^{\Q^k}$, we have proven that weak (periodic) detectability can be verified 
in $2$-$\EXPTIME$ (Theorem~\ref{thm6_det_MPautomata} and Theorem~\ref{thm10_det_MPautomata}), however, here 
the two properties become undecidable in labeled timed automata. We will use the following undecidable problem
to do reduction.

\begin{problem}[Universality]\label{prob1_det_LTA}
	Let $\bar\Acal^{\frakT}_{single}$ be a timed automaton with a single state, a single event $a$, and using 
	clock constants $0$ and $1$ only. Decide whether $L(\bar\Acal^{\frakT}_{single})=\TW^*(\{a\})$, i.e., whether
	$\bar\Acal^{\frakT}_{single}$ accepts every timed word over alphabet $\{a\}$.
\end{problem}

\begin{lemma}[\cite{Adams2007UndecidUniversalityRestrictedTA}]\label{lem1_det_LTA}
	Problem~\ref{prob1_det_LTA} is undecidable.
\end{lemma}

By Lemma~\ref{lem1_det_LTA} (\cite[Theorem~1]{Adams2007UndecidUniversalityRestrictedTA}), we prove 
the undecidability of weak (periodic) detectability.

\begin{theorem}\label{thm3_det_LTA}
	The weak (periodic) detectability of timed automata is undecidable.
\end{theorem}

\begin{proof}
	Given a timed automaton $\bar\Acal^{\frakT}_{single}$ with a single state $q_0$, a single event $a$, and 
	using clock constants $0$ and $1$ only. Add a new initial state $q_0'$ and self-loop on $q_0'$ with event $a$
	and clock constraint \textbf{true}. The newly added timed automaton is denoted by $\tilde\Acal^{\frakT}_{single}$.
	It is easy to see that $\tilde\Acal^{\frakT}_{single}$ is deterministic and accepts every timed word
	over alphabet $\{a\}$ (i.e., $L(\tilde\Acal^{\frakT}_{single})=\TW^{*}(\{a\})$), and
	$L^{\omega}(\tilde\Acal^{\frakT}_{single})=\TW^{\omega}(\{a\})$.
	Set $\ell(a)=a$, i.e., $a$ is observable.
	As mentioned above, every deterministic timed automaton is detectable. Choose a
	timed word $w$ in $\TW^*(\{a\})\setminus L(\bar\Acal^{\frakT}_{single})$ (if any), construct an $\omega$-timed
	word $w'=w(a,\tim(w)+1)(a,\tim(w)+2)\dots$$\in L^{\omega}(\tilde\Acal^{\frakT}_{single})$, then 
	for every timed label sequence $\gamma\sqsubset\ell(\tau(w'))$, if $|\gamma|\ge|w|$ then $\Mt(\bar\Acal^{\frakT}
	_{single}\cup \tilde\Acal^{\frakT}_{single},\gamma)=\{q_0'\}$. Hence 
	$L(\bar\Acal^{\frakT}_{single})=\TW^*(\{a\})$ if and only if $\bar\Acal^{\frakT}_{single}\cup \tilde
	\Acal^{\frakT}_{single}$ is not weakly (periodically) detectable. 
\end{proof}

\begin{remark}
	In \cite{Tripakis2002DiagnosisTimedAutomata}, for labeled timed automata, it is proven that the
	diagnosability verification problem is $\PSPACE$-complete, and a \emph{diagnoser} (supposed
	to be recursive, see \cite{Bouyer2005FaultDiagnosisTimedAutomata}) is also constructed and fault diagnosis 
	is done in $2$-$\EXPTIME$ in the size of a given labeled timed automaton and in the length of a given
	timed labeled sequence. This diagnoser can be regarded as an observer, because it can be used to do 
	state estimation. However, it cannot be computed with complexity upper bounds, because if so, then 
	it could be used to verify weak (periodic) detectability and current-state opacity which are undecidable
	(Theorem~\ref{thm3_det_LTA}, \cite{Cassez2009TimedOpacityUndecidable}) in labeled timed automata.
\end{remark}

\begin{remark}
	The results above in this section were all obtained over weakly monotone time, i.e., in a ($\omega$-)timed 
	word, its time sequence is (not necessarily strictly) increasing. All these results also hold over strongly 
	monotone time, i.e., in a ($\omega$-)timed word, its time sequence is strictly increasing. The reachability 
	problem is also
	$\PSPACE$-hard in deterministic timed automata over strongly monotone time \cite{Alur1994TimedAutomaton},
	so the strong (periodic) detectability verification problem is also $\PSPACE$-hard in labeled timed automata
	over strongly monotone time. The $\PSPACE$-easiness of verifying strong detectability of labeled timed automata
	can be obtained by using the region automata of timed automata over strongly monotone time
	\cite{Alur1994TimedAutomaton} (the difference of region automata over strongly monotone time and over
	weakly monotone time can be found in \cite{Ouaknine2003UnivLangIncOpenCloseTA}). The undecidability of 
	the weak (periodic) detectability problem can be proved by using the undecidability of the universality
	problem for timed automata with a single state, a single event, and using clock constants $1$, $2$, and $3$
	only, over strongly monotone time (\cite[Theorem~2]{Adams2007UndecidUniversalityRestrictedTA}).
\end{remark}

\section{conclusion}
In this paper, we extended the notions of concurrent composition, observer, and detector
from labeled finite-state automata to labeled weighted automata over monoids.
By using these extended notions, we 
gave equivalent conditions for four fundamental notions of detectability (i.e., 
strong (periodic) detectability and weak (periodic) detectability) for such automata. Particularly,
for a labeled weighted automaton $\Acal^{\Q^k}$ over the monoid $(\Q^k,+)$,
we proved that its concurrent composition, observer, and detector
can be computed in $\NP$, $2$-$\EXPTIME$, and $2$-$\EXPTIME$, respectively.
Moreover, for $\Acal^{\Q^k}$, we gave a $\coNP$ upper bound on verifying its strong detectability,
and $2$-$\EXPTIME$ upper bounds on verifying its strong periodic detectability and weak
(periodic) detectability.
We also gave $\coNP$ lower bounds for verifying strong (periodic) detectability of
labeled deterministic weighted automata over monoid $(\N,+,0)$.

The original methods developed in the current paper have been extended to labeled real-time automata which are a
subclass of labeled timed automata with
a single clock and whose clock constrains are all intervals in $\R_{\ge0}$ with rational or infinite endpoints 
\cite{Dima2002RealTimeAutomata}.
Four definitions of state-based opacity were formulated for labeled real-time automata in 
\cite{Zhang2021StateOpacityRTA}, and were also verified in $2$-$\EXPTIME$
by using the observers of labeled real-time automata
computable in $2$-$\EXPTIME$. The results in the current paper and those in \cite{Zhang2021StateOpacityRTA}
provide all technical details for computing observers of labeled real-time automata.

In addition, in order to differentiate labeled weighted automata over monoids from labeled timed automata,
we also initially explored detectability in labeled timed automata, and proved that the strong detectability 
verification problem is PSPACE-complete, while weak (periodic) detectability is undecidable.

It is the first time that the detectability verification results for general labeled weighted automata
over monoid $(\Q^k,+)$ were obtained algorithmically. The original methods proposed
in the current paper will provide foundations for characterizing other fundamental properties such as 
diagnosability and opacity, for general labeled weighted automata over monoids.
The algorithms for computing the observers of $\Acal^{\Q^k}$ could also imply that many
results obtained in labeled finite-state automata under the supervisory control framework can be 
extended to $\Acal^{\Q^k}$.
Several related open problems are as follows: whether the detectors of a general $\Acal^{\Q^k}$ can be 
computed in $\coNP$, whether strong periodic detectability of a general $\Acal^{\Q^k}$ can be verified 
in $\coNP$, what the complexity low bounds for verifying weak (periodic) detectability of
$\Acal^{\Q^k}$ are, studies on basic properties of labeled weighted automata over other monoids, etc.

\section*{Appendix}

\begin{remark}\label{rem3_det_MPautomata}
	Now we illustrate how to compute the observer of labeled weighted automaton $\Acal_0^{\N}$ in 
	Fig.~\ref{fig10_det_MPautomata} defined in \cite{Li2021ObserverSpecialTimedAutomata}. 
	First, we unfold states of $\Acal_0^\N$. Because the maximum among the weights of all outgoing transitions
	of $q_0$ is $10$, we unfold $q_0$ to $q_{0,1},\dots,q_{0,11}$. Similarly, we unfold $q_i$ to $q_{i,1},q_{i,2}$,
	$i\in\llb 1,4 \rrb$. Then we obtain the intermediate automaton as in Fig.~\ref{fig16_det_MPautomata}, where 
	``$1$'' means time delay ``$1$'' without transition, ``$0$'' means an unobservable transition, and $a$
	means an observable transition.
	Second, we compute the powerset construction of the intermediate automaton as in Fig.~\ref{fig17_det_MPautomata}.
	Third, we remove redundant states from the powerset construction and then recover states of $\Acal_0^{\N}$
	as in Fig.~\ref{fig18_det_MPautomata}, where all $q_{i,\max(q_i)}$ were removed expect for those as
	ending states of $a$-transitions, then change all the remaining $q_{i,j}$ to $q_i$, now we obtain the observer.
	Fourth, the observer is simplified as in Fig.~\ref{fig19_det_MPautomata} (by accumulating $1$-transitions),
	which is similar to but a little
	different from the observer $\Acal_{0obs}^{\N}$ defined in the current paper (as in 
	Fig.~\ref{fig15_det_MPautomata}). 

	In \cite{Li2021ObserverSpecialTimedAutomata}, for divergence-free $\Acal^{\Q_{\ge0}}$, an observer (obtained
	in the above third step) has at most $\sum_{k=0}^{|Q|}\binom{|Q|}{k}(2^M)^k=(1+2^M)^{|Q|}$ states and at most
	$(1+2^M)^{|Q|}( 2^M+|\Sig|)$ transitions; for general $\Acal^{\Q_{\ge0}}$, an observer  (obtained in the
	above third step, e.g., in Fig.~\ref{fig18_det_MPautomata}) has at most $2^{2^M|Q|}$ states and at most $
	2^{2^M|Q|} (2^M+|\Sig|)$ transitions, where $|Q|$ denotes the number of states of $\Acal^{\Q_{\ge0}}$;
	$M=\log_2{(\max_{i\in\llb 1,l \rrb}
	\{Nm_i/n_i\})}$; $m_1/n_1,\dots,m_l/n_l$ are enumerations of all positive weights of $\Acal^{\Q_{\ge0}}$;
	$m_i$ and $n_i$, $i\in\llb 1,l \rrb$, are relatively positive prime integers; $N$ is the least common multiple of
	$n_1,\dots,n_l$.
	\begin{figure}[!htbp]
        \centering
	\begin{tikzpicture}
	[>=stealth',shorten >=1pt,thick,auto,node distance=2.0 cm, scale = 0.84, transform shape,
	->,>=stealth,inner sep=2pt, initial text = 0]

	\tikzstyle{emptynode}=[inner sep=0,outer sep=0]

	\node[initial, state, initial where = left] (q01) {$q_{0,1}$};
	\node[state] (q02) [right of = q01] {$q_{0,2}$};
	\node[state] (q03) [right of = q02] {$q_{0,3}$};
	\node[emptynode] (empty1) [right of = q03] {$\cdots$};
	\node[state] (q010) [right of = empty1] {$q_{0,10}$};
	\node[state] (q011) [right of = q010] {$q_{0,11}$};
	\node[state] (q11) [below of = q01] {$q_{1,1}$};
	\node[state] (q12) [right of = q11] {$q_{1,2}$};
	\node[state] (q21) [below of = q11] {$q_{2,1}$};
	\node[state] (q22) [right of = q21] {$q_{2,2}$};
	\node[state] (q31) [below of = q21] {$q_{3,1}$};
	\node[state] (q32) [right of = q31] {$q_{3,2}$};
	\node[state] (q41) [below of = q31] {$q_{4,1}$};
	\node[state] (q42) [right of = q41] {$q_{4,2}$};

	\path [->]
	(q01) edge node [above, sloped] {$1$} (q02)
	(q02) edge node [above, sloped] {$1$} (q03)
	(q03) edge node [above, sloped] {$1$} (empty1)
	(empty1) edge node [above, sloped] {$1$} (q010)
	(q010) edge node [above, sloped] {$1$} (q011)
	(q01) edge node [above, sloped] {$1$} (q02)
	(q11) edge node [above, sloped] {$1$} (q12)
	(q21) edge node [above, sloped] {$1$} (q22)
	(q31) edge node [above, sloped] {$1$} (q32)
	(q41) edge node [above, sloped] {$1$} (q42)
	(q22) edge [bend right] node [above, sloped] {$0$} (q21)
	(q32) edge [bend right] node [above, sloped] {$a$} (q31)
	(q42) edge [bend right] node [above, sloped] {$a$} (q41)
	;

	\node at (5,-0.8) {$0$};

	\draw 
	(q011) .. controls (10.0,-1.0) and (0,-1) .. (q11) 
	;

	\node at (-0.5,-1.2) {$0$};

	\draw 
	(q02) .. controls (-1.0,-1.5) .. (q21) 
	;

	\node at (2.7,-4.3) {$a$};

	\draw 
	(q12) .. controls (4.0,-5.3) and (0.0,-5.0) .. (q31) 
	;

	\node at (-0.5,-5.2) {$a$};

	\draw 
	(q22) .. controls (-1.0,-5.5) .. (q41) 
	;

    \end{tikzpicture}
	\caption{The first step of computing the observer (defined in \cite{Li2021ObserverSpecialTimedAutomata})
	of automaton $\Acal_0^{\N}$ in Fig.~\ref{fig10_det_MPautomata}.}
	\label{fig16_det_MPautomata} 
	\end{figure}
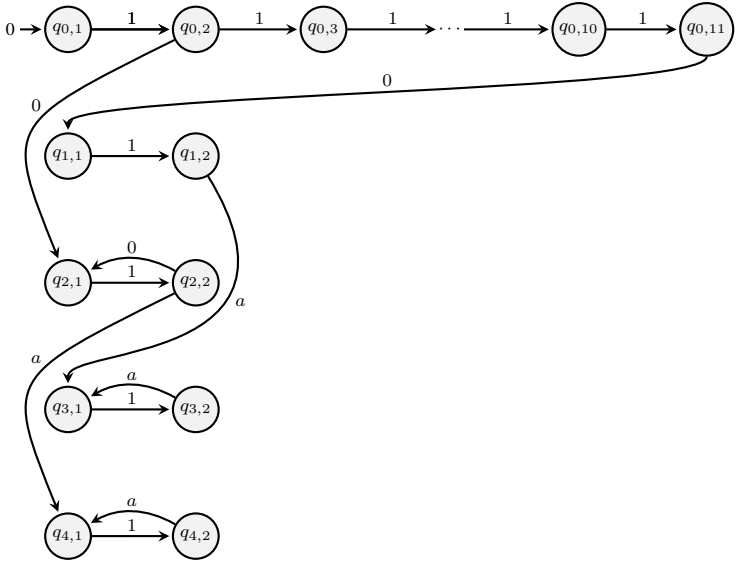
	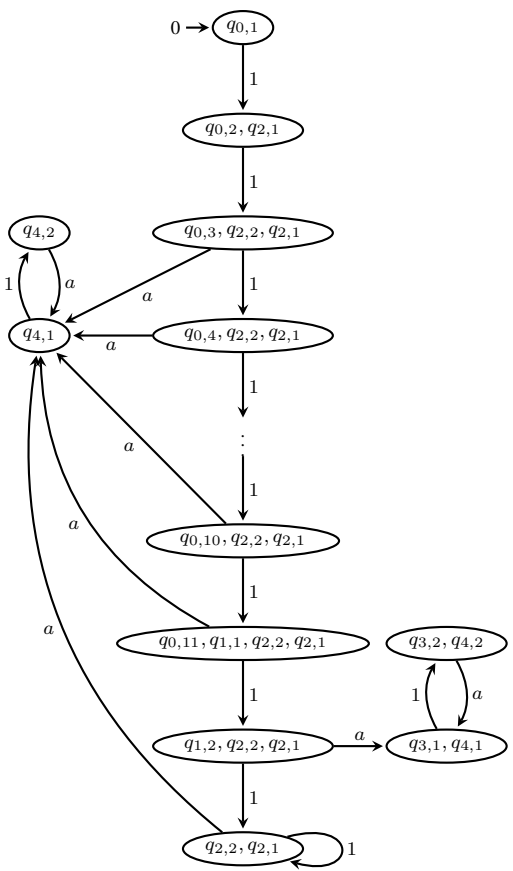
\begin{figure}[!htbp]
        \centering
	\rotatebox{90}{
	\begin{tikzpicture}
	[>=stealth',shorten >=1pt,thick,auto,node distance=1.5 cm, scale = 0.9, transform shape,
	->,>=stealth,inner sep=2pt, initial text = 0]

	\tikzstyle{emptynode}=[inner sep=0,outer sep=0]

	\node[initial, elliptic state, initial where = left] (q01) {$q_{0,1}$};
	\node[elliptic state] (q02q21) [below of = q01] {$q_{0,2},q_{2,1}$};
	\node[elliptic state] (q03q22q21) [below of = q02q21] {$q_{0,3},q_{2,2},q_{2,1}$};
	\node[elliptic state] (q04q22q21) [below of = q03q22q21] {$q_{0,4},q_{2,2},q_{2,1}$};
	\node[emptynode] (empty1) [below of = q04q22q21] {$\vdots$};
	\node[elliptic state] (q010q22q21) [below of = empty1] {$q_{0,10},q_{2,2},q_{2,1}$};
	\node[elliptic state] (q011q11q22q21) [below of = q010q22q21] {$q_{0,11},q_{1,1},q_{2,2},q_{2,1}$};
	\node[elliptic state] (q12q22q21) [below of = q011q11q22q21] {$q_{1,2},q_{2,2},q_{2,1}$};
	\node[elliptic state] (q22q21) [below of = q12q22q21] {$q_{2,2},q_{2,1}$};

	\node[emptynode] (empty2) [right of = q12q22q21] {};
	\node[elliptic state] (q31q41) [right of = empty2] {$q_{3,1},q_{4,1}$};

	\node[emptynode] (empty3) [right of = q31q41] {};
	\node[elliptic state] (q32q42) [above of = q31q41] {$q_{3,2},q_{4,2}$};

	\node[emptynode] (empty4) [left of = q04q22q21] {};
	\node[elliptic state] (q41) [left of = empty4] {$q_{4,1}$};

	\node[elliptic state] (q42) [above of = q41] {$q_{4,2}$};
	
	\path [->]
	(q01) edge node {$1$} (q02q21)
	(q02q21) edge node {$1$} (q03q22q21)
	(q03q22q21) edge node {$1$} (q04q22q21)
	(q04q22q21) edge node {$1$} (empty1)
	(empty1) edge node {$1$} (q010q22q21)
	(q010q22q21) edge node {$1$} (q011q11q22q21)
	(q011q11q22q21) edge node {$1$} (q12q22q21)
	(q12q22q21) edge node {$1$} (q22q21)
	(q22q21) edge [loop right] node {$1$} (q22q21)

	(q12q22q21) edge node {$a$} (q31q41)

	(q31q41) edge [bend left] node {$1$} (q32q42)
	(q32q42) edge [bend left] node {$a$} (q31q41)

	(q03q22q21) edge node {$a$} (q41)
	(q04q22q21) edge node {$a$} (q41)
	(q010q22q21) edge node {$a$} (q41)
	(q011q11q22q21) edge [bend left] node {$a$} (q41)
	(q22q21) edge [bend left] node {$a$} (q41)

	(q41) edge [bend left] node {$1$} (q42)
	(q42) edge [bend left] node {$a$} (q41)
	;

    \end{tikzpicture}
	}
	\caption{The second step of computing the observer (defined in \cite{Li2021ObserverSpecialTimedAutomata})
	of automaton $\Acal_0^{\N}$ in Fig.~\ref{fig10_det_MPautomata}.}
	\label{fig17_det_MPautomata} 
	\end{figure}
	\begin{figure}[!htbp]
        \centering
	\rotatebox{90}{
	\begin{tikzpicture}
	[>=stealth',shorten >=1pt,thick,auto,node distance=1.5 cm, scale = 0.9, transform shape,
	->,>=stealth,inner sep=2pt, initial text = 0]

	\tikzstyle{emptynode}=[inner sep=0,outer sep=0]

	\node[initial, elliptic state, initial where = left] (q01) {$q_{0}$};
	\node[elliptic state] (q02q21) [below of = q01] {$q_{0},q_{2}$};
	\node[elliptic state] (q03q22q21) [below of = q02q21] {$q_{0},q_{2}$};
	\node[elliptic state] (q04q22q21) [below of = q03q22q21] {$q_{0},q_{2}$};
	\node[emptynode] (empty1) [below of = q04q22q21] {$\vdots$};
	\node[elliptic state] (q010q22q21) [below of = empty1] {$q_{0},q_{2}$};
	\node[elliptic state] (q011q11q22q21) [below of = q010q22q21] {$q_{1},q_{2}$};
	\node[elliptic state] (q12q22q21) [below of = q011q11q22q21] {$q_{1},q_{2}$};
	\node[elliptic state] (q22q21) [below of = q12q22q21] {$q_{2}$};

	\node[emptynode] (empty2) [right of = q12q22q21] {};
	\node[elliptic state] (q31q41) [right of = empty2] {$q_{3},q_{4}$};

	\node[emptynode] (empty3) [right of = q31q41] {};
	\node[elliptic state] (q32q42) [above of = q31q41] {$\emptyset$};

	\node[emptynode] (empty4) [left of = q04q22q21] {};
	\node[elliptic state] (q41) [left of = empty4] {$q_{4}$};

	\node[elliptic state] (q42) [above of = q41] {$\emptyset$};
	
	\path [->]
	(q01) edge node {$1$} (q02q21)
	(q02q21) edge node {$1$} (q03q22q21)
	(q03q22q21) edge node {$1$} (q04q22q21)
	(q04q22q21) edge node {$1$} (empty1)
	(empty1) edge node {$1$} (q010q22q21)
	(q010q22q21) edge node {$1$} (q011q11q22q21)
	(q011q11q22q21) edge node {$1$} (q12q22q21)
	(q12q22q21) edge node {$1$} (q22q21)
	(q22q21) edge [loop right] node {$1$} (q22q21)

	(q12q22q21) edge node {$a$} (q31q41)

	(q31q41) edge [bend left] node {$1$} (q32q42)
	(q32q42) edge [bend left] node {$a$} (q31q41)

	(q03q22q21) edge node {$a$} (q41)
	(q04q22q21) edge node {$a$} (q41)
	(q010q22q21) edge node {$a$} (q41)
	(q011q11q22q21) edge [bend left] node {$a$} (q41)
	(q22q21) edge [bend left] node {$a$} (q41)

	(q41) edge [bend left] node {$1$} (q42)
	(q42) edge [bend left] node {$a$} (q41)
	;

    \end{tikzpicture}
	}
	\caption{The third step of computing the observer (defined in \cite{Li2021ObserverSpecialTimedAutomata})
	of automaton $\Acal_0^{\N}$ in Fig.~\ref{fig10_det_MPautomata}. In this step, the observer was obtained.}
	\label{fig18_det_MPautomata} 
	\end{figure}
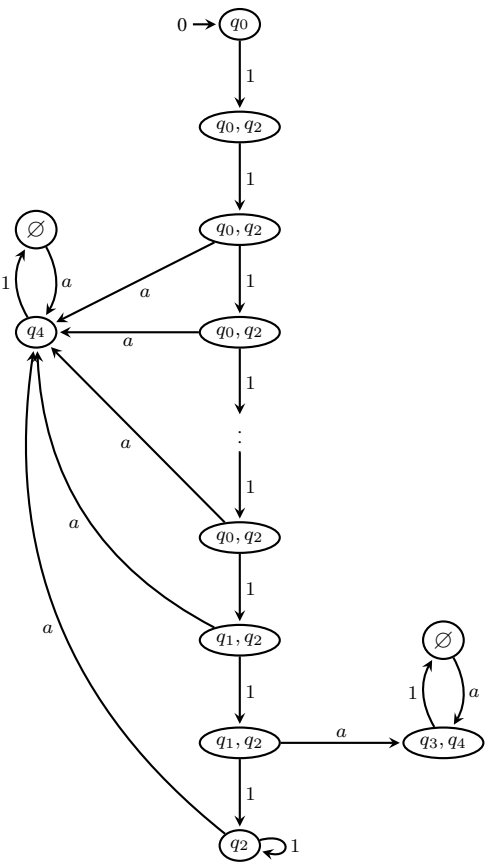

	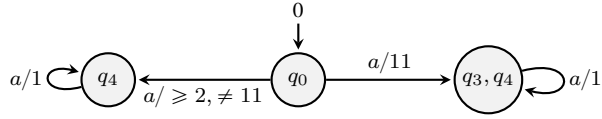
\begin{figure}[!htbp]
        \centering
	\begin{tikzpicture}
	[>=stealth',shorten >=1pt,thick,auto,node distance=2.5 cm, scale = 1.0, transform shape,
	->,>=stealth,inner sep=2pt, initial text = 0]

	\tikzstyle{emptynode}=[inner sep=0,outer sep=0]

	\node[initial, state, initial where = above] (q0) {$q_{0}$};
	\node[state] (q4) [left of = q0] {$q_{4}$};
	\node[state] (q3q4) [right of = q0] {$q_3,q_{4}$};
	
	\path [->]
	(q0) edge node {$a/\ge 2,\ne 11$} (q4)
	(q0) edge node {$a/11$} (q3q4)
	(q4) edge [loop left] node {$a/1$} (q4)
	(q3q4) edge [loop right] node {$a/1$} (q3q4)
	;

    \end{tikzpicture}
	\caption{The fourth step of computing the observer (defined in \cite{Li2021ObserverSpecialTimedAutomata})
	of automaton $\Acal_0^{\N}$ in Fig.~\ref{fig10_det_MPautomata}. In this step, the observer was simplified.}
	\label{fig19_det_MPautomata} 
	\end{figure}
\end{remark}

\begin{remark}\label{rem2_det_MPautomata}
	We now compare the \emph{current-state estimate}~\eqref{CSE_det_MPautomata} with the \emph{current-state estimate}
	used in \cite{Shu2007Detectability_DES,Shu2011GDetectabilityDES,Zhang2017PSPACEHardnessWeakDetectabilityDES,Masopust2018ComplexityDetectabilityDES}
	and \emph{the set of $\gamma$-consistent states} of labeled weighted automata over \emph{semirings}
	used in \cite{Lai2021DetUnambiguousWAutomata,Lai2019StateEstimationMPA}, where $\gamma$ is a weighted  label sequence.

	As mentioned before, a labeled finite-state automaton studied in 
	\cite{Shu2007Detectability_DES,Shu2011GDetectabilityDES,Masopust2018ComplexityDetectabilityDES,Zhang2017PSPACEHardnessWeakDetectabilityDES} 
	can be regarded as an automaton $\Acal^{\N}$ such that all 
	unobservable transitions
	are instantaneous and every two observable transitions with the same label have
	the same weight in $\N$. In this regard, \eqref{CSE_det_MPautomata} reduces to the current-state 
	estimate in \cite{Shu2007Detectability_DES,Shu2011GDetectabilityDES,Masopust2018ComplexityDetectabilityDES,Zhang2017PSPACEHardnessWeakDetectabilityDES} in form. However, because labeled finite-state automata
	are untimed models, in \cite{Shu2007Detectability_DES,Shu2011GDetectabilityDES,Masopust2018ComplexityDetectabilityDES,Zhang2017PSPACEHardnessWeakDetectabilityDES}, it is not specified how much time is needed for the execution
	of an observable transition.

	The set of $\gamma$-consistent states is the counterpart of the current-state estimate in 
	labeled weighted automata over semirings. Consider a labeled weighted automaton $\Acal^{\Q}=(\Q,Q,E,Q_0,\Dt,
	\alpha,\mu,\Sig,\ell)$ over monoid $(\Q,+,0)$ (as in \eqref{LWA_monoid_det_MPautomata}) and a labeled
	weighted automaton $\Acal^{\underQ}$ over 
	semiring $\underQ=(\Q\cup\{-\infty\},\max,+,-\infty,0)$ such that their only difference lies in that
	in $\Acal^{\Q}$ the weights are chosen from monoid $(\Q,+)$, but the weights in $\Acal^{\underQ}$ are chosen
	from semiring $\underQ$. Consider a path 
	$$\pi:=q_0\xrightarrow[]{e_1}q_1\xrightarrow[]{e_2}\cdots\xrightarrow[]{e_n}q_n$$
	as in \eqref{path_det_MPautomaton}, its weighted word (as in \eqref{timedword_det_MPautomaton}) is 
	$$\tau(\pi):=(e_1,t_1)(e_2,t_2)\dots(e_n,t_n),$$
	where for every $i\in\llb 1,n\rrb$, $t_i=\sum_{j=1}^{i}\mu(e_j)_{q_{j-1}q_j}$.
	Recall that the weight of $\pi$ is defined by $t_n$.

	In \cite{Lai2021DetUnambiguousWAutomata,Lai2019StateEstimationMPA},
	the weighted sequence of $\pi$ is defined by 
	$$\s(\pi)=(e_1,t_1')(e_2,t_2')\dots(e_n,t_n'),$$
	where for every $i\in\llb 1,n\rrb$, $t_i'$ is the maximum among the weights of all paths from $Q_0$ to $q_i$ under
	event sequence $e_1\dots e_i$. The label sequence $\ell(\s(\pi))$ of a weighted sequence $\s(\pi)$ is defined
	in the same way as the label sequence of a weighted word, i.e., $\ell$ erases $(e_i,t_i')$ if $e_i$ is unobservable,
	and maps $(e_i,t_i')$ to $(\ell(e_i),t_i)$ otherwise. Given $\gamma\in(\Sig\times\Q)^*$, the set of
	$\gamma$-consistent
	states is defined by 
	\begin{equation}
	\begin{split}
		C(\gamma)=\{q\in Q| &(\exists\text{ a path }\pi=q_0\xrightarrow[]{s}q)\\
		&[(q_0\in Q_0)\wedge(s\in E^*)\wedge (\ell(\s(\pi))=\gamma)]\}.
	\end{split}
	\end{equation}

	There are two differences between the set $C(\gamma)$ of $\gamma$-consistent states and current-state
	estimate $\Mt(\Acal^{\Q},\gamma)$:
	(A) in the former two operations ``$\max$'' and ``$+$'' are considered, in the latter only ``$+$'' is considered,
	(B) in the former after the last observable event in $s$, all unobservable paths are considered,
	in the latter after the last observable event in $s$, only unobservable, instantaneous paths are considered.
	The first difference is major, but the second is minor and neglectable. The first difference shows that
	$\Acal^{\underline{Q_{\ge0}}}$ can represent a max-plus timed system, while $\Acal^{\Q_{\ge0}}$
	can represent a real-time system.
	For unambiguous $\Acal^{\Q}$ and
	$\Acal^{\underQ}$, the difference (A) vanishes because in such automata, under every event sequence,
	there exists at most one path from $Q_0$ to any state. Hence, the unique difference between the definitions 
	of detectability of unambiguous $\Acal^{\underQ}$ in \cite{Lai2021DetUnambiguousWAutomata} and those 
	of unambiguous $\Acal^{\Q}$ in the current paper comes from the difference (B). If one method can be used to 
	verify the detectability in \cite{Lai2021DetUnambiguousWAutomata}, then it can be slightly modified to verify
	the detectability in the current paper, and vice versa. So (B) is minor and neglectable.

	In a self-composition $\CCa(\Acal^{unam,\Q})$ as in Definition~\ref{def_CC_MPautomata},
	for each state $(q_1,q_2)$ of $\CCa(\Acal^{unam,\Q})$, for all states $q_3,q_4$ of $\Acal^{unam,\Q}$
	such that $q_3$ (resp., $q_4$) is reachable from $q_1$ (resp., $q_2$) through some unobservable (not necessarily
	instantaneous) path, add two unobservable transitions $(q_1,q_2) \rightarrow (q_3,q_4)$ and $(q_3,q_4) \rightarrow
	(q_1,q_2)$ into $\CCa(\Acal^{unam,\Q})$. Then after putting the updated self-composition into 
	Theorem~\ref{thm1_det_MPautomata}, we obtain an equivalent condition for the strong detectability of 
	$\Acal^{unam,\underQ}$ studied in \cite{Lai2021DetUnambiguousWAutomata}. On the other hand, for an observer
	$\Acal_{obs}^{unam,\Q}$ as in Definition~\ref{def_observer_MPautomata}, at each state $x\in X$, we add every state
	$q$ that is reachable from some state in $x$ through some (not necessarily instantaneous) unobservable
	path into $x$, then after putting the updated observer $\tilde\Acal_{obs}^{unam,\Q}$
	into Theorem~\ref{thm8_det_MPautomata}, Theorem~\ref{thm2_det_MPautomata}, and 
	Theorem~\ref{thm7_det_MPautomata}, we obtain equivalent conditions for the strong periodic detectability,
	weak detectability, and weak approximate detectability of $\Acal^{unam,\underQ}$ studied in 
	\cite{Lai2021DetUnambiguousWAutomata}. Recall that the above verification does not depend on the divergence-freeness
	assumption, but the verification in \cite{Lai2021DetUnambiguousWAutomata} depends on the assumption.
	
	On the other hand, the deterministic finite automaton $G_{obs}'$ 
	returned by \cite[Algorithm~1]{Lai2021DetUnambiguousWAutomata} can be used to verify the detectability studied
	in the current paper for divergence-free $\Acal^{unam,\Q}$ and $\Acal^{unam,\underQ}$, while the observer
	$G_{obs}$ returned by \cite[Algorithm~1]{Lai2021DetUnambiguousWAutomata} can be used to verify
	the detectability defined in \cite{Lai2021DetUnambiguousWAutomata} for divergence-free $\Acal^{unam,\underQ}$.
	Recall that $G_{obs}$ was obtained from $G_{obs}'$ in the same way as obtaining $\tilde\Acal_{obs}^{unam,\Q}$
	from $\Acal_{obs}^{unam,\Q}$.

	Now we use automaton $\Acal_{1}^{\N}$ (which can also be regarded
	as automaton $\Acal_1^{\underN}$ over semiring $\underN$) in Fig.~\ref{fig6_det_MPautomata} to
	show the essential differences between the two notions. Reconsider the
	paths $\pi_1,\dots,\pi_5$ in \eqref{eqn12_det_MPautomata}. By definition, one has
	\begin{subequations}\label{eqn15_det_MPautomata} 
		\begin{align}
			\sigma(\pi_1)=(a,1)(b,3), & \quad \ell(\sigma(\pi_1))=(\rho,1)(\rho,3),\\
			\sigma(\pi_2)=(a,1)(b,3), & \quad \ell(\sigma(\pi_2))=(\rho,1)(\rho,3),\\
			\sigma(\pi_3)=(a,1)(u,2)(b,4), & \quad \ell(\sigma(\pi_3))=(\rho,1)(\rho,4),\\
			\sigma(\pi_4)=(a,1)(u,2)(b,4), & \quad \ell(\sigma(\pi_4))=(\rho,1)(\rho,4),\\
			\sigma(\pi_5)=(a,1)(b,3)(u,4), & \quad \ell(\sigma(\pi_5))=(\rho,1)(\rho,3).
		\end{align}
	\end{subequations}
	Then 
	\begin{subequations}\label{eqn16_det_MPautomata}
		\begin{align}
			&C((\rho,1)(\rho,2))=\emptyset,\label{eqn16_1_det_MPautomata}\\
			&C((\rho,1)(\rho,3))=\{q_3,q_4\}.\label{eqn16_2_det_MPautomata}
		\end{align}
	\end{subequations}

	Compared with \eqref{eqn17_det_MPautomata} (i.e., $\Mt(\Acal_1^{\N},(\rho,1)(\rho,2))=\Mt(\Acal_1^{\N},
	(\rho,1)(\rho,3))=\{q_3\}$), one can see the difference
	between $\Mt(\Acal_1^{\N},(\rho,1)(\rho,2))$ and $C((\rho,1)(\rho,2))$ is caused by (A), but the 
	difference between $\Mt(\Acal_1^{\N},(\rho,1)(\rho,3))$ and $C((\rho,1)(\rho,3))$ is caused by (B).
	Although there is a transition sequence $q_0\xrightarrow[]{a/1}q_2\xrightarrow[]{b/1}q_3$,
	there exists no weighted sequence $(a,1)(b,2)$, so $C((\rho,1)(\rho,2))=\emptyset$.

	The above difference (A) also induces another remarkable difference between $\Acal^{\Q}$ and $\Acal^{\underQ}$.
	In Remark~\ref{rem6_det_MPautomata}, we show that for any weighted label sequence $\gamma(\s_2,t_2)
	\in(\Sig\times {\Q})^+$, one can compute $\Mt(\Acal^{\Q},\gamma(\s_2,t_2))$ from $\Mt(\Acal^{\Q},
	\gamma)$. However, the results in \cite[Example~7]{Lai2019StateEstimationMPA} show that generally
	$C(\gamma(\s_2,t_2))$ cannot be computed from $C(\gamma)$, but must be computed from the initial time.
	This shows that a notion of observer for general $\Acal^{\underQ}$ might not be computable with complexity
	upper bounds.  
\end{remark}

\begin{remark}\label{rem7_det_MPautomata}
	The method of computing observer $\Acal_{obs}^{unam,\underQ}$ for divergence-free $\Acal^{unam,\underQ}$
in \cite{Lai2021DetUnambiguousWAutomata}
is as follows. Due to the feature of unambiguity, a given $\Acal^{unam,\underQ}$
is firstly\footnote{Informally, this step is to 
aggregate every path $q_0\xrightarrow[]{s_1} q_1\xrightarrow[]{e_2}q_2$, where $s_1$ is a sequence of unobservable
events of length no greater than the number of states and $e_2$ is an observable event,
to a path $q_0\xrightarrow[]{e_2}q_2$ whose weight 
is equal to the weight of $q_0\xrightarrow[]{s_1} q_1\xrightarrow[]{e_2}q_2$. After this step,
the obtained structure may not be a weighted automaton any more, because a path
$q_0\xrightarrow[]{e_2}q_2$ may have two different weights; however, after regarding every pair of 
event $e$ and the weight of a transition under $e$ as a new event, then a labeled finite-state automaton
is obtained.} transformed to a labeled finite-state automaton $\Acal'$ in exponential time, 
the subsequent procedure of computing $\Acal_{obs}^{unam,\underQ}$
is almost the same as the procedure of computing the observer $\Acal'_{obs}$ of 
$\Acal'$ as in \cite{Shu2007Detectability_DES}, hence the size of $\Acal_{obs}^{unam,\underQ}$
is exponential in the size of $\Acal^{unam,\underQ}$, which is the same
as the case that the size of the observer $\Acal_{obs}$ is exponential in that 
of automaton $\Acal$. Note that for $\Acal^{unam,\underQ}$
that contains an unobservable cycle (that is, $\Acal^{unam,\underQ}$ is not divergence-free),
generally the method in \cite{Lai2021DetUnambiguousWAutomata} cannot be used to verify 
detectability of such automata. For example, consider automaton $\Acal_0^{\underN}$ 
over semiring $\underN$ shown in Fig.~\ref{fig10_det_MPautomata}, there is an unobservable self-loop
on state $q_2$. Every number denotes the execution time of the corresponding transition,
e.g., when
$\Acal_0^{\underN}$ is in state $q_0$ and event $u$ occurs, $\Acal_0^{\underN}$ transitions to state $q_1$,
the execution time of this transition is $10$. Hence when we observe $a$ at instant $11$,
we know that $\Acal_0^{\underN}$ can be in states $q_3$ or $q_4$. However, by using the method in 
\cite{Lai2021DetUnambiguousWAutomata}, after the first step as mentioned above, we obtain the finite-state
automaton shown in Fig.~\ref{fig11_det_MPautomata}, by which we know that when we observe $a$ at
instant $11$, $\Acal_0^{\underN}$ can only be in state $q_3$. Hence, the detectability of automata
like $\Acal_0^{\underN}$ cannot be verified by using the method in \cite{Lai2021DetUnambiguousWAutomata}.
On the other hand, it is not pointed out in 
\cite{Lai2021DetUnambiguousWAutomata} that whether an observer $\Acal_{obs}^{unam,\underQ}$
is computable heavily depends on the weights. If the weights are real numbers, then generally
the observer is uncomputable, because there exist
only countably infinitely many computable real numbers \cite{Turing1936ComputableNumbers}. 

\begin{figure}[!htbp]
        \centering
	\begin{tikzpicture}
	[>=stealth',shorten >=1pt,thick,auto,node distance=2.5 cm, scale = 1.0, transform shape,
	->,>=stealth,inner sep=2pt]

	\tikzstyle{emptynode}=[inner sep=0,outer sep=0]

	\node[initial, state, initial where = above] (q0) {$q_0$};
	\node[state] (q3) [left of = q0] {$q_3$};
	\node[state] (q4) [right of = q0] {$q_{4}$};

	\path [->]
	(q0) edge node [above, sloped] {$(a,11)$} (q3)
	(q0) edge node [above, sloped] {$\begin{matrix}(a,2)(a,3)\\(a,4)(a,5)\end{matrix}$} (q4)
	(q3) edge [loop left] node {$(a,1)$} (q3)
	(q4) edge [loop right] node {$(a,1)$} (q4)
	;

     \end{tikzpicture}
	 \caption{Finite-state automaton computed from $\Acal_0^{\underN}$ in Fig~\ref{fig10_det_MPautomata}
	 by using the method in \cite{Lai2021DetUnambiguousWAutomata}.}
	\label{fig11_det_MPautomata}
\end{figure}
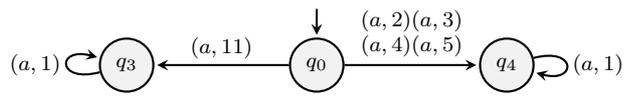
\end{remark}

\end{document}